\documentclass[12pt,reqno]{amsart}
\usepackage{amsmath}
\usepackage{amsxtra}
\usepackage{amscd}
\usepackage{amsthm}
\usepackage{amsfonts}
\usepackage{amssymb}
\usepackage{mathtools}
\usepackage{eucal}
\usepackage{cases}
\usepackage{paralist} 
\usepackage{array}
\usepackage{etex} 
\usepackage{color}
\usepackage{float}
\usepackage{cite}
\usepackage{tensor}

\usepackage{graphicx,color}

\definecolor{darkblue}{rgb}{0,0,.5}
\definecolor{darkred}{rgb}{.5,0,0}
\definecolor{darkgreen}{rgb}{0,0.5,0}

\usepackage[hang,flushmargin]{footmisc} 
\usepackage[arrow, matrix, curve]{xy}
\usepackage{perpage}
\MakePerPage{footnote}

\usepackage{hyperref}
\hypersetup {
	pdfstartview=FitH,
	colorlinks,
	citecolor=darkgreen,
	linkcolor=darkred,
	urlcolor=darkblue
}

\allowdisplaybreaks

\def\smallskip{\vskip\smallskipamount}
\def\medskip{\vskip\medskipamount}
\def\bigskip{\vskip\bigskipamount}

\def\({\left(}
\def\){\right)}
\def\[{\left[}
\def\]{\right]}
\newcommand{\bra}[1]{\langle #1 |}        
\newcommand{\ket}[1]{{| #1 \rangle}}      
\newcommand{\br}[1]{{\langle #1 \rangle}}  

\newcommand{\bea}{\begin{eqnarray}}
	\newcommand{\ena}{\end{eqnarray}}
\def\bel{\begin{eqnarray}}
	\def\enl{\end{eqnarray}}

\def\ba{\begin{eqnarray}}
	\def\ea{\end{eqnarray}}

\def\be{\begin{equation}}
	\def\ee{\end{equation}}

\newcommand{\sltr}{\mathfrak{sl}_3}

\newcommand{\sln}{\mathfrak{sl}_n}

\newcommand{\id}{{\rm id}}

\newcommand{\tr}{{\rm tr}}

\newcommand{\End}{\mathop{\rm End}}

\newenvironment{tenumerate}{
	\begin{enumerate}
		
	}{\end{enumerate}}
\newcommand{\bi}{\begin{tenumerate}}
	\newcommand{\ei}{\end{tenumerate}}
\newcommand{\isoto}[1][]%
{{\mathop{\buildrel{\sim}\over\longrightarrow}\limits_{#1}}}

\def\2{\frac{1}{2}} \def\4{\frac{1}{4}}

\def\6{\partial}

\def\+{\dagger}

\def\<{\langle} \def\>{\rangle}

\numberwithin{equation}{section}
\newtheorem{thm}{Theorem}[section]
\newtheorem{prop}[thm]{Proposition}
\newtheorem{ppt}[thm]{Properties}
\newtheorem{conj}[thm]{Conjecture}

\newtheorem{rem}[thm]{Remark}
\newtheorem{cor}[thm]{Corollary}
\newtheorem{definition}[thm]{Definition}

\newcommand{\recbinom}{\genfrac{[}{]}{0pt}{}}
\newcommand{\rightarrowdbl}{\rightarrow\mathrel{\mkern-14mu}\rightarrow}
\newcommand{\xrightarrowdbl}[2][]{\xrightarrow[#1]{#2}\mathrel{\mkern-14mu}\rightarrow}

\newcounter{saveenumi}
\newcommand{\seti}{\setcounter{saveenumi}{\value{enumi}}}
\newcommand{\conti}{\setcounter{enumi}{\value{saveenumi}}}

\begin{document}
	\begin{title}[]
		{On the properties of the density matrix of the $\boldsymbol{\mathfrak{sl}_{n+1}}$-invariant model}
	\end{title}
	\date{\today}
	\author{ H.~Jürgens, H.~Boos}
	\address{Math Department, University of Modena and Reggio Emilia, It-41125,
		Modena, Italy}\email{hjuergens@unimore.it}
	\address{Physics Department, University of Wuppertal, D-42097,
		Wuppertal, Germany}\email{hboos@uni-wuppertal.de}

	\begin{abstract}
		We present an ansatz of generalizing the construction of recursion relations for the correlation functions of the $\mathfrak{sl}_2$-invariant fundamental exchange model in the thermodynamic limit by Jimbo, Miwa, Smirnov, Takeyama and one of our present authors in 2004 for higher rank. Due to the structure of the correlators as functions of their inhomogeneity parameters, a recursion formula for the reduced density matrix was proven. In the case of $\sltr$, we use the explicit results of Klümper and Ribeiro, and Nirov, Hutsalyuk and one of our present authors for the reduced density matrix of up to operator length three to verify whether it is possible to relate the residues of the density matrix of length $n$ to the density matrix of length smaller than $n$ as in $\mathfrak{sl}_2$. This is unclear, since the reduced quantum Knizhnik--Zamolodchikov equation splits into two parts for higher rank. In fact, we show two relations, one of which is a straightforward generalisation to the $\mathfrak{sl}_2$ case and one which is completely new. This allows us to construct an analogue of the operator $X_k$ which we call Snail Operator. In the $\mathfrak{sl}_2$-case, this operator has many nice properties including in particular the fact that only one irreducible representation of the Yangian $Y(\mathfrak{sl}_2)$, the Kirillov--Reshetikhin module $W_k$, contributed the residue at $\lambda_i-\lambda_j=-(k+1)$. Here, we give an overview of the mathematical background, T-systems, and show a new application of the extended T-systems introduced by Mukhin and Young in 2012 regarding the Snail Operator.
	\end{abstract}
	
	\maketitle
	
	\tableofcontents
	\section{ Introduction}
	\subsection{The $\mathfrak{sl}_2$ case \cite{BJMST}}
	Let's consider the Hamiltonian $H_{XXX} = \frac{1}{2}\sum_j\left(\sigma^x_j\sigma^x_{j+1}+\sigma^y_j\sigma^y_{j+1}+\sigma^z_j\sigma^z_{j+1}\right)$ of the rational $\mathfrak{sl}_2$-invariant model which corresponds to the gapless case $q\to 1$, $\Delta = \frac{q+q^{-1}}{2}$ of the Heisenberg XXZ spin chain. One can now consider a more general inhomogeneous chain generated by the transfer matrix \begin{align*}
		\operatorname{tr}\left(R_{a,-L+1}(\lambda)\cdots R_{a,0}(\lambda)R_{a,1}(\lambda-\lambda_1)\cdots R_{a,n}(\lambda-\lambda_n)R_{a,n+1}(\lambda)\cdots R_{a,L}(\lambda)\right),
	\end{align*}
	which is still exactly solvable.
	The reduced density matrix of this model in the thermodynamic limit is obtained in terms of multiple integrals from the vertex operator approach for the gapped regime by analytic continuation. Factorisation of these multiple integrals into products of single integrals was proven in papers \cite{BK}\cite{BKS}. However, the construction introduced in 2004 in the paper \cite{BJMST} provides another way of describing the correlation functions (i.e. the reduced density matrix). In particular the conjecture
	\begin{align}
		\left[D_{1,\dots,n}(\lambda_1,\dots,\lambda_n)\right]_{\epsilon_1\dots\epsilon_n} ^{\bar{\epsilon}_1\dots\bar{\epsilon}_n}:= \bra{0}(\tensor{E}{_{\epsilon_1}^{\bar{\epsilon}_1}})_1\cdots (\tensor{E}{_{\epsilon_n}^{\bar{\epsilon}_n}})_n\ket{0}= \sum \prod \omega(\lambda_i-\lambda_j) f(\lambda_1,\dots,\lambda_n)
		\label{eqn:factorisation_prop}
	\end{align}
	where $\omega(\lambda)$ is a single transcendental function and the functions $f(\lambda_1,\dots,\lambda_n)$ are rational was proven. This is due to a recursion relation where the correlation function is presented in terms of a transfer matrix over an auxiliary space of 'fractional dimension'. It is an analytic continuation of an operator $X_k$ with respect to $k$ \footnote{In our notation we also call this operator the 'Snail Operator'.}. For the homogeneous limit, odd integer values of the $\zeta$-function appear as the coefficients in the Taylor series of $\omega$ (cf. \cite{BJMST}). This result was generalized for the XXZ spin chain and further led to the fermionic structure of the correlation functions described in the series of papers "Hidden Grassmann structure in the XXZ-model" (\cite{BJMST2}, \cite{BJMST_HGS}, \cite{BJMST_HGS2}, etc.).
	\subsection{The $\mathfrak{sl}_3$ case}
	For the rational $\mathfrak{sl}_3$-invariant model the vertex operator approach still works by taking the limit $q\to 1$ from the massive regime. Though, it is not known whether a factorisation property similar to (\ref{eqn:factorisation_prop}) does exist. As for now, not much is known about the correlation functions. Nevertheless, the rqKZ-equation could be generalized for the $\mathfrak{sl}_3$ case and an attempt to explicitly solve it led to the first explicit results for the short range correlation functions of up to operator length three \cite{BHN},\cite{KR}. Since the construction for $\mathfrak{sl}_2$ is mainly based on the rqKZ-equation and some fusion relations, there is hope for it to have a $\mathfrak{sl}_3$ generalization as well. Especially the operator $X_k$ in \cite{BJMST} seems to have a promising generalization. This will be explained in more detail later. Anyway, we still don't know whether all the other properties of the operator $X_k$ can be generalized and if there is a way to prove a recursion relation for the correlation functions just like in the case of $\mathfrak{sl}_2$.
	\subsection{Structure}
	Let us now give an overview of the structure of the paper.
	
	In section \ref{sect:gen_def} we introduce our notation and some important definitions which we will use throughout the paper. We explain how our graphical notation is understood and introduce our main object of interest, the reduced density matrix $D_m$.
	
	In section \ref{sect:sl2snailconstr} we review the construction of the paper \cite{BJMST} in our notation. We use our graphical notation to visualize the construction as well as the operator $X_k$ \footnote{which is the Snail Operator $\tilde{X}_k$ in our notation}. In section \ref{subsect:T-systems} we explain how T-systems play the main role in proving that only the Kirillov--Reshetikhin module $W_k$ contributes to the operator $X_k$. We clarify how this is the main reason that we can write the operator $X_k$ as a transfer matrix over an auxiliary space of 'fractional dimension' $\lambda\in\mathbb{C}$.
	
	In section \ref{sect:sl3snailconstr} we present our Ansatz for the generalization of this construction. As a first step, we introduce new reduction relations for the residues as they constitute the starting point of it. We present our new results for $\mathfrak{sl}_{n+1}$ graphically and discuss how the generalization for higher rank works. We also introduce our definition $\tilde{X}_k$ of the Snail Operator for rank $n\in\mathbb{N}$, thereby generalising the $\mathfrak{sl}_2$ (rank $1$) result in a straightforward way.
	
	In section \ref{subsect:extT-systems} we review the extended T-systems introduced in the paper \cite{MY} and show how they can be used to analyse the representations that contribute to $\tilde{X}_k$. As we shall see, the result for $\mathfrak{sl}_2$ naturally generalizes as certain minimal snake modules take the role of the Kirillov--Reshetikhin modules for rank $n>1$. Finally, we compare our findings with the $\mathfrak{sl}_2$ case and point out the next steps that shall follow.
	\section{General definitions}
	\label{sect:gen_def}
	\subsection{Cartan data}
	Let $\mathfrak{g}$ be a finite-dimensional simple Lie algebra of rank $n$ over $\mathbb{C}$ and let $\mathfrak{h}$ be a Cartan subalgebra of $\mathfrak{g}$. We normalize the invariant inner product $\br{\cdot,\cdot}$ on $\mathfrak{g}$ such that the square length of the maximal root equals $2$. Furthermore, we identify $\mathfrak{h}$ and $\mathfrak{h}^*$ through $h\mapsto\br{\cdot,h}$. Let $I=\{1,\dots,n\}$ and let $\{\alpha_i\}_{i\in I}$ be the set of simple roots with the corresponding simple coroots $\{\alpha_i^\vee\}_{i\in I}$ and fundamental weights $\{\omega_i\}_{i\in I}$. Let $A:=(a_{ij})$ denote the Cartan matrix and let $d_i$, $i=1,\dots,n$, be the relatively prime integers such that $B=(b_{ij})=(d_i a_{ij})$ is symmetric. We have
	\begin{align}
		2\br{\alpha_i,\alpha_j} = a_{ij}\br{\alpha_i,\alpha_i}, \quad 2\br{\alpha_i,\omega_j} = \delta_{i,j}\br{\alpha_i,\alpha_i}, \quad b_{ij} = d^\vee \br{\alpha_i,\alpha_j},
	\end{align}
	where $d^\vee$ is the maximal number of edges connecting two nodes in the Dynkin diagram of $\mathfrak{g}$. At last we denote the (positive) weight and root lattice by $P$ ($P^+$) and $Q$ ($Q^+$), respectively. We have that $P$ ($P^+$) and $Q$ ($Q^+$) are the $\mathbb{Z}$-span ($\mathbb{Z}_{\geq0}$-span) of the fundamental weights and simple roots, respectively. Then we have a partial order $\leq$ on $P$ in which $\lambda\leq\lambda'$ iff $\lambda'-\lambda\in Q^+$.
	\subsection{The Yangian $Y(\mathfrak{g})$ and the quantum affine algebra $U_q(\tilde{\mathfrak{g}})$}
	Since we will work with finite dimensional representations of the Yangian $Y(\mathfrak{g})$ for the XXX spin chain and the rational six vertex model, we shall recap its second Drinfeld realization. As we will also consider finite-dimensional (type 1) representations of the quantum affine algebra $U_q(\tilde{\mathfrak{g}})$ later, it is helpful to write down its second Drinfeld realization as well. We refer to the book \cite{CPBook} and the paper \cite{CP1991}, where almost all the information can be found.
	\begin{definition}[the second Drinfeld realization of the Yangian]
		The second Drinfeld realization of the Yangian $Y(\mathfrak{g})$ is the associative algebra with generators $X^\pm_{i,r}$, $H_{i,r}$, $i=1,\dots,n$, $r\in\mathbb{N}$ and defining relations
		\begin{align*}
			[H_{i,r},H_{j,s}]&=0,\\
			[H_{i,0},X^\pm_{j,s}]&=\pm d_i a_{ij}X^\pm_{j,s},\\
			[H_{i,r+1},X^\pm_{j,s}]-[H_{i,r},X^\pm_{j,s+1}]&=\pm \frac{1}{2}d_i a_{ij}(H_{i,r}X^\pm_{j,s}+X^\pm_{j,s}H_{i,r}),\\
			[X^+_{i,r},X^-_{j,s}]&=\delta_{i,j}H_{i,r+s},\\
			[X^\pm_{i,r+1},X^\pm_{j,s}]-[X^\pm_{i,r},X^\pm_{j,s+1}]&=\pm \frac{1}{2}d_i a_{ij}(X^\pm_{i,r}X^\pm_{j,s}+X^\pm_{j,s}X^\pm_{i,r}),\\
			\sum_{\sigma\in S_m}[X^\pm_{i,r_{\sigma(1)}},[X^\pm_{i,r_{\sigma(2)}}&,\dots,[X^\pm_{i,r_{\sigma(m)}},X^\pm_{j,s}]\cdots]]=0,
		\end{align*}
		for all sequences of non-negative integers $r_1,\dots,r_m$, where $m=1-a_{ij}$ and $S_m$ is the symmetric group of degree $m$. $\odot$
	\end{definition}
	Let us also recap that this Hopf algebra has a one-parameter group of automorphisms $\tau_a,\, a\in\mathbb{C}$. It is one of the main reasons for its importance as it allows us to define a one parameter family of modules to any given module by adjoining different parameters $a\in\mathbb{C}$ by $\tau_a$, i.e. pulling back by $\tau_a$. The proof of the following proposition is given in \cite{CPBook} (proposition 12.1.5). 
	\begin{prop}[the Hopf algebra automorphism $\tau_a$]
		\label{prop:tau_a_yangian}
		There is a one-parameter group of Hopf algebra automorphisms $\tau_a$ of $Y(\mathfrak{g})$, $a\in\mathbb{C}$, given by
		\begin{align}
			\tau_a(H_{i,r})=\sum_{s=0}^{r}\binom{r}{s}a^{r-s}H_{i,s},\quad \tau_a(X_{i,r}^\pm)=\sum_{s=0}^{r}\binom{r}{s}a^{r-s}X_{i,s}^\pm.
			\label{eqn:tau_a_yangian}
		\end{align}
		$\odot$
	\end{prop}
	Thus, if $V$ is any $Y(\mathfrak{g})$ module, it is convenient to write $V(\lambda)$ for the pullback of $V$ by $\tau_\lambda$.
	\begin{rem}
		\label{rem:ev_a_for_yangian}
		If $\mathfrak{g}$ is of type $A$, one has in addition an evaluation homomorphism $\operatorname{ev}_a:Y(\mathfrak{g})\rightarrowdbl U(\mathfrak{g})$ for all $a\in \mathbb{C}$ such that the composition $U(\mathfrak{g})\lhook\joinrel\xrightarrow{\iota} Y(\mathfrak{g})\xrightarrowdbl{\operatorname{ev}_a}U(\mathfrak{g})$ is the identity map (cf. \cite{CPBook} proposition 12.1.15). A representation defined through the pullback by $\operatorname{ev}_a$ is then called an evaluation representation. If $V$ is any representation of $\mathfrak{g}$, we have $$\tau_b^*\operatorname{ev}_a^*(V)=\tau_0^*\operatorname{ev}_{a+b}^*(V)=\tau_{a+b}^*\operatorname{ev}_0^*(V)\text{ for all }a,b\in\mathbb{C}$$ in this case. We may therefore write $V(\lambda):=\operatorname{ev}_\lambda^*(V)$ such that $V(\lambda)=\tau_\lambda^*(V)$ by identifying $V$ with $\operatorname{ev}_0^*(V)$ as a representation of the Yangian. Finally, we should emphasize that $\operatorname{ev}_a$ is just an algebra homomorphism, not a Hopf algebra homomorphism. $\odot$
	\end{rem}
	
	For the definition of the $q$-deformed case, we may assume that $q\in \mathbb{C}$ is not a root of unity and use the definitions
	\begin{align*}
		[n]_q:=\frac{q^n-q^{-n}}{q-q^{-1}},\quad[n]_q!:=[n]_q[n-1]_q\cdots[1]_q,\quad\recbinom{n}{m}_q:=\frac{[n]_q!}{[n-m]_q![m]_q!}
	\end{align*}
	for the $q$-number, $q$-factorial and $q$-binomial, respectively.
	\begin{definition}[the second Drinfeld realization of $U_q(\tilde{\mathfrak{g}})$]
		\label{def:second_drin_quantum_aff}
		Let $\tilde{\mathfrak{g}}$ be the untwisted affine Lie algebra associated to $\mathfrak{g}$. The second Drinfeld realization of the quantum affine algebra $U_q(\tilde{\mathfrak{g}})$ is the associative algebra with generators $\mathcal{C}^{\pm1/2}$, $\mathcal{K}_i^{\pm1}$, $\mathcal{H}_{i,r}$, $\mathcal{X}^\pm_{i,s}$, $i=1,\dots,n$, $r\in\mathbb{Z}\backslash\{0\}$, $s\in\mathbb{Z}$ and defining relations
		\begin{align*}
			\mathcal{K}_i\mathcal{K}_i^{-1} =  \mathcal{K}_i^{-1}&\mathcal{K}_i = 1, \mathcal{C}^{1/2}\mathcal{C}^{-1/2}=1,\\
			\mathcal{C}^{\pm1/2}&\text{ are}\text{ central},\\
			[\mathcal{K}_i,\mathcal{K}_j]&=[\mathcal{K}_i,\mathcal{H}_{j,r}]=0,\\
			[\mathcal{H}_{i,r},\mathcal{H}_{j,s}]&=\delta_{r,-s}\frac{1}{r}[ra_{ij}]_{q_i}\frac{\mathcal{C}^r-\mathcal{C}^{-r}}{q_j-q_j^{-1}},\\
			\mathcal{K}_i\mathcal{X}^\pm_{j,r}\mathcal{K}_i^{-1}&=q_i^{\pm a_{ij}}\mathcal{X}_{j,r}^\pm,\\
			[\mathcal{H}_{i,r},\mathcal{X}^\pm_{j,s}]&=\pm \frac{1}{r}[ra_{ij}]_{q_i}\mathcal{C}^{\mp|r|/2}\mathcal{X}^\pm_{j,r+s},\\
			\mathcal{X}_{i,r+1}^\pm\mathcal{X}_{j,s}^\pm-q_i^{\pm a_{ij}}\mathcal{X}_{j,s}^\pm\mathcal{X}_{i,r+1}^\pm&=q_i^{\pm a_{ij}}\mathcal{X}_{i,r}^\pm\mathcal{X}_{j,s+1}^\pm-\mathcal{X}_{j,s+1}^\pm\mathcal{X}_{i,r}^\pm\\
			[\mathcal{X}_{i,r}^+,\mathcal{X}_{j,s}^-]=\delta_{i,j}&\frac{\mathcal{C}^{(r-s)/2}\varPhi_{i,r+s}^+-\mathcal{C}^{-(r-s)/2}\varPhi_{i,r+s}^-}{q_i-q_i^{-1}}\\
			\sum_{\sigma\in S_m}\sum_{k=0}^{m}(-1)\recbinom{m}{k}_{q_i}
			\mathcal{X}^\pm_{i,r_{\sigma(1)}}&\cdots\mathcal{X}^\pm_{i,r_{\sigma(k)}}\mathcal{X}_{j,s}^\pm\mathcal{X}_{i,r_{\sigma(k+1)}}^\pm\cdots\mathcal{X}^\pm_{i,r_{\sigma(m)}}=0,\, i\neq j,
		\end{align*}
		for all sequences of non-negative integers $r_1,\dots,r_m$, where $m=1-a_{ij}$, $q_i=q^{d_i}$ and the elements $\varPhi_{i,r}^\pm$ are determined by equating coefficients of powers of $u$ in the formal power series
		\begin{align}
			\Phi^\pm = \sum_{r=0}^{\infty}\varPhi_{i,\pm r}^\pm u^{\pm r}= \mathcal{K}_i^{\pm1}\exp\left(\pm(q_i-q_i^{-1})\sum_{s=1}^{\infty}\mathcal{H}_{i,\pm s}u^{\pm s}\right)
		\end{align}
		$\odot$
	\end{definition}
	\begin{rem}
		If $A$ is the generalized Cartan matrix of $\tilde{\mathfrak{g}}$, then $\tilde{\mathfrak{g}} = L(A)'$, which can be defined as the 1-dimensional central extension of the loop algebra $\mathcal{L}(\mathfrak{g})=\mathbb{C}[t,t^{-1}]\otimes_\mathbb{C}\mathfrak{g}$, i.e. $\tilde{\mathfrak{g}}=\mathcal{L}(\mathfrak{g})\oplus\mathbb{C}c$. We refer to $L(A)$ by $\hat{\mathfrak{g}}$, which is obtained from $\tilde{g}$ by adjoining an element $d$ that acts as a derivation. Thus we have $\hat{\mathfrak{g}}=\tilde{\mathfrak{g}}\oplus\mathbb{C}d=\mathcal{L}(\mathfrak{g})\oplus\mathbb{C}c\oplus\mathbb{C}d$. The difference is that the simple roots in the Cartan subalgebra of $\hat{\mathfrak{g}}$ are linear dependent, whereas $d$ removes the degeneracy for $\hat{\mathfrak{g}}$.
		As a consequence, $U_q(\hat{\mathfrak{g}})$ may be obtained from $U_q(\tilde{\mathfrak{g}})$ by introducing additional generators $\mathcal{D}^{\pm1}$ and relations
		\begin{align*}
			\mathcal{D}\mathcal{D}^{-1}=\mathcal{D}^{-1}\mathcal{D}=&1\\
			\mathcal{D}\mathcal{H}_{i,r}\mathcal{D}^{-1}=q^r\mathcal{H}_{i,r},\,[\mathcal{D},\mathcal{K}_i]&=[\mathcal{D},\mathcal{C}]=0,\\
			\mathcal{D}\mathcal{X}_{i,r}^\pm\mathcal{D}^{-1}=q^r&\mathcal{X}_{i,r}^\pm.
		\end{align*}
		We emphasize this fact, because the notation in the literature can be confusing. Our definitions are in agreement with the definitions of Kac \cite{Kac}, Carter \cite{Carter}, the book of Chari and Pressley on quantum groups \cite{CPBook} and the paper of Nirov and Razumov \cite{NR}. $\odot$
	\end{rem}
	\begin{rem}[type 1]
		\label{rem:type_1}
		A representation $V$ of $U_q(\tilde{\mathfrak{g}})$ is said to be of type 1, if the $\mathcal{K}_i$ act semisimply with eigenvalues which are integer powers of $q$ and $\mathcal{C}$ acts as the identity. In fact, these are representations of the quotient of $U_q(\tilde{\mathfrak{g}})$ by the ideal generated by $\mathcal{C}^{1/2}-1$, which is the quantum loop algebra $U_q(\mathcal{L}(\mathfrak{g}))$. Conversely, any finite dimensional representation can be obtained by twisting with certain algebra automorphisms (cf. \cite{CPBook} Prop. 12.2.3).
		Thus, we can only consider finite dimensional (type 1) representations of $U_q(\mathcal{L}(\mathfrak{g}))$ as stated above. $\odot$
	\end{rem}
	Similar to the Yangian one has a one-parameter group of automorphisms $\tau_a$, $a\in\mathbb{C}^\times$, of $U_q(\tilde{\mathfrak{g}})$ given by 
	\begin{align}
		\tau_a(\mathcal{X}_{i,r}^\pm)=a^r\mathcal{X}_{i,r}^\pm,\quad\tau_a(\mathcal{H}_{i,s})=a^s\mathcal{H}_{i,s},\quad\tau_a(\mathcal{K}_i^{\pm1})=\mathcal{K}_i^{\pm1},\quad\tau_a(\mathcal{C}^{\pm1/2})=\mathcal{C}^{\pm1/2}.
		\label{eqn:tau_a_q-deformed}
	\end{align}
	If $\mathfrak{g}$ is of type $A$, there is also an evaluation homomorphism $\operatorname{ev}_a$, $a\in\mathbb{C}$, due to Jimbo (1986). Though, for $n\geq2$ it takes values in $U_q(\mathfrak{gl}_{n+1})$ instead of $U_q(\mathfrak{sl}_{n+1})$. Nevertheless, regarding a finite dimensional type 1 representation of $U_q(\mathfrak{sl}_{n+1})$ as a type 1 representation of $U_q(\mathfrak{gl}_{n+1})$, the pullback by $\operatorname{ev}_a$ is well defined and again called \textit{evaluation representation} (cf. \cite{CPBook} 12.2.C). Furthermore, it can be stated that every finite-dimensional irreducible type 1 representation of $U_q(\widetilde{\mathfrak{sl}}_{n+1})$ is isomorphic to a subquotient of a tensor product of evaluation representations (cf. \cite{CPBook} corr. 12.2.14).
	
	To close this part, we cite the last remark of chapter 12.2 in \cite{CPBook}, which explains the similarity between the representation theory of the Yangians and quantum loop algebras.
	\begin{rem}
		The close similarity of the representation theory of the Yangians and quantum loop algebras can be described by an observation due to Drinfeld \cite{D} (1987) at the end of section 6. Let $U_h(\mathcal{L}(\mathfrak{g}))$ be the algebra generated by the elements $\mathcal{H}_{i,r}$, $\mathcal{X}_{i,r}^\pm$ for $i=1,\dots,n$, $r\in\mathbb{Z}$ with defining relations as in definition \ref{def:second_drin_quantum_aff}, but with $q$ replaced by $e^h$, $\mathcal{K}_i$ by $e^{d_ih\mathcal{H}_{i,0}}$ and $\mathcal{C}^{1/2}$ by $1$. Let $\varphi$ be the map $U_h(\mathcal{L}(\mathfrak{g}))\xrightarrow{h=0} U(\mathcal{L}(\mathfrak{g}))\xrightarrow{u=1} U(\mathfrak{g})$ and $A$ be the $\mathbb{C}[[h]]$-subalgebra of $U_h(\mathcal{L}(\mathfrak{g}))\otimes_{\mathbb{C}[[h]]}\mathbb{C}((h))$ generated by $U_h(\mathcal{L}(\mathfrak{g}))$ and $h^{-1}\ker(\varphi)$. Then, $A/hA \cong Y(\mathfrak{g})$. $\odot$
	\end{rem}
	Moreover, there is a general theorem which proves the equivalence of categories of finite dimensional representations of Yangians and of quantum affine algebras given in \cite{GT}.
	\subsection{Graphical notation}
	\label{subsect:graphical_notation}
	Let $V$ be a $n+1$ dimensional complex vector space. We represent $V$ graphically by an oriented line. An operator $O \in \End(V)$ is associated with a symbol, for instance a point, on the line. Let $\ket{v}\in V$ and $A, B \in \End(V)$, then $A\cdot B \ket{v}$ is depicted in Figure \ref{fig:ABv}.
	\begin{figure}[H]
		\centering
		\includegraphics[scale = .6, trim = 6.5cm 11.5cm 12.5cm 8.5cm]{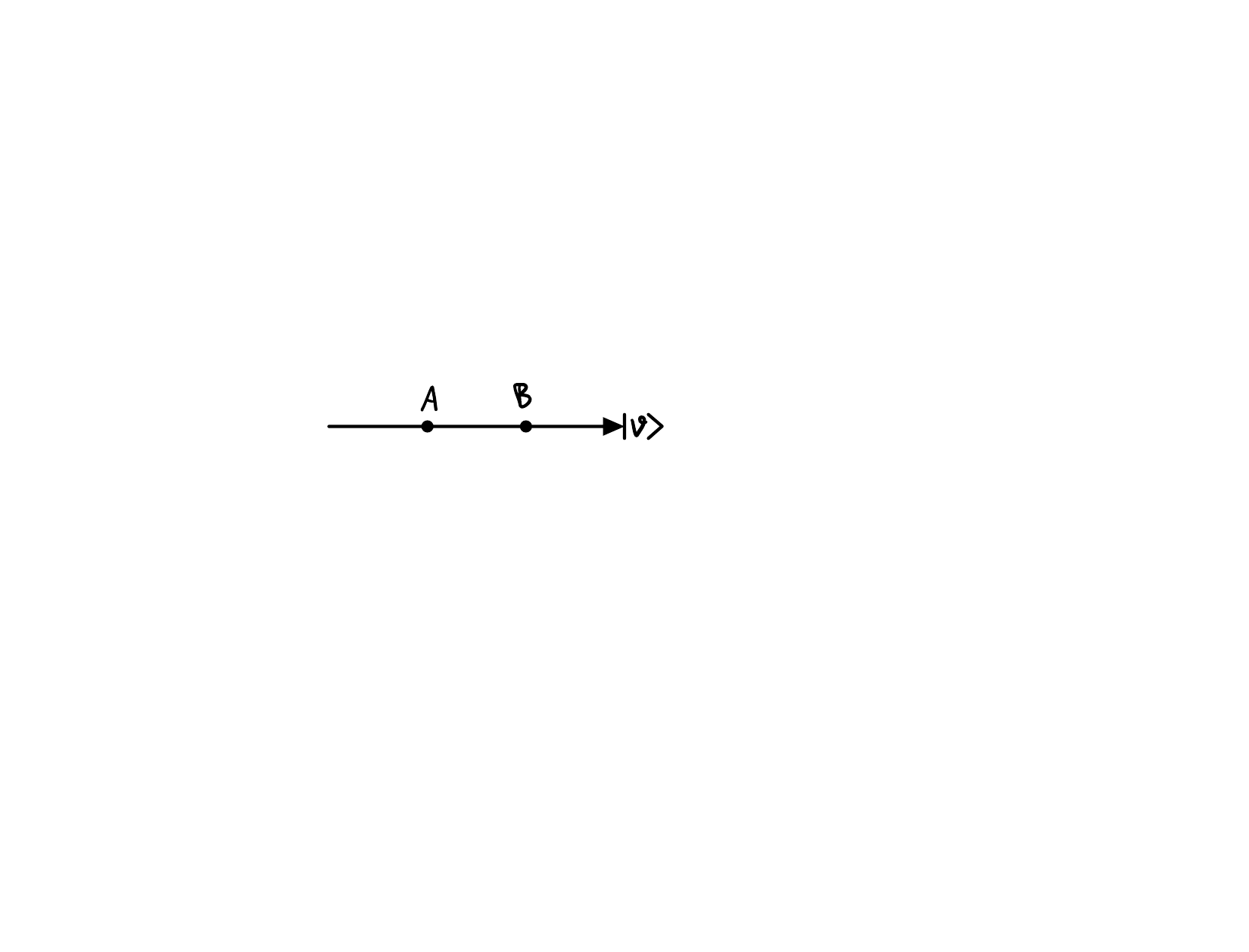}
		\caption{$A\cdot B \ket{v}$.}
		\label{fig:ABv}
	\end{figure}
	We can interpret $V$ as a fundamental representation of the Yangian $Y(\mathfrak{sl}_{n+1})$ by using the evaluation homomorphism $\operatorname{ev}_a:Y(\mathfrak{g})\to U(\mathfrak{g})$ and defining $V\equiv V(a):= \operatorname{ev}_a^*V$ (see remark \ref{rem:ev_a_for_yangian}) where $V$ is the fundamental representation of $\mathfrak{g}$ to the fundamental weight $\omega_1$.\footnote{Of course, any fundamental weight $\omega_i$, $i=1,\dots,n$, can be taken here.} Using the automorphism $\tau_\lambda$ (see proposition \ref{prop:tau_a_yangian}) we can define a family of (fundamental) representations with spectral parameter $\lambda\in\mathbb{C}$ with it, i.e. $V(\lambda):=\tau_\lambda^*(V)$ the pullback of $V$ by $\tau_\lambda$. This can be done for any $\mathfrak{g}$, but in general we don't have an evaluation homomorphism. If a representation $V$ of $\mathfrak{g}$ can be lifted, the ambiguity of defining an origin ("$\operatorname{ev}_0^*$") remains (see \cite{CPBook} Thm. 12.5.3). Thus, we may identify the rational $R$-matrix $(\pi_{V(\lambda)}\otimes\pi_{V(\mu)})(\mathcal{R}) =: R(\lambda-\mu)\in \End(V(\lambda)\otimes V(\mu))$ with any positively oriented vertex between such oriented lines with spectral parameters $\lambda$ and $\mu$, respectively, where $\mathcal{R}$ is the (pseudo-)universal $R$-matrix of the Yangian $Y(\mathfrak{g})$. Note that the orientation of the vertex is naturally induced by the orientation of the lines (c.f. Figure \ref{fig:rational_R-matrix}).
	\begin{figure}[H]
		\centering
		\includegraphics[scale = .6, trim = 10cm 8cm 10cm 8cm]{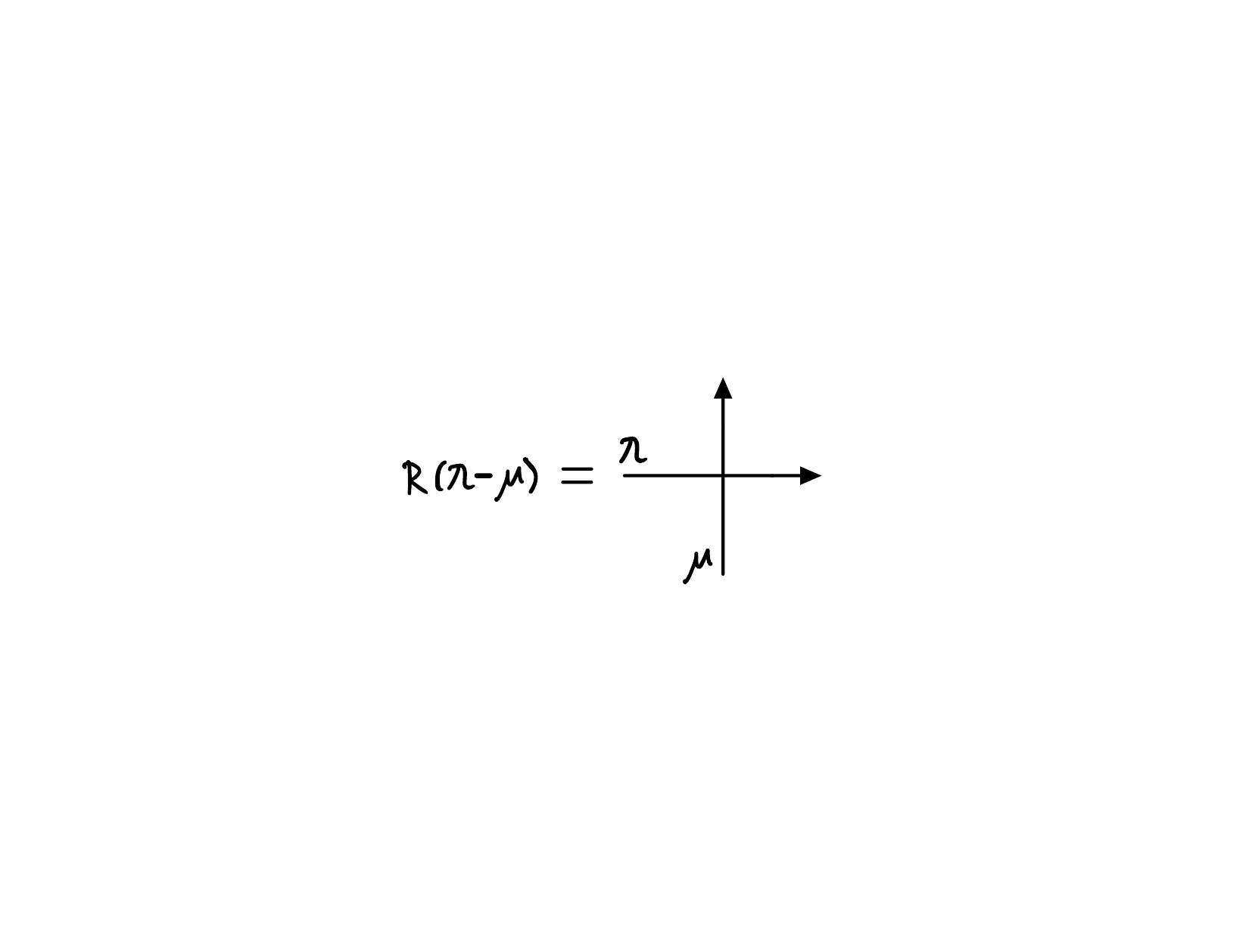}
		\caption{The rational $R$-matrix.}
		\label{fig:rational_R-matrix}
	\end{figure}
	More generally, if $V$ and $W$ are arbitrary $Y(\mathfrak{g})$-modules, we may represent them by different lines and identify $(\pi_V\otimes\pi_W)(\mathcal{R})=: R_{VW}$ with any positive oriented vertex between them. We thus obtain the Yang--Baxter equation (YBE) $R_{UV}R_{UW}R_{VW} = R_{VW}R_{UW}R_{UV}$ for arbitrary $Y(\mathfrak{g})$-modules $U$, $V$ and $W$. It is depicted in Figure \ref{fig:YBE_UVW}.
	\begin{figure}[H]
		\centering
		\includegraphics[scale = .6, trim = 5cm 6cm 6cm 6cm]{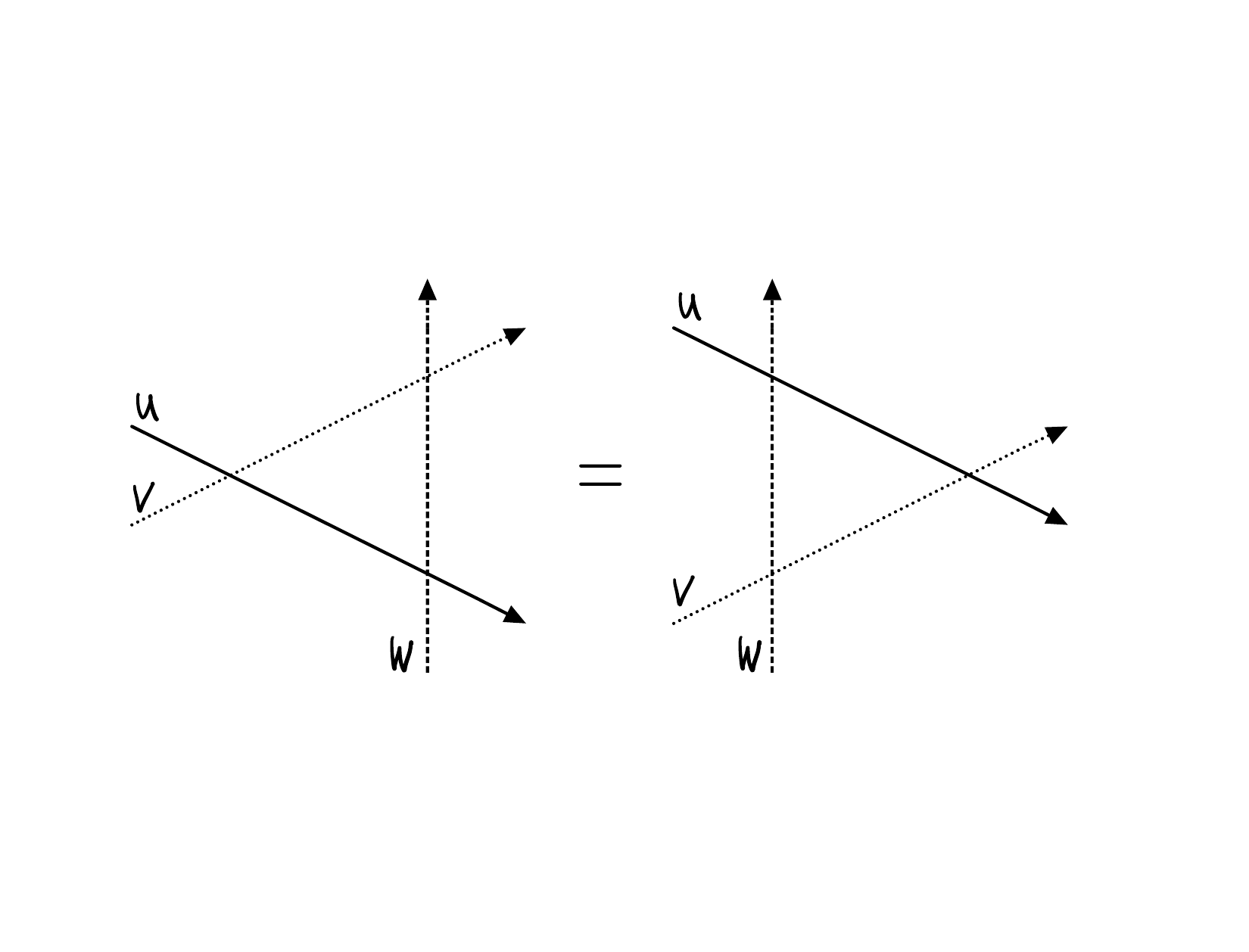}
		\caption{The Yang--Baxter equation (YBE).}
		\label{fig:YBE_UVW}
	\end{figure}
	Let us also introduce the singlet, which we can use to connect a fundamental representation $V_1$ (black) with it's dual (antifundamental) representation $\overline{V}_{\bar{1}}$ (blue).\footnote{ Of course, a singlet in the tensor product of any representation with its dual can be used in this way.} We denote it by a cross. It is depicted in figure \ref{fig:Singlet_and_P_minus} in the top left, whereas its dual vector is depicted in the bottom left. Considering the tensor products between them, we obtain $(n+1)P^-_{1,\bar{1}}$ and $(n+1)P^-_{\bar{1},1}$ (second and third picture from the left in figure \ref{fig:Singlet_and_P_minus}), where $P^-_{1,\bar{1}}$ and $P^-_{\bar{1},1}$ are the projectors onto the singlet in $V_1\otimes \overline{V}_{\bar{1}}$ and $\overline{V}_{\bar{1}}\otimes V_1$, respectively.\footnote{ This also fixes the normalization of the singlet.} The third and fourth picture in figure \ref{fig:Singlet_and_P_minus} show $(n+1)P^-_{1,\bar{1}}P_{\bar{1},1}$ and $(n+1)P^-_{\bar{1},1}P_{1,\bar{1}}$, where in this case $P_{a,b}$ is the permutation operator $V_a\otimes V_b\to V_b\otimes V_a:\;a\otimes b\mapsto b\otimes a$.
	\begin{figure}[H]
		\centering
		\includegraphics[scale = .6, trim = 5cm 8.5cm 6cm 8.5cm]{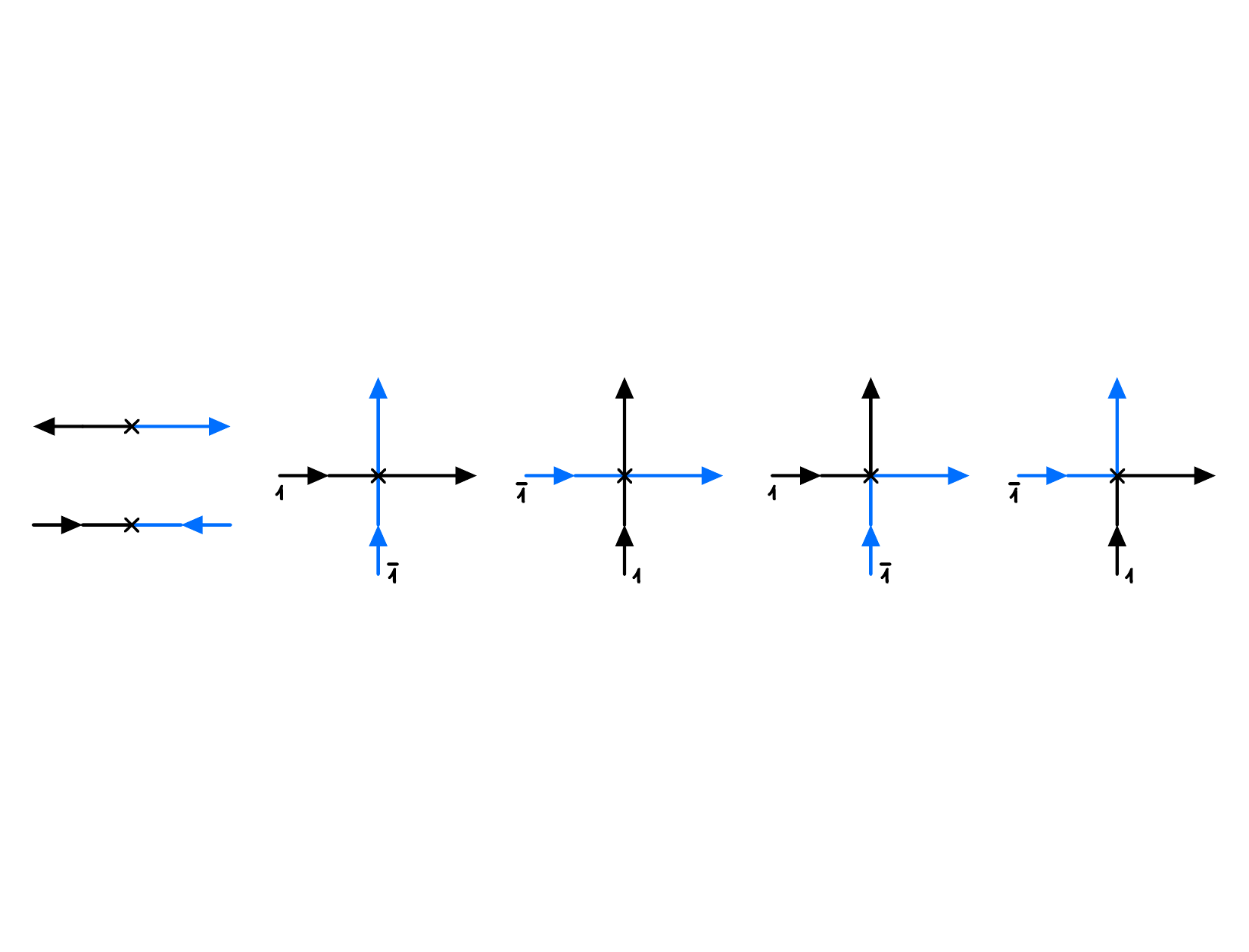}
		\caption{The Singlet and the projector $P^{-}$}
		\label{fig:Singlet_and_P_minus}
	\end{figure}
	To close the subsection let's remark that we use hooks to connect the opposite ends of a line (cf. section \ref{subsect:rat_six_vert_and_density_matrix}). It should be clear that a closed loop canonically corresponds to taking the trace in the corresponding space - "a (closed) line with no open ends corresponds to a constant".
	\subsection{The rational $R$-matrix of type $A_n$}
	\label{def:rational_R-matrix}
	Let $\mathfrak{g}$ be the finite-dimensional simple Lie algebra of type $A_n$ over $\mathbb{C}$ (i.e. $\mathfrak{g}= \mathfrak{sl}_{n+1}(\mathbb{C})$) and denote by $I = \{1,\dots,n\}$ the set of vertices in its Dynkin diagram. Let $\omega_i$ ($i\in I$) be the set of fundamental weights and $\mathcal{P}=\mathbb{Z}\{\omega_i: \, i\in I\}$ be the weight lattice. We use the Khoroshkin--Tolstoy formula \cite{BGKNR} for the (pseudo-)universal $R$-matrix of $U_q(\hat{\mathfrak{g}})$ and take the rational limit $q\to 1$ to define the rational $R$-matrix for two fundamental representations $V(\lambda)$ and $V(\mu)$ with fundamental weight $\omega_1$ as in \cite{BHN}.
	\begin{definition}[the rational $R$-matrix]
		We obtain the rational $R$-matrix $R(\lambda-\mu)\in \End(V(\lambda)\otimes V(\mu))$
		\begin{align}
			R(\lambda-\mu) = \frac{\rho(\lambda-\mu)}{\lambda-\mu+1}\left((\lambda-\mu)1+P\right),\quad\rho(\lambda) = -\frac{\Gamma(\frac{\lambda}{n+1})\Gamma(\frac{1}{n+1}-\frac{\lambda}{n+1})}{\Gamma(-\frac{\lambda}{n+1})\Gamma(\frac{1}{n+1}+\frac{\lambda}{n+1})},
			\label{eqn:Rff}
		\end{align}
		where $1$ is the identity and $P$ is the permutation operator ($P(a\otimes b)= b\otimes a$).\footnote{ By identifying $V(\lambda)$ and $V(\mu)$ as representations of $U(\mathfrak{g}) \lhook\joinrel\xrightarrow{\iota} Y(\mathfrak{g})$} $\odot$
	\end{definition}
	In this sense, we may therefore just write $V$ instead of $V(\lambda)$ and in particular $R\in \End(V\otimes V)$, whenever the spectral parameters of the corresponding lines are clear from the context.
	Of course, since $V(\lambda)\otimes V(\mu)$ is an irreducible $Y(\mathfrak{g})$-module for general $\lambda$ and $\mu$, this definition coincides with the direct definition above through the (pseudo-)universal $R$-matrix of the Yangian up to a scalar factor. To be more precise, let us close this part with the following remark on the scalar prefactor $\rho(\lambda)$.
	\begin{rem}
		The scalar prefactor $\rho(\lambda)$ satisfies the functional relations
		\begin{align}
			\rho(\lambda)\rho(-\lambda)=1,\qquad\rho(\lambda)\rho(n+1-\lambda)=\frac{\lambda(\lambda-(n+1))}{(\lambda-1)(\lambda-n)}.
		\end{align}
		As stated above and in \cite{BHN}, it can be regarded as a quasi-classical or rational limit $q\to1$ of the Khoroshkin--Tolstoy formula \cite{BGKNR}. Alternatively, it can be obtained using the algebraic structure of the Yangian double of $\mathfrak{sl}_{n+1}$ \cite{KT}. $\odot$
	\end{rem}
	\subsection{The dual modules and crossing symmetry}
	\label{subsect:dual modules and crossing}
	For our discussion in this paper, it is essential to provide a good understanding of the dual module $V^*$ (and $^*V$) of $V$, where $V$ is assumed to be either a representation of $Y(\mathfrak{g})$ or $U_q(\tilde{\mathfrak{g}})$. So let $\mathcal{A}$ be either $Y(\mathfrak{g})$ or $U_q(\tilde{\mathfrak{g}})$ and $V$ be a finite dimensional (type 1) representation of $\mathcal{A}$. Since $V$ is a left module, the dual space $V^\star$ is obviously a right module. It can be made into a left module in two ways via the antipode $S$ as follows. We define the dual module $V^*$ to be the left module with the module operation of $a\in\mathcal{A}$ given by $$\br{av|w}=:\br{v|S(a)w},\quad v\in V^\star,\quad w\in V,$$ where $\br{\cdot|\cdot}$ is the dual pairing. For the other dual module $^*V$ we simply replace $S$ by $S^{-1}$ (cf. \cite{NR} for the case of $U_q(\tilde{\mathfrak{g}})$). Since the antipode $S$ is an anti-automorphism of $\mathcal{A}$, so is $S^{-1}$. Thus $V^*$ and $^*V$ are left modules and we call $V^*$ and $^*V$ the dual modules. Since the $\tau_a$ given by (\ref{eqn:tau_a_yangian}) and (\ref{eqn:tau_a_q-deformed}) are Hopf algebra automorphisms, they commute with the action of the antipode $S$, i.e. $S\circ\tau_a=\tau_a\circ S$. We can therefore just write $V^*(\lambda):=\tau_\lambda^*(V^*)= (\tau_\lambda^*(V))^*=(V(\lambda))^*$ (and $^*V(\lambda)$). Now let $\mathcal{R}$ be the (pseudo-)universal $R$-matrix \footnote{ As the Yangian and the quantum affine algebras are not quasitriangular \cite{CPBook}} of $\mathcal{A}$. Using $(S\otimes\id)(\mathcal{R})=\mathcal{R}^{-1}$ and the transposition map $t:\End(V)\to\End(V^\star)$ defined for $M\in\End(V)$ by $$\br{M^tv|w}=\br{v|Mw},\quad v\in V^\star,\quad w\in V,$$ we come to the equation
	\begin{align}
		(\pi_{V^*(\lambda)}\otimes\pi_{V(\mu)})(\mathcal{R})=(\pi_{V(\lambda)}\otimes\pi_{V(\mu)})(\mathcal{R}^{-1})^{t_1},
		\label{eqn:crossing}
	\end{align}
	where $t_1=t\otimes\id$. Equations of the form (\ref{eqn:crossing}) that include $\mathcal{R}$, $\mathcal{R}^{-1}$ and the transpositions $t_1=t\otimes1$ or $t_2=1\otimes t$ are called \textit{crossing relations}. They are obtained by using either $(S\otimes\id)(\mathcal{R})=\mathcal{R}^{-1}$ or $(\id\otimes S^{-1})(\mathcal{R})=\mathcal{R}^{-1}$ and the definition of the dual modules (cf. \cite{NR}).
	
	However, in the case when $\mathcal{A}$ is the Yangian, $\mathfrak{g}$ is of type $A_n$ and $V$ a representation of $\mathfrak{g}$ \footnote{Of course, any representation of the Yangian is naturally a representation of $\mathfrak{g}\lhook\joinrel\xrightarrow{\iota} Y(\mathfrak{g})$}, there is another unique way to define a dual representation of $V(\lambda)=\operatorname{ev}_\lambda^*(V)$. Since the antipode of $\mathfrak{g}$ \footnote{$S(x)=-x$, for all $x\in\mathfrak{g}$.} is its own inverse, there is only one dual representation $V^*$ of $\mathfrak{g}$ and we can define $V^\circledast(\lambda):=\operatorname{ev}_\lambda^*(V^*)$. Note that this definition is different from the definition above since $\operatorname{ev}_a$ is not a Hopf algebra homomorphism. Using the relation $\operatorname{ev}_a\circ S = S\circ \operatorname{ev}_{a-\frac{c}{4}}$, where $c$ is the eigenvalue of the Casimir element of $\mathfrak{g}$ \footnote{It is defined through the standard invariant bilinear form on $\mathfrak{g}$.} in the adjoint representation, we find the relation $V^\circledast(\lambda) = V^*(\lambda-\frac{n+1}{2}) ={}^*V(\lambda+\frac{n+1}{2})$, where we have used that $c=2(n+1)$ for $\mathfrak{g}=L(A_n)=\mathfrak{sl}_{n+1}$.
	
	On one side, for $\mathfrak{g}=\mathfrak{sl}_2$ we have $V^*\cong V$ and therefore $V^\circledast(\lambda)\cong V(\lambda)$ compared to $V^*(\lambda)\cong V(\lambda+1)$ (and $^*V(\lambda)\cong V(\lambda-1)$). On the other side, the square of the antipode and $\tau_a$ are related through $S^2=\tau_{\frac{c}{2}}$ in general. Thus we have $V^{\circledast\circledast}(\lambda)\cong V(\lambda)$, $V^{**}(\lambda)\cong V(\lambda+\frac{c}{2})$ and $^{**}V(\lambda)\cong V(\lambda-\frac{c}{2})$ for the double duals. Similarly, in the case when $\mathcal{A}$ is the quantum affine algebra $U_q(\tilde{\mathfrak{g}})$, we have $S^2=\tau_{q^{c}}\circ \text{Ad}_{q^x}$ (cf. \cite{NR} equation (2.59)), where $x= 2\nu^{-1}(\rho)$, $\rho = \sum_{i=1}^{n}\omega_i$ and $\nu$ is the isomorphism from $\mathfrak{g}$ to $\mathfrak{g}^\star$ by the standard invariant bilinear form and therefore $V^{**}(\lambda)\cong V(\lambda q^{c})$ and $^{**}V(\lambda)\cong V(\lambda q^{-c})$. If $\mathfrak{g}=\mathfrak{sl}_2$ we have $V^*(\lambda)\cong V(\lambda q^{2})$ (and $^*V(\lambda)\cong V(\lambda q^{-2})$). Note that the spectral parameters of the Yangian are additive, whereas the spectral parameters of the quantum affine algebras are multiplicative in our definition. Hence, we can write down the following definition.
	\begin{definition}
		\label{def:dual_rep}
		Let $V$ be a finite dimensional (type 1) representation of $Y(\mathfrak{g})$ (respectively $U_q(\tilde{\mathfrak{g}})$). Then, we denote by $V^\circledast$ the unique dual representation of $\mathfrak{g}\lhook\joinrel\xrightarrow{\iota} Y(\mathfrak{g})$ (respectively $U_q(\mathfrak{g})\lhook\joinrel\xrightarrow{\iota} U_q(\tilde{\mathfrak{g}})$) such that $V^{\circledast\circledast}\cong V$ as a representation of $\mathcal{A}$. As we have seen above, it is $V^\circledast=\tau^*_{-\frac{c}{4}}V^* = \tau^*_{\frac{c}{4}}{}^*V$ (respectively $V^\circledast\cong\tau^*_{q^{-c/2}}V^* \cong \tau^*_{q^{c/2}}{}^*V$). $\odot$
	\end{definition}
	\begin{rem}
		In fact, from the point of view of the finite-dimensional representation theory $V$ and $V^\circledast$ are characterized by the loop weights with corresponding loop parameters and their Drinfeld polynomials have the same roots (see \cite{CPBook} Ch. 12). $\odot$
	\end{rem}
	Using this definition, we can define the crossing transforms of the $R$-matrix as follows.
	\begin{definition}[the crossing transforms]
		\label{def:crossing_transforms}
		Let $V$ be a finite dimensional (type 1) representation of $Y(\mathfrak{g})$ (respectively $U_q(\tilde{\mathfrak{g}})$). Let $\mathcal{R}$ be the (pseudo-)universal $R$-matrix of $Y(\mathfrak{g})$ (respectively $U_q(\tilde{\mathfrak{g}})$). Let $R(\lambda,\mu):=(\pi_{V(\lambda)}\otimes\pi_{V(\mu)})(\mathcal{R})$. Then we define the crossing transforms of $R$ as
		\begin{align}
			R^\circledast(\lambda|\mu):=(\pi_{V^\circledast(\lambda)}\otimes\pi_{V(\mu)})(\mathcal{R})\quad\text{and}\quad R^{\circledast\circledast}(\lambda|\mu):=(\pi_{V(\lambda)}\otimes\pi_{V^\circledast(\mu)})(\mathcal{R}).
		\end{align}
		In the case of the Yangian we have
		\begin{align}
			\notag
			&R^\circledast(\lambda|\mu)= R^\circledast(\lambda-\mu)=(\pi_{V^\circledast(\lambda)}\otimes\pi_{V(\mu)})(\mathcal{R}) = (\pi_{V(\lambda-\frac{c}{4})}\otimes\pi_{V(\mu)})(\mathcal{R}^{-1})^{t_1}\quad\text{and}\\
			\notag
			&R^{\circledast\circledast}(\lambda|\mu)=R^{\circledast\circledast}(\lambda-\mu)=(\pi_{V(\lambda)}\otimes\pi_{V^\circledast(\mu)})(\mathcal{R})=(\pi_{V(\lambda)}\otimes\pi_{V(\mu+\frac{c}{4})})(\mathcal{R}^{-1})^{t_2}.
		\end{align}
		In the quantum affine case we get
		\begin{align}
			\notag
			&R^\circledast(\lambda|\mu)= R^\circledast(\lambda/\mu)=(\pi_{V^\circledast(\lambda)}\otimes\pi_{V(\mu)})(\mathcal{R}) = (\pi_{V(\lambda q^{-c/2})}\otimes\pi_{V(\mu)})(\mathcal{R}^{-1})^{t_1}\quad\text{and}\\
			\notag
			&R^{\circledast\circledast}(\lambda|\mu)=R^{\circledast\circledast}(\lambda/\mu)=(\pi_{V(\lambda)}\otimes\pi_{V^\circledast(\mu)})(\mathcal{R})=(\pi_{V(\lambda)}\otimes\pi_{V(\mu q^{c/2})})(\mathcal{R}^{-1})^{t_2}.
		\end{align}
		$\odot$
	\end{definition}
	Note that there are only two crossing transforms in this case since $V^{\circledast\circledast}\cong V$.
	\begin{definition}[charge conjugation operator]
		Let $V$ be an irreducible finite dimensional (type 1) representation of $Y(\mathfrak{g})$ (respectively $U_q(\tilde{\mathfrak{g}})$). Let $e_i$, $i=0,\dots,m$, be an ordered basis of $V$ such that $e_i$ is a weight vector of weight $\lambda_i\geq\lambda_{i+1}$. Then we define the charge conjugation operator $C=C^{-1}\in\End(V)$ by mapping $e_i$ to $e_{m-i}$, i.e. reversing the order of the basis.
		For instance if $V$ is the fist fundamental representation and $\mathfrak{g}$ of type A, the highest weight vector $e_0$ of weight $\omega_1$ is mapped to the lowest weight vector $e_n$ of weight $-\omega_n$ and vice versa. On its dual space $V^\star$ the dual basis $e_i^\star\eqcolon e^i$ ($\br{e_i^\star|e_j}=\delta^i_j$ \footnote{$\delta$ is the dual map or dual pairing as above.}) of weight $-\lambda_i<-\lambda_{i+1}$ is mapped to $e^\star_{n-i}=e^{n-i}$, $i=0,\dots,n$. Then we have again an ordered basis $\bar{e}_i \coloneq Ce_i^\star=e^\star_{n-i}=e^{n-i}$ such that $\bar{e}_0$ is the highest and $\bar{e}_n$ is the lowest weight vector of weight $\omega_n$ and $-\omega_1$, respectively. However, by abuse of notation we will use the same symbol $C$ in any case. $\odot$
	\end{definition}
	Finally, we can define the (rational) $R$-matrix acting on the tensor product of fundamental and antifundamental representations of $\mathfrak{sl}_{n+1}\lhook\joinrel\xrightarrow{\iota}Y(\mathfrak{sl}_{n+1})$ using its crossing transforms.
	\begin{definition}[$R$-matrix for fundamental and antifundamental representations]
		\label{def:r_matrix_for_f_and_fbar}
		Let $\mathfrak{g}:=\mathfrak{sl}_{n+1}$ and $V$ be the fundamental representation of $\mathfrak{g}$ with highest weight $\omega_1$. Taking into account the fact that highest weight vectors are mapped to lowest weight vectors for the dual representations, we use the charge conjugation operator $C$ to define the $R$-matrix acting on the tensor products of fundamental and antifundamental representations via its crossing transforms (def. \ref{def:crossing_transforms}) similarly to \cite{BHN}, i.e. a change of basis.
		\begin{align}
			\label{eqn:R_bar_fbarf}
			&\bar{R}(\lambda) := (1\otimes C) R^{\circledast\circledast}(\lambda) (1\otimes C^{-1}) = (1\otimes C) \left(R(-\lambda-\frac{n+1}{2})\right)^{t_2} (1\otimes C^{-1})\quad\text{and}\\
			\label{eqn:R_barbar_barff}
			&\bar{\bar{R}}(\lambda) := (C\otimes 1) R^{\circledast}(\lambda) (C^{-1}\otimes 1) = (C\otimes 1) \left(R(-\lambda-\frac{n+1}{2})\right)^{t_1} (C^{-1}\otimes 1),
		\end{align}
		$\odot$
	\end{definition}
	Using this definition and the explicit form of the rational $R$-matrix in definition \ref{def:rational_R-matrix}, we get the expression
	\begin{align}
		\label{eqn:Rbar_fbarf_and_Rbarbar_barff}
		\bar{R}(\lambda)=\frac{\bar{\rho}(\lambda)}{\lambda+\frac{n-1}{2}}((\lambda+\frac{n+1}{2})1-\widetilde{C}\otimes \widetilde{C})=\bar{\bar{R}}(\lambda),
	\end{align}
	where $\bar{\rho}(\lambda):=\rho(\lambda+\frac{n+1}{2})^{-1}$, $\widetilde{C}$ is the morphism from either $V\otimes V^\star\to \mathbb{C}$ or $\mathbb{C}\to V^\star\otimes V$ obtained from $C$ by the natural isomorphism $\sim$. As for $C$, by abuse of notation, we use the same symbol $\widetilde{C}$ for the morphisms from $V^\star\otimes V\to \mathbb{C}$ and $\mathbb{C}\to V\otimes V^\star$ such that $\widetilde{C}\otimes \widetilde{C}\in \End(V\otimes V^\star)$ or $\End(V^\star\otimes V)$. Note that $\bar{R}$ and $\bar{\bar{R}}$ can't be equal as operators, since they act on different spaces, but the operators $1$ (identity) and $\widetilde{C}\otimes \widetilde{C}$ are defined on either $V\otimes V^\star$ or $V^\star\otimes V$ as explained.\footnote{ Roughly speaking, we use the fact that the maps $1$ and $C$ are functors on appropriate categories\\ such that $\bar{R}(\lambda)$ and $\bar{\bar{R}}(\lambda)$ are given in terms of the same (endo)functors, but evaluated on different objects.} Note also that the spectral parameters are slightly shifted compared to \cite{BHN}.
	\subsection{The rational six vertex model and its reduced density matrix}
	\label{subsect:rat_six_vert_and_density_matrix}
	Using the rational $R$-matrix (\ref{eqn:Rff}) and our graphical notation, we can define the partition function of the \textit{rational $\mathfrak{sl}_{n+1}$-invariant model} and its reduced density matrix $D_m$. Let the rational $\mathfrak{sl}_{n+1}$-invariant model be defined on a lattice of the size $L\times M$. As we are interested in the thermodynamic limit when $L,M\to\infty$, we may consider periodic boundary conditions and an even number of horizontal lines $M=2N$, i.e., our model is defined on the square lattice with vertices on the torus $\mathbb{T}_{L,2N}=\mathbb{Z}/L\times \mathbb{Z}/2N$. We can now define the partition function of the model by giving every horizontal and vertical line of the lattice an orientation and a spectral parameter, i.e., fixing the representation $V(\lambda)$ corresponding to a line by making use of our graphical notation defined above. Then the partition function is graphically represented by the lattice itself with oriented horizontal and vertical lines and corresponding spectral parameters. Let's now fix one orientation for the vertical lines, let's say "up" and introduce a staggering for the orientation of the horizontal lines. This is useful for the definition of temperature and the correspondence to the density operator $\rho = \exp(-H_L/T)$ of the $SU(n+1)$ spin chain with periodic boundary conditions (p.b.c.) in the limit $N\to\infty$. To be more precise, define the monodromy matrices
	$T_{a}(\lambda;(\mu_i)_1^L):= R_{a,L}(\lambda-\mu_L)R_{a,L-1}(\lambda-\mu_{L-1})\cdots R_{a,1}(\lambda-\mu_1)$ and $\overline{T}_{a}(\lambda;(\mu_i)_1^L):=T^{-1}_{a}(\lambda;(\mu_i)_1^L)= R_{1,a}(\mu_1-\lambda)R_{2,a}(\mu_{2}-\lambda)\cdots R_{L,a}(\mu_L-\lambda)$ as in Figure \ref{fig:monodromy_matrices}. Then, the transfer matrices $t$ and $\overline{t}$ are defined through the traces of the monodromy matrices $t(\lambda;(\mu_i)_1^L) := \tr_a(T_{a}(\lambda;(\mu_i)_1^L))$ and $\overline{t}(\lambda;(\mu_i)_1^L) := \tr_a(\overline{T}_{a}(\lambda;(\mu_i)_1^L))$ over the auxiliary space $a$. This is how the spaces corresponding to the horizontal lines are called, whereas the spaces corresponding to vertical lines are called the quantum space. Setting the parameters $\mu_i$, $i=1,\dots,L$, of the vertical lines to zero and the parameters for the $N$ horizontal lines pointing towards the right (resp. left) to $\frac{\beta}{2N}$ (resp. $-\frac{\beta}{2N}$), we have
	\begin{align*}
		\exp\left(-\beta H_L\right) = \lim_{N\to\infty}\left(t\left(-\frac{\beta}{2N}\right)\overline{t}\left(\frac{\beta}{2N}\right)\right)^N
	\end{align*}
	in operator norm on the space $V^{\otimes L}$, where $H_L$ is the Hamiltonian of the $SU(n+1)$ spin chain of length $L$ with p.b.c. and $\beta = 1/T$ the inverse temperature. In this case we omit writing the dependence on the $\mu_i=0$, $i=1,\dots,L$.
	\begin{figure}[H]
		\centering
		\includegraphics[scale = .425, trim = 5.5cm 5cm 5cm 5cm]{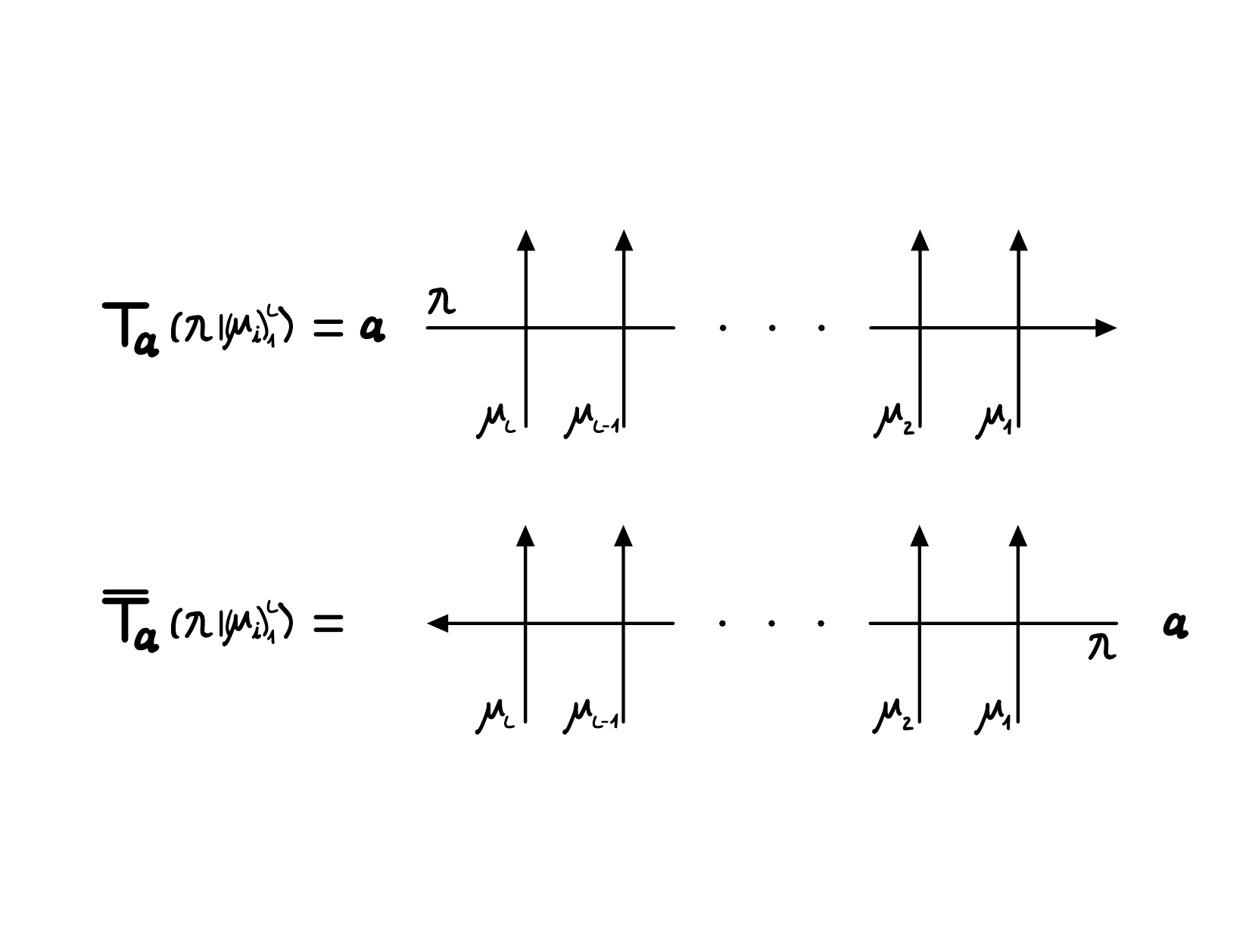}
		\caption{The monodromy matrices $T_{a}(\lambda;(\mu_i)_1^L)$ and $\overline{T}_{a}(\lambda;(\mu_i)_1^L)$.}
		\label{fig:monodromy_matrices}
	\end{figure}
	The partition function of this model is depicted in figure \ref{fig:partition_function}, where we introduce hooks to show that a line forms a closed loop ,i.e. we take the trace over the corresponding space.
	\begin{figure}[H]
		\centering
		\includegraphics[scale = .425, trim = 5cm 3cm 6cm 2cm]{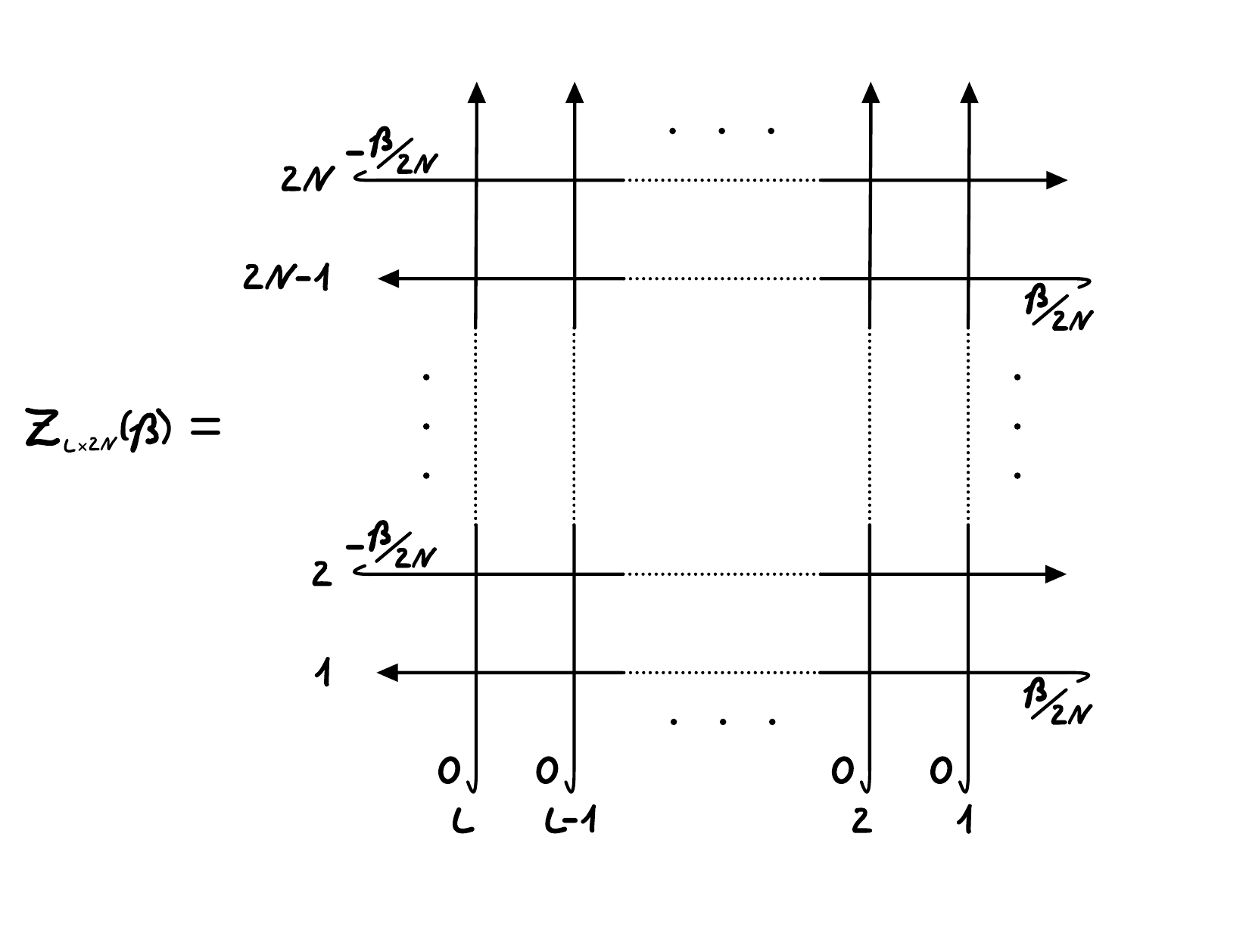}
		\caption{Partition function of the $L\times 2N$ staggered vertex model.}
		\label{fig:partition_function}
	\end{figure}
	It is the trace over the spaces $1,\dots,L$ of the product of the transfer matrices $$Z_{L\times 2N}(\beta) = \tr_{1,\dots,L}\left[\left(t\left(-\frac{\beta}{2N}\right)\overline{t}\left(\frac{\beta}{2N}\right)\right)^N\right].$$ Note that in the limit $N\to\infty$ we obtain the partition function of the $SU(n+1)$ spin chain $Z_L(\beta)= \tr_{1,\dots,L}[\exp(-\beta H_L)]$.
	\begin{definition}
		\label{def:gen_density_matrix}
		Let $\mathbb{N} \ni m\leq L$ and consider the observables $O \in \End(V^{\otimes m})$. We define the (generalized) reduced density matrix $D_m(\mu_1,\dots,\mu_m) \in \left(\End(V^{\otimes m})\right)^\star$ by
		\begin{align}
			\notag &D_m(\mu_1,\dots,\mu_m)(O) :=\frac{1}{Z_{L\times 2N}(\beta)}\times \\ &\tr_{1,\dots,L}\left[O_{1,\dots,m}\left(t\left(-\frac{\beta}{2N};\mu_m,\dots,\mu_1,0,\dots,0\right)\overline{t}\left(\frac{\beta}{2N};\mu_m,\dots,\mu_1,0,\dots,0\right)\right)^N\right],
		\end{align}
		where we identify $O\in\End(V^{\otimes m})$ with $O_{1,\dots,m}\in\End(V^{\otimes L})$ by 
		\begin{align*}
			\End(V^{\otimes m})\lhook\joinrel\xrightarrow{\iota_{mL}}\End(V^{\otimes L}): 	O\mapsto O_{1,\dots,m} := O\otimes 1^{\otimes (L-m)},
		\end{align*}
		where $1\in \End(V)$ is the identity operator. $\odot$
	\end{definition}
	\begin{rem}
		Using the periodicity it should be clear that every observable $O\in \End(V^{\otimes L})$ that acts non trivially on a segment $V^{\otimes m}$ of length $m\leq L$ can be identified with an element $O\in \End(V^{\otimes m})$. In this sense, we have a morphism $D$ which represents the density matrix for every $m\in\mathbb{N}$. It is the unique morphism defined on the direct limit $\varinjlim\End(V^{\otimes m})$ \footnote{To be precise, we should replace $\End(V^{\otimes m})$ with $\End(V^{\otimes L})$ for all $m\geq L$ in the case when $L$ is kept finite.}. We write $D=:\varinjlim D_m$. $\odot$
	\end{rem}
	\begin{rem}
		Note that we choose the numbering of the parameters $\mu_i$, $i=1,\dots,m$, from right to left as in \cite{BHN}. This is of course a matter of taste but it may be easier to have a corresponding notation. We also choose the numbering of the spaces according to the parameters from right to left. However, setting the parameters $\mu_i$, $i=1,\dots,m$, to zero we obtain the reduced density matrix $D_m$ of the $L\times 2N$ staggered vertex model with temperature $T$ (cf. figure \ref{fig:reduced_density_finite}). Furthermore, by abuse of notation we may also write $D_m(\mu_1,\dots,\mu_m)$ in any of the limits $N\to\infty,\,L \to \infty$ and $T\to 0$. $\odot$
	\end{rem}
	Let us now try to provide a good understanding of the reduced qKZ equations on the infinite lattice ($L,N\to\infty$). We put the upper index $(0)$ if we refer to the usual density matrix $D_m=:D_m^{(0)}$ and introduce an additional density matrix $D_m^{(1)}$ where the vertical line with parameter $\mu_1$ is replaced by an antifundamental line with the same parameter. It is depicted in figure \ref{fig:reduced_density_antifundamental}. In the finite lattice case, we obtain the difference equations derived in the appendix of \cite{BHN}. They are given by
	\begin{align}
		&\notag D^{(0)}_m(\beta_j,\mu_2,\dots,\mu_m|\beta_1,\dots,\beta_N;\bar{\beta}_1,\dots,\bar{\beta}_N)\left(A^{(1)}_{1,\bar{1}|2,\dots,m}(\beta_j|\mu_2,\dots,\mu_m)\left(X_{\bar{1},2,\dots,m}\right)\right)=\\ &D^{(1)}_m(\beta_j-\frac{n+1}{2},\mu_2,\dots,\mu_m|\beta_1,\dots,\beta_N;\bar{\beta}_1,\dots,\bar{\beta}_N)\left(X_{\bar{1},2,\dots,m}\right)\quad\text{and}\\
		&\notag D^{(1)}_m(\bar{\beta}_j,\mu_2,\dots,\mu_m|\beta_1,\dots,\beta_N;\bar{\beta}_1,\dots,\bar{\beta}_N)\left(A^{(2)}_{\bar{1},1|2,\dots,m}(\bar{\beta}_j|\mu_2,\dots,\mu_m)\left(X_{1,\dots,m}\right)\right)=\\ &D^{(0)}_m(\bar{\beta}_j-\frac{n+1}{2},\mu_2,\dots,\mu_m|\beta_1,\dots,\beta_N;\bar{\beta}_1,\dots,\bar{\beta}_N)\left(X_{1,\dots,m}\right)
	\end{align}
	in general, where the $N$ horizontal lines pointing towards the left have the spectral parameters $\beta_j$, $j=1,\dots,N$, and the $N$ horizontal lines pointing towards the right are replaced by antifundamental lines with parameters $\bar{\beta_j}=\frac{n+1}{2}-\beta_j$ pointing towards the left by means of the crossing symmetry (cf. subsection \ref{subsect:dual modules and crossing}). Setting $\beta_j$ to $\frac{\beta}{2N}$ we come back to the density matrices $D^{(0)}$ and $D^{(1)}$ for the temperature $T=\frac{1}{\beta}$ with homogeneous horizontal parameters as defined above. Taking the limits $N$ and $L$ to infinity carefully, we obtain real difference equations usually referred to as the reduced-qKZ-equation (cf. \cite{KNR} and \cite{BHN}). A precise definition of the operators $A^{(1)}$ and $A^{(2)}$ is given in the appendix \ref{App:A}. For the discussion of the limits we refer to the paper \cite{BHN}.\newpage
	\begin{figure}[H]
		\centering
		\includegraphics[scale = .5, trim = 5.75cm 1.75cm 5cm 2.5cm]{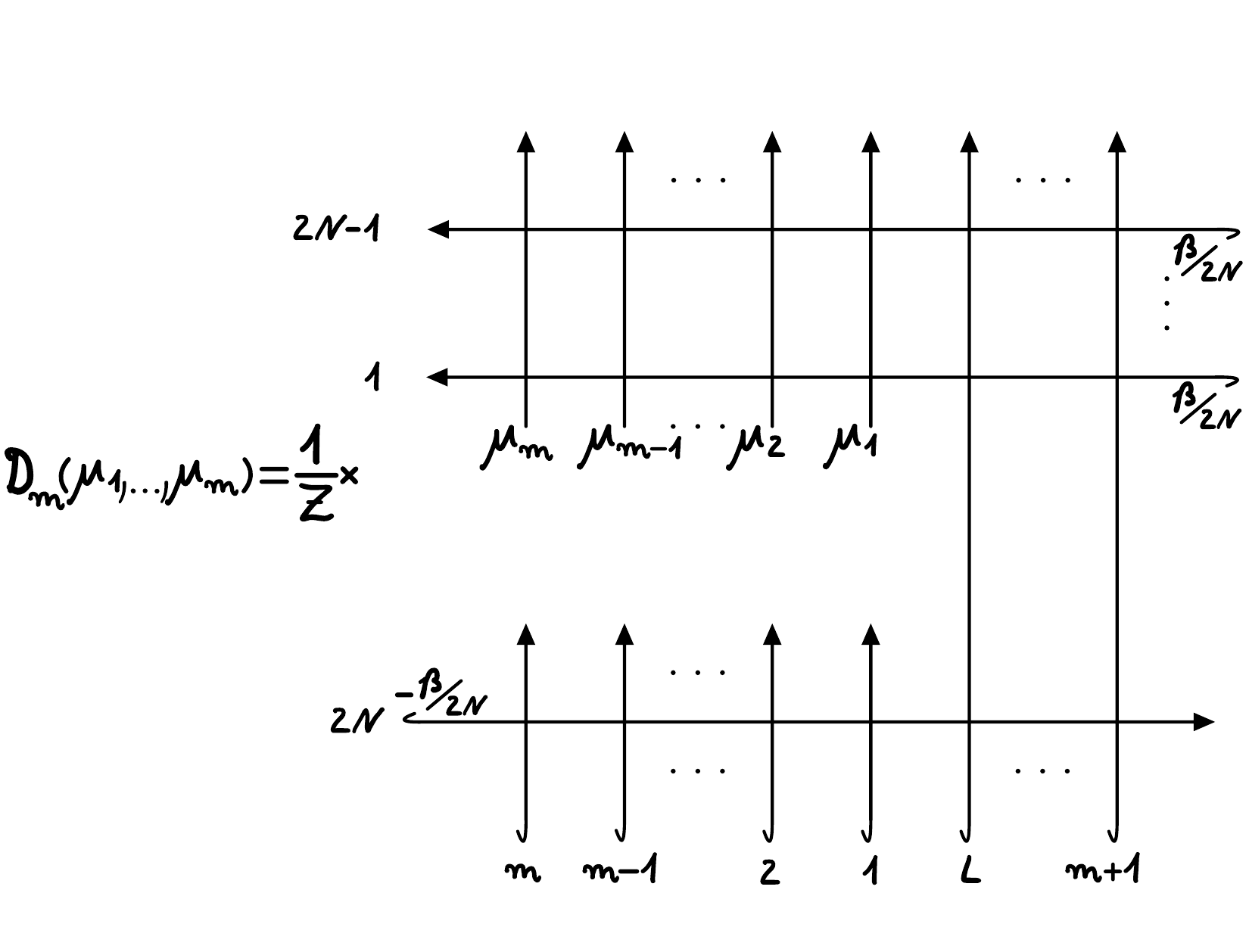}
		\caption{The reduced density matrix $D_m\equiv D_m^{(0)}$ of the $L\times 2N$ staggered six vertex model.}
		\label{fig:reduced_density_finite}
	\end{figure}
	\begin{figure}[H]
		\centering
		\includegraphics[scale = .5, trim = 5.75cm 1.75cm 5cm 3cm]{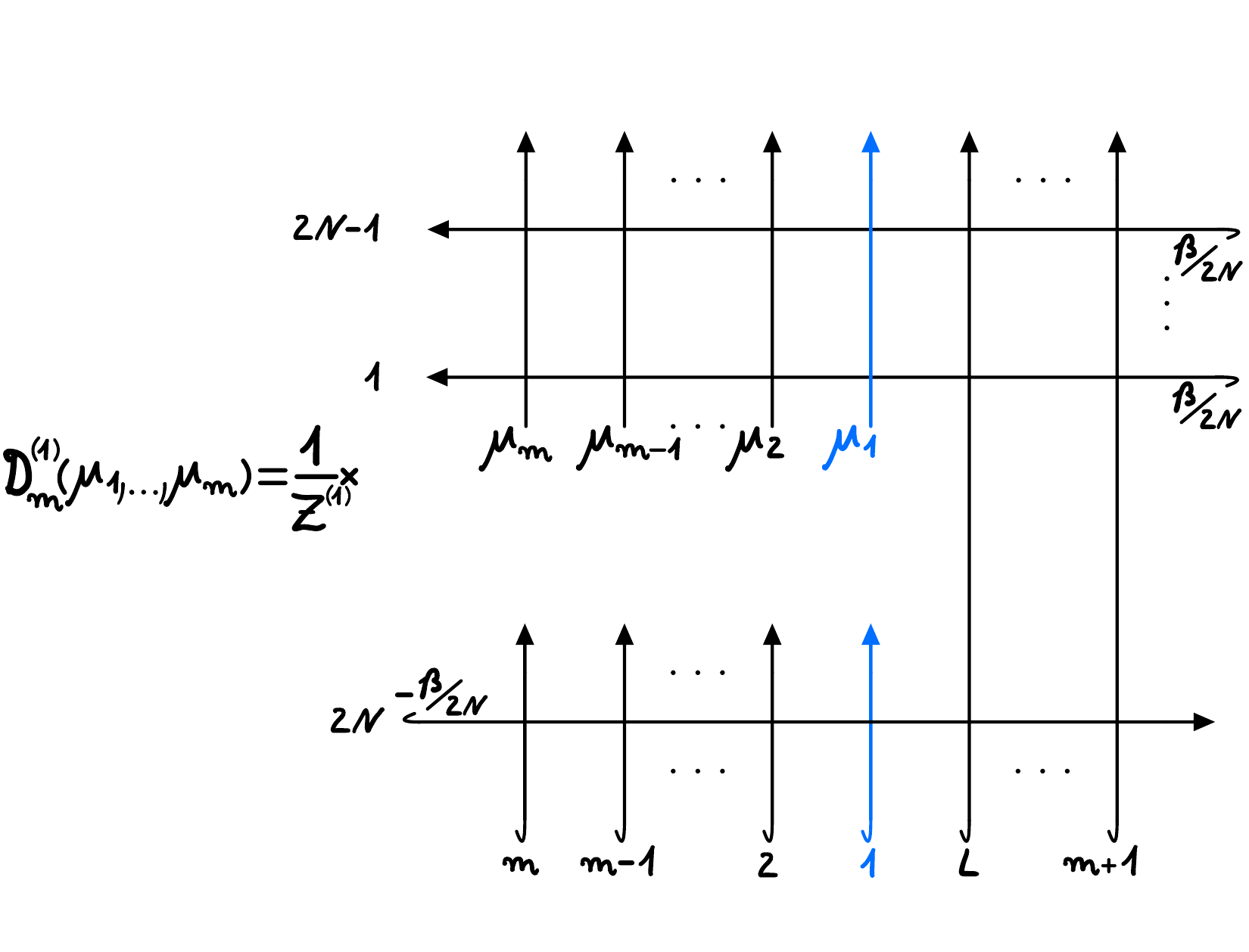}
		\caption{The reduced density matrix $D_m^{(1)}$ of the $L\times 2N$ staggered six vertex model.}
		\label{fig:reduced_density_antifundamental}
	\end{figure}
	\newpage
	\section{The construction for $\mathfrak{sl}_2$}
	\label{sect:sl2snailconstr}
	Now after we have introduced our graphical notation, we are in position to review the construction in the paper \cite{BJMST}. It is based on the fact that the density matrix $D$ satisfies a certain difference equation in the infinite lattice limit, the reduced quantum Knizhnik--Zamolodchikov (rqKZ) equation. Furthermore, $D_m$ satisfies a list of properties \cite{JM}. We write the indices $1,\dots,m$ whenever we need to clarify the corresponding spaces for $D$, i.e. $D_{1,\dots,m}\equiv D_m$.
	\begin{prop}
		\label{prop:properties_of_D_sl2}
		The functional $D$ \footnote{Using the transposition isomorphism $\End(V^{\otimes m})^\star\cong \End(V^{\otimes m})$ we identify $D_m:\End(V^{\otimes m})\to\mathbb{C}$ for any $m\in\mathbb{N}$ (cf. definition \ref{def:gen_density_matrix}) with an element in $\End(V^{\otimes m})$.} possesses the following properties:
		\begin{enumerate}
			\item $D_m$ is invariant under the action of $\mathfrak{sl}_2$.\label{sl2_inv}
			\item $D_m$ satisfies the $R$-matrix relations \label{R_matrix_rel_sl2}
			\begin{align*}
				&D_{1,\dots,i+1,i,\dots,m}(\lambda_1,\dots,\lambda_{i+1},\lambda_i,\dots,\lambda_m) =\\ &R_{i+1,i}(\lambda_{i+1,i})D_{1,\dots,m}(\lambda_1,\dots,\lambda_m)R_{i,i+1}(\lambda_{i,i+1}).
			\end{align*}
			\item $D$ has the left-right reduction property \label{left-right_sl2}
			\begin{align*}
				\operatorname{tr}_1(D_{1,\dots,m}(\lambda_1,\dots,\lambda_m)) &= 	D_{2,\dots,m}(\lambda_2,\dots,\lambda_m)\\
				\operatorname{tr}_m(D_{1,\dots,m}(\lambda_1,\dots,\lambda_m)) &= D_{1,\dots,m-1}(\lambda_1,\dots,\lambda_{m-1})
			\end{align*}
			for all $m\in\mathbb{N}$, where $D_{1,\dots,m-1}(\lambda_1,\dots,\lambda_{m-1}):=1$ for $m=1$.
			\item The rqKZ-equation \label{rqKZ_sl2}
			\begin{align*}
				D_{1,\dots,m}(\lambda_1-1,\lambda_2,\dots,\lambda_m) = A_{\bar{1},1|2,\dots,m}(\lambda_1|\lambda_2,\dots,\lambda_m)(D_{\bar{1},2,\dots,m}(\lambda_1,\lambda_2,\dots,\lambda_m)).
			\end{align*}
			\item $D_{1,..,m}(\lambda_1,\dots,\lambda_m)$ is meromorphic in $\lambda_1,\dots,\lambda_m$ with at most simple poles at $\lambda_i-\lambda_j \in \mathbb{Z}\backslash\{0,\pm 1\}$.\label{analyt_sl2}
			\item For all $0<\delta<\pi$:
			\begin{align*}
				\lim\limits_{\substack{\lambda_1\to \infty \\ \lambda_1\in S_\delta}} D_{1,\dots,m}(\lambda_1,\dots,\lambda_m) = \frac{1}{2}\boldsymbol{1}_1D_{2,\dots,m}(\lambda_2,\dots,\lambda_m),
			\end{align*}
			where $S_\delta :=\{\lambda\in\mathbb{C}|\delta<|\arg(\lambda)|<\pi-\delta\}$. \label{asympt_sl2} $\odot$
			\seti	
		\end{enumerate}
	\end{prop}
	To keep the notation short we omit the indices $2,\dots,m$ and write $A_{\bar{1},1}(\lambda_1|\lambda_2,\dots,\lambda_m)$ instead of $A_{\bar{1},1|2,\dots,m}(\lambda_1|\lambda_2,\dots,\lambda_m)$ from now on. $A_{\bar{1},1}(\lambda_1|\lambda_2,\dots,\lambda_m)$ is given by either $A^{(1)}_{1,\bar{1}}(\lambda_1|\lambda_2,\dots,\lambda_m)$ or $A^{(2)}_{\bar{1},1}(\lambda_1|\lambda_2,\dots,\lambda_m)$ in the Appendix \ref{App:A}. It simplifies because the fundamental and antifundamental representations are isomorphic in the case of $\mathfrak{sl}_2$ ($V\cong\overline{V}$). $A_{\bar{1},1}(\lambda_1|\lambda_2,\dots,\lambda_m)$ is depicted in figure \ref{fig:A_1bar1}, where the cross on the bottom right stands for the operator $-2PP^-$ with $P^-$ being the projector onto the singlet in $V_1(\lambda_1)\otimes V_{\bar{1}}(\lambda_1-1)$. Identifying $V(\lambda_1)$ and $V(\lambda_1-1)$ as representations of $U(\mathfrak{g}) \lhook\joinrel\xrightarrow{\iota} Y(\mathfrak{g})$ as above, we have $-2PP^-=2P^-$, which is how we define any other vertex with a cross that doesn't interchange the spectral parameters of the horizontal and vertical line.\footnote{In other words it doesn't change the orientation of the outgoing lines.}
	\begin{figure}[H]
		\centering
		\includegraphics[scale = .55, trim = 1cm 9.5cm 3cm 6cm]{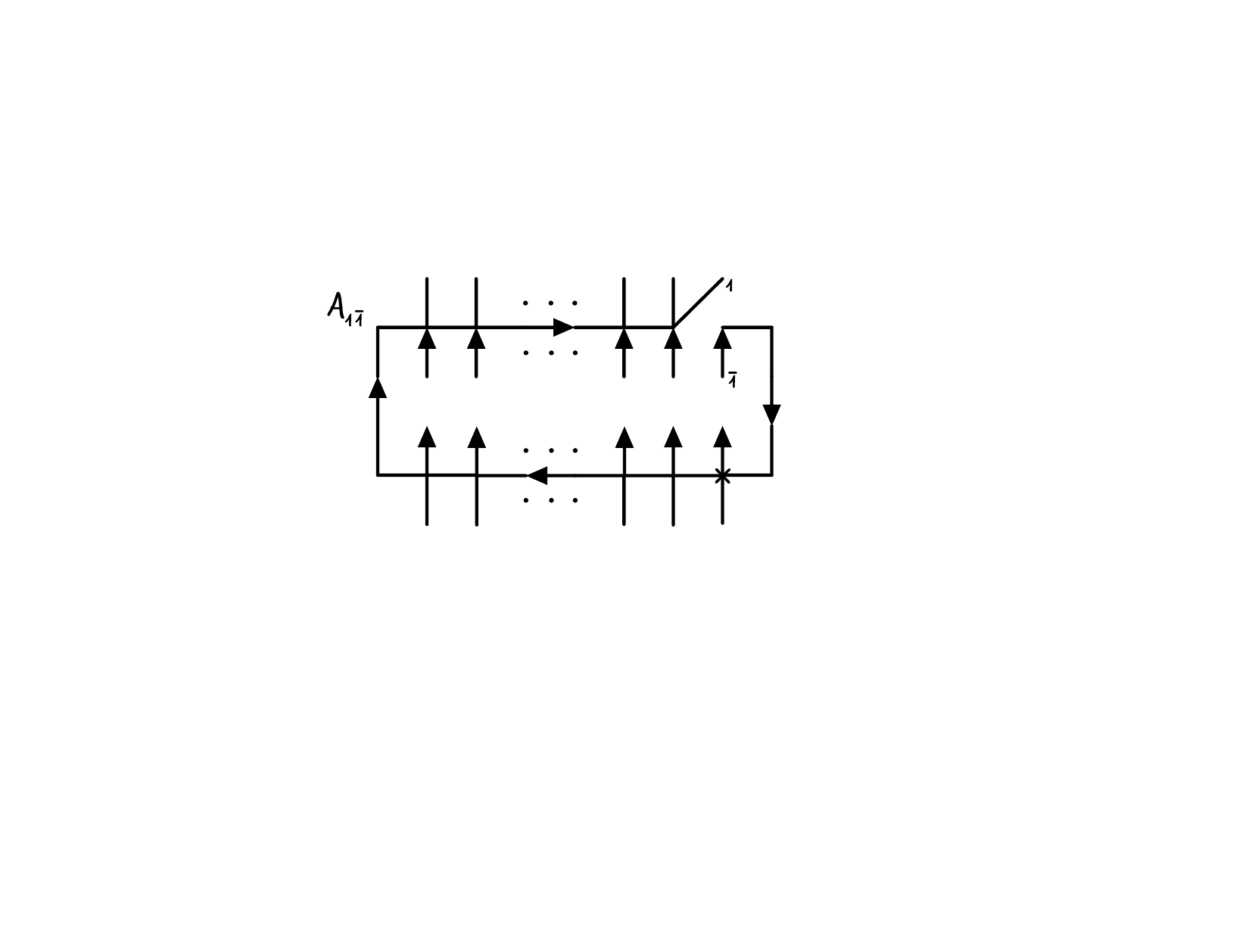}
		\caption{The operator $A_{1,\bar{1}}$}
		\label{fig:A_1bar1}
	\end{figure}
	The analytic properties (\ref{analyt_sl2}) and (\ref{asympt_sl2}) in proposition \ref{prop:properties_of_D_sl2} are obtained from an integral formula constructed in \cite{JM} and \cite{KIEU}. They are derived in the Appendix B of \cite{BJMST}.
	\begin{rem}\hspace{1em}
		\label{rem:rel7sl2}
		\begin{itemize}
			\item Due to $\rho(\lambda)\rho(-\lambda) = 1$ and $\rho(\lambda-1)\rho(\lambda) = -\frac{\lambda}{\lambda-1}$ the coefficients in (\ref{R_matrix_rel_sl2}) and (\ref{rqKZ_sl2}) in proposition \ref{prop:properties_of_D_sl2} are rational.
			\item $D_{1,\dots,m}(\lambda_1,\dots,\lambda_m)$ is translational invariant
			\begin{align*}
				D_{1,\dots,m}(\lambda_1+u,\dots,\lambda_m+u) = D_{1,\dots,m}(\lambda_1,\dots,\lambda_m)
			\end{align*}
			\item $D_{1,\dots,m}(\lambda_1,\dots,\lambda_m)$ fulfils the spin conservation rule
			\begin{align*}
				\left[D_{1,\dots,m}(\lambda_1,\dots,\lambda_m)\right]_{\epsilon_1\dots\epsilon_m}^{\bar{\epsilon}_1\dots\bar{\epsilon}_m} = 0\quad \text{if}\quad m_1(\epsilon)\neq m_1(\bar{\epsilon}),
			\end{align*}
			where the components of $D$ are given by
			\begin{align*}
				\left[D_{1,\dots,m}(\lambda_1,\dots,\lambda_m)\right]_{\epsilon_1\dots\epsilon_m}^{\bar{\epsilon}_1\dots\bar{\epsilon}_m}:=D_{1,\dots,m}(\lambda_1,\dots,\lambda_m)\left((\tensor{E}{_{\epsilon_1}^{\bar{\epsilon}_1}})_1\cdots (\tensor{E}{_{\epsilon_m}^{\bar{\epsilon}_m}})_m\right),
			\end{align*}
			$\tensor{E}{_{\epsilon_i}^{\bar{\epsilon}_i}} = e_{\epsilon_i}\otimes e^{\bar{\epsilon}_i}$ and $m_1(\epsilon)$ is the number of $\epsilon_i$, $i=1,\dots,m$, with $\epsilon_i=1$.
			\item (\ref{sl2_inv}) - (\ref{asympt_sl2}) in proposition \ref{prop:properties_of_D_sl2} determine $D_m$ completely (see\cite{BJMST}).
		\end{itemize}
		\begin{enumerate}
			\conti
			\item From (\ref{R_matrix_rel_sl2}), (\ref{left-right_sl2}), (\ref{rqKZ_sl2}) and the analyticity of $D_m$ at $\lambda_1=\lambda_2$ we obtain
			\begin{align*}
				P^-_{12} D_{1,2,\dots,m}(\lambda-1,\lambda,\dots,\lambda_n) = P^-_{12}D_{3,\dots,m}(\lambda_3,\dots,\lambda_n),
			\end{align*}
			where $(P^{-}_{1\bar{1}})^2=P^{-}_{1\bar{1}}$ is the projector onto the singlet. $\odot$
			\label{proj_red_rel_sl2}
		\end{enumerate}
	\end{rem}
	As $D_m$ is meromorphic in $\lambda_1$ with at most simple poles, it is completely determined by its residues and asymptotic behaviour. Using the rqKZ equation repeatedly, a relation of the form
	\begin{align}
		\notag 
		&\underset{\lambda_{1,j} = \pm(k+1)}{\text{res}} D_{1,\dots,m}(\lambda_1,\dots,\lambda_m) =\\
		&\underset{\lambda_{1,j} = \pm(k+1)}{\text{res}} \left\{\mp\frac{\omega(\lambda_{1,j})}{1-\lambda_{1,j}^2}\tilde{X}^{[1,j]}(\lambda_1,\dots,\lambda_m)\right\}(D_{m-2}(\lambda_2,\dots,\hat{\lambda_j},\dots,\lambda_m))
		\label{eqn:red_rel_res_sl2}
	\end{align}
	for the residues of $D_{1,\dots,m}(\lambda_1,\dots,\lambda_m)$ is proven in \cite{BJMST}, where $\frac{\omega(\lambda_{1,j})}{1-\lambda_{1,j}^2}\tilde{X}^{[1,j]}(\lambda_1,\dots,\lambda_m)$ is a single meromorphic function.\footnote{ We use the short hand notation $\lambda_{ij}:=\lambda_i-\lambda_j$.} Furthermore, the asymptotics of $D_m$ were calculated such that a reduction relation for the reduced density matrix $D_m$ could be obtained using Liouville's theorem. Let us derive equation \ref{eqn:red_rel_res_sl2} in our notation. Using the $R$-matrix relations (\ref{R_matrix_rel_sl2}) in proposition \ref{prop:properties_of_D_sl2} we can suppose $j=2$. Then, by applying the rqKZ equation (\ref{rqKZ_sl2}) in proposition \ref{prop:properties_of_D_sl2} repeatedly, $D_{1,\dots,m}(\lambda_1-k-1,\lambda_2,\dots,\lambda_m)$ can be expressed in terms of several $A$'s with shifted arguments acting on $D_{1,\dots,m}(\lambda_1-1,\lambda_2,\dots,\lambda_m)$
	\begin{align}
		\notag
		&\underset{\lambda_{1} =\lambda_2 -(k+1)}{\text{res}} D_{1,\dots,m}(\lambda_1,\dots,\lambda_m) = \underset{\lambda_{1} = \lambda_2}{\text{res}} D_{1,2,\dots,m}(\lambda_1-(k+1),\lambda_2,\dots,\lambda_m) =\\
		&\underset{\lambda_{1} = \lambda_2}{\text{res}} \{A_{\bar{b},b}(\lambda_1-k|\lambda_2,\dots,\lambda_m)\cdots A_{a,\bar{a}}(\lambda_1-1|\lambda_2,\dots,\lambda_m)D_{a,2,\dots,m}(\lambda_1-1,\lambda_2,\dots,\lambda_m)\},
		\label{eqn:rqKZ_k_shift_sl2}
	\end{align}
	where the indices are $a=b=1$ if $k$ is even and $a=\bar{1}=\bar{b}$ when $k$ is odd setting $\bar{\bar{1}}:=1$.
	Finally, one uses the relation (\ref{proj_red_rel_sl2}) in remark \ref{rem:rel7sl2} to see that $D_{3,\dots,m}(\lambda_3,\dots,\lambda_m)$ can be pulled out of the residue as it doesn't depend on $\lambda_{12}=:\lambda_1-\lambda_2$. This is only possible, because the residue of $A_{a,\bar{a}}(\lambda_1-1|\lambda_2,\dots,\lambda_m)$ at $\lambda_1=\lambda_2$ contains the projector $P^-_{a,2}$ just at the right position, i.e. we apply the relation (\ref{proj_red_rel_sl2}) in remark \ref{rem:rel7sl2}. We obtain
	\begin{align*}
		&\underset{\lambda_{1} =\lambda_2 -(k+1)}{\text{res}} D_{1,\dots,m}(\lambda_1,\dots,\lambda_m) =\\
		&\underset{\lambda_{1} = \lambda_2}{\text{res}} \{A_{\bar{b},b}(\lambda_1-k|\lambda_2,\dots,\lambda_m)\cdots A_{a,\bar{a}}(\lambda_1-1|\lambda_2,\dots,\lambda_n)\}D_{3,\dots,m}(\lambda_3,\dots,\lambda_m)
	\end{align*}
	and see that the product $\underset{\lambda_{1} = \lambda_2}{\text{res}} \{A_{\bar{b},b}(\lambda_1-k|\lambda_2,\dots,\lambda_m)\cdots A_{a,\bar{a}}(\lambda_1-1|\lambda_2,\dots,\lambda_m)\}$ is one (quite unhandy) expression for the Snail Operator $\tilde{X}_k$. Suppose for simplicity $k=2$, then we apply the rqKZ equation two times (figure \ref{fig:rqKZ_sl2_2_loops}). \footnote{ In the figures we use $\mu$ instead of $\lambda$ for the spectral parameters. It is of course just a matter of taste.}
	\begin{figure}[H]
		\centering
		\begin{minipage}[b]{.4\linewidth} 
			\includegraphics[scale = .45, trim = 5cm 6.25cm 3cm 2.75cm]{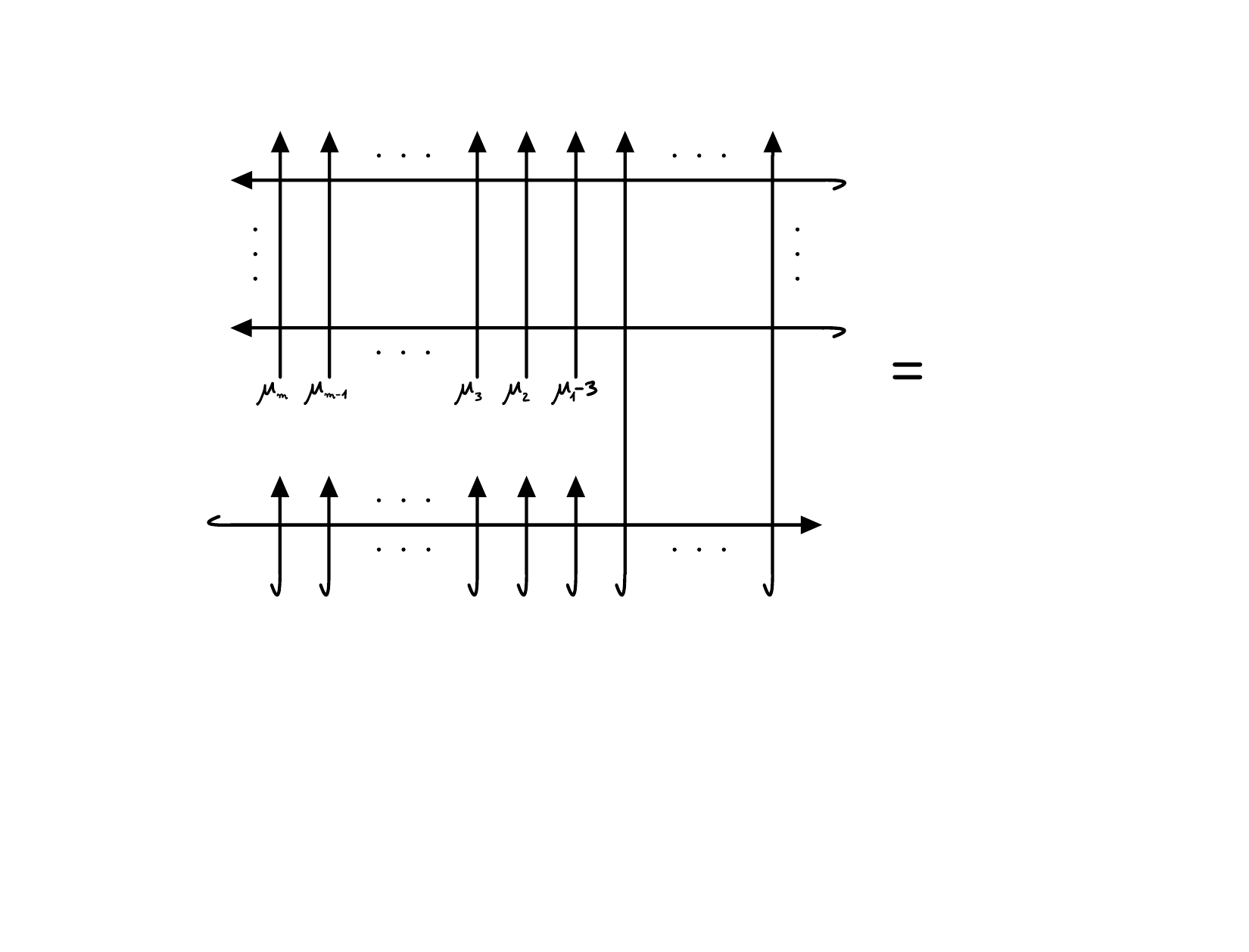}
		\end{minipage}
		\hspace{.1\linewidth}
		\begin{minipage}[b]{.4\linewidth} 
			\includegraphics[scale = .45, trim = 7cm 3.75cm 0cm 2.75cm]{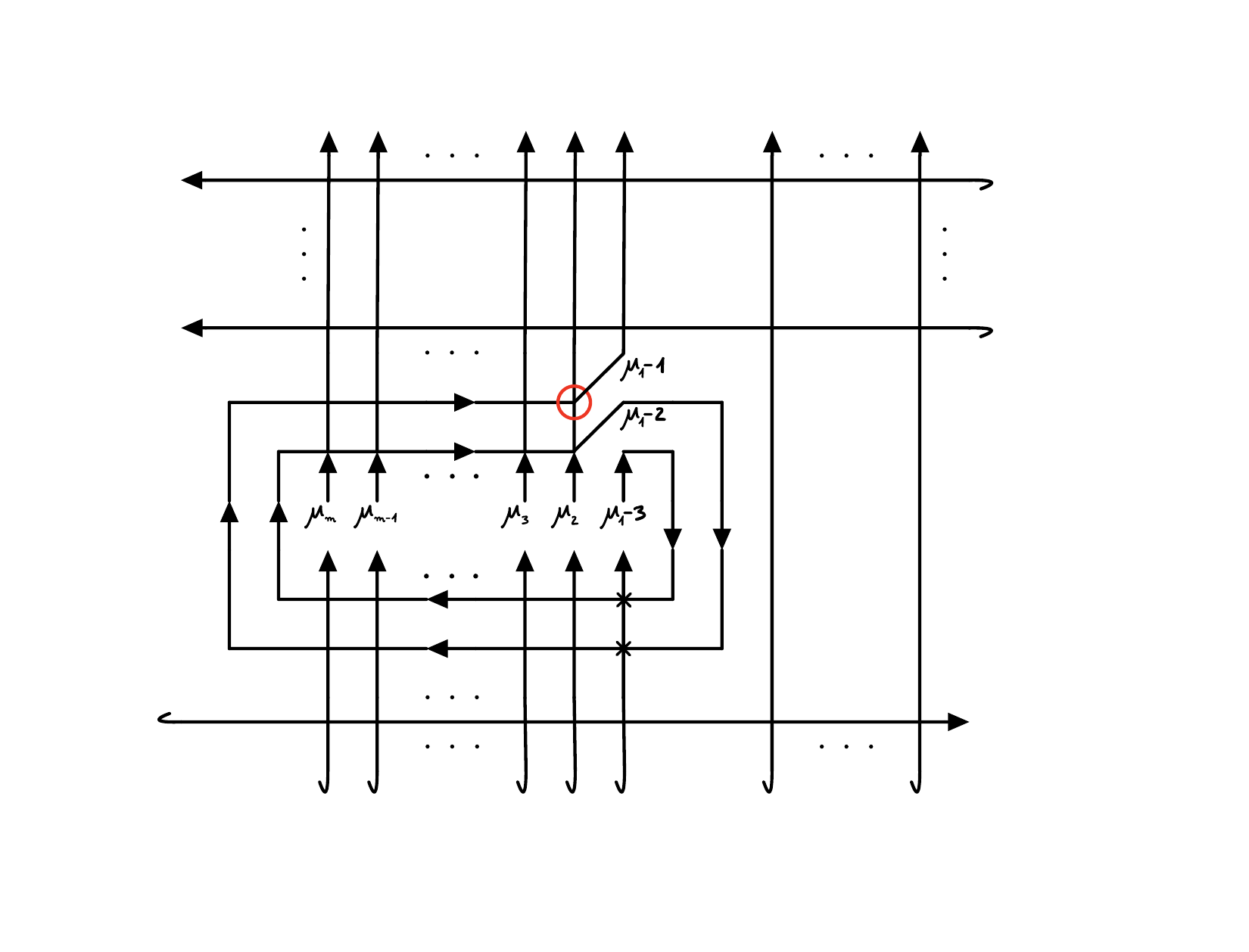}
		\end{minipage}
		\caption{}
		\label{fig:rqKZ_sl2_2_loops}
	\end{figure}
	Taking the residue at $\mu_1 = \mu_2$, $R_{12}(\mu_1-\mu_2-1)$ in the red circle (figure \ref{fig:rqKZ_sl2_2_loops}) reduces to $2P^-_{12}$ up to a scalar prefactor. As a consequence, we can apply the relation (\ref{proj_red_rel_sl2}) in remark \ref{rem:rel7sl2} to obtain the result in figure \ref{fig:Snail_with_two_loops_sl2} \footnote{To be precise, figure \ref{fig:Snail_with_two_loops_sl2} has to be understood as the limit $\mu_1\to\mu_2-1$ of $(\mu_1-\mu_2+1)$ times figure \ref{fig:rqKZ_sl2_2_loops}.}, where we have split the operator $2P^-_{12}$ into the tensor product of a singlet in $V_1\otimes V_2$ (a cross with two ingoing lines) and its dual in $V_1^\star\otimes V_2^\star$ (a cross with two outgoing lines) respectively \footnote{ The $\mathfrak{sl}_2$ case of the definition in section \ref{subsect:graphical_notation}}. The operator in the box with the dashed red line (multiplied by the scalar prefactor) is the Snail Operator for $\mathfrak{sl}_2$ and $k=2$ loops.\newpage
	\begin{figure}[H]
		\centering
		\includegraphics[scale = .5, trim = 5.75cm 3.5cm 5cm 3cm]{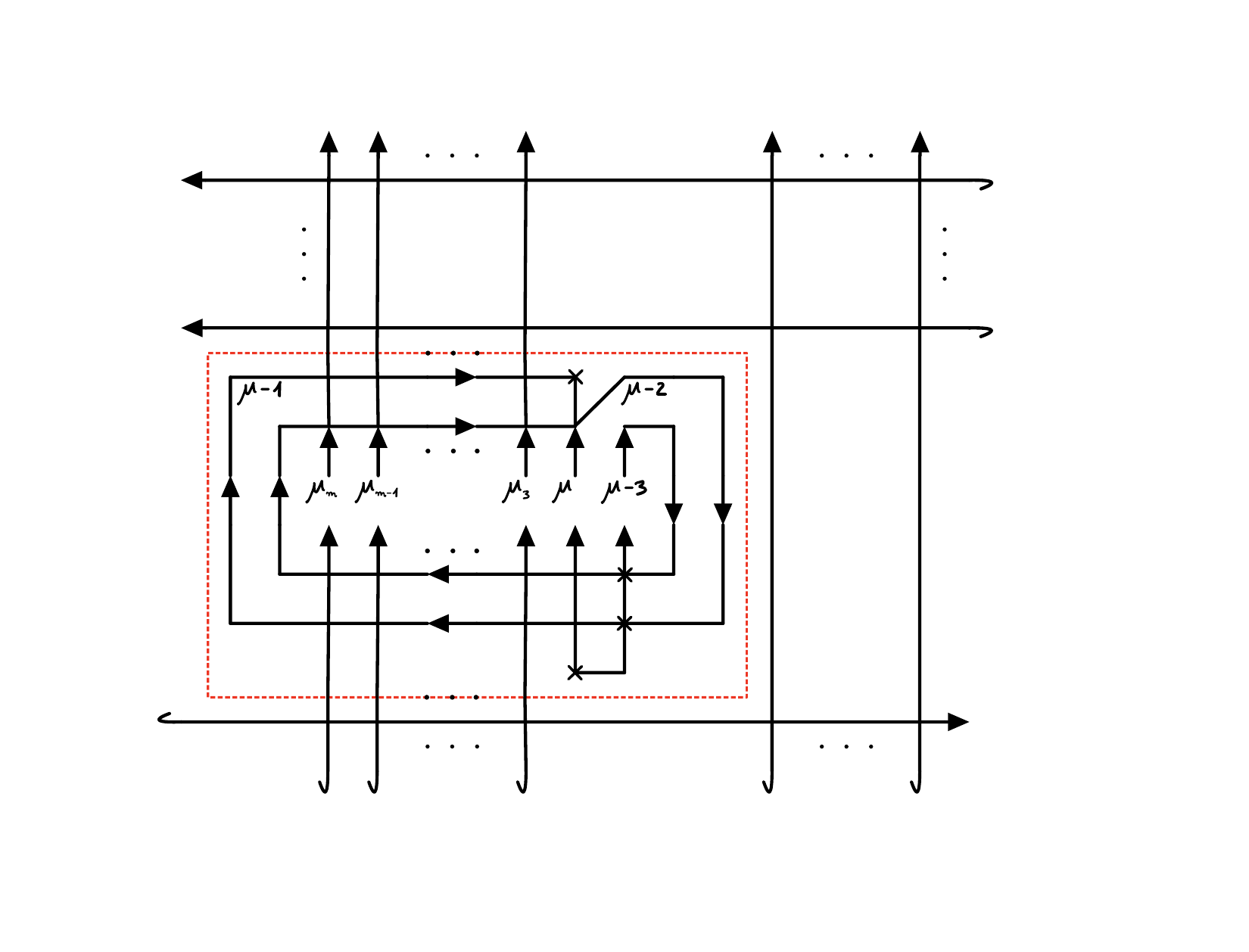}
		\caption{The Snail Operator with two loops ($k=2$).}
		\label{fig:Snail_with_two_loops_sl2}
	\end{figure}
	Due to the fact that we have projectors in the last two loops and the spectral parameters of successive lines differ by exactly one, the Snail Operator can be further simplified by applying identities similar to figure \ref{fig:fusion_rel_sl2}. Note that we omit the prefactor of the $R$-matrix in this figure, i.e. taking the numerical $R$-matrix $r(\lambda) := \lambda+P$ instead of $R(\lambda)$. The arguments of the $r$'s are written next to the vertices.
	\begin{figure}[H]
		\centering
		\includegraphics[scale = .6, trim = 5.75cm 9cm 5cm 4.5cm]{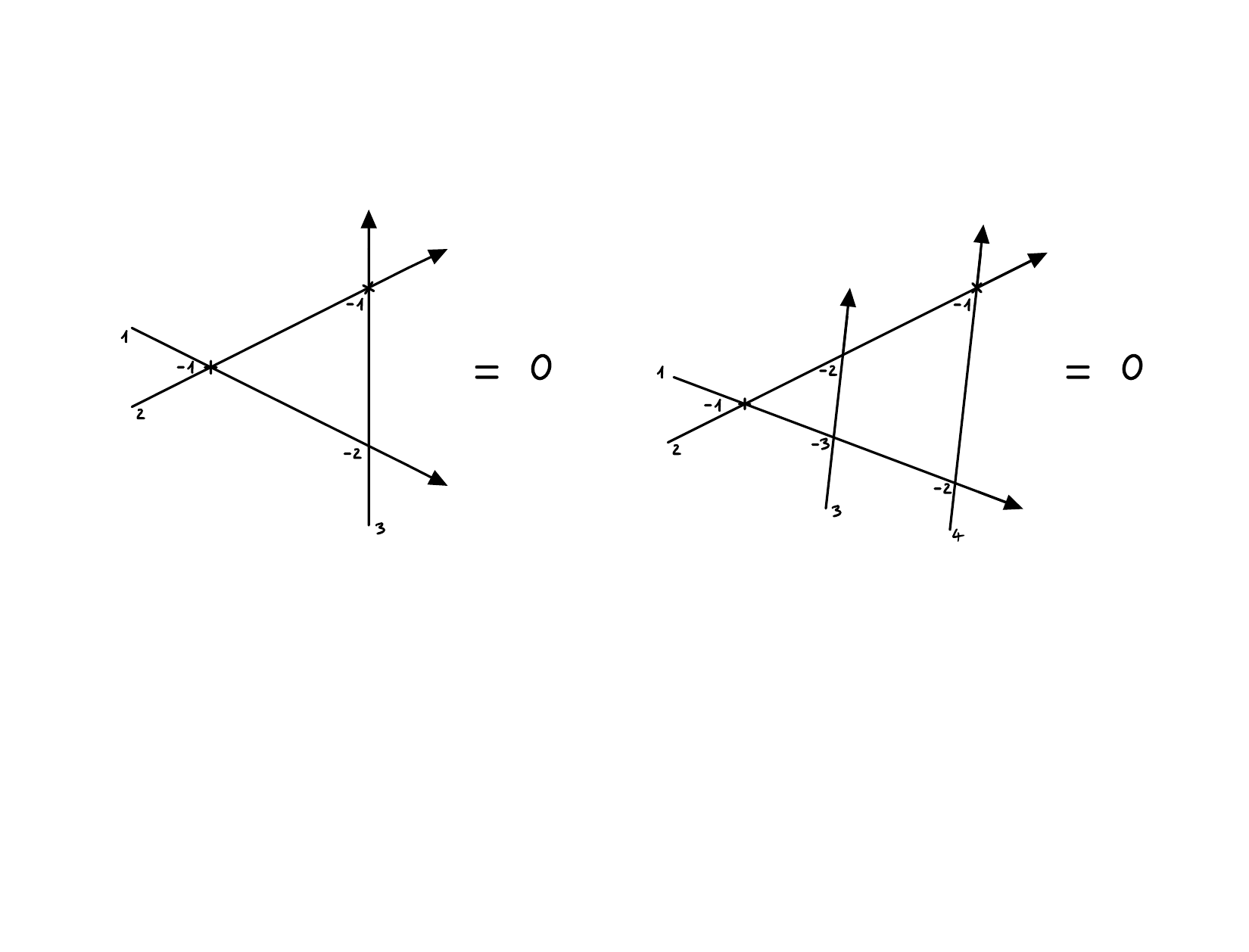}
		\caption{}
		\label{fig:fusion_rel_sl2}
	\end{figure}
	In general, in \cite{BJMST} it was shown that the $k$ loops of the Snail Operator fuse to a single irreducible representation of the Yangian, the Kirillov--Reshetikhin module $W_k$. As a representation of $\mathfrak{sl}_2$ it is just the irreducible spin-$k/2$ representation. This was proven in a combinatoric way. As we are going to investigate the generalization for higher rank, we shall explain the algebraic structure behind this in the next subsection. However, the relation (\ref{proj_red_rel_sl2}) in remark \ref{rem:rel7sl2} is the key identity which was used to reduce $D_m$ to $D_{m-2}$ such that it can be pulled out of the residue. Therefore, let us try to understand it graphically. After applying the rqKZ equation once, we use the fact that $R(0) = P$ (figure \ref{fig:proj_red_rel_sl2_01} inside the red circles) and use the YBE to pull the second line out to the left (long green arrows in figure \ref{fig:proj_red_rel_sl2_01}). After that, we split up the two operators $2P^-$ (the two crosses in figure \ref{fig:proj_red_rel_sl2_01}) into the tensor product of a singlet and its dual. The result is shown in figure \ref{fig:proj_red_rel_sl2_02} on the left. Finally, the singlet and its dual on the straight lines cancel out, as they are now considered as mappings from $V^\star$ to $V$ and $V$ to $V^\star$, respectively. \footnote{ The singlet is considered as a mapping from $^*V$ to $V$ in this case. In general it relates either $^*V$ and $V$ or $V^*$ and $V$. Similarly for its dual.} Of course, they correspond to the different ways of interpreting the charge conjugation operator $C$, having the isomorphism $V\cong V^*$ for $\mathfrak{sl}_2$ in mind. Using the left-right reduction property (\ref{left-right_sl2}) in proposition \ref{prop:properties_of_D_sl2}, we obtain the result on the right in figure \ref{fig:proj_red_rel_sl2_02}.
	\begin{figure}[H]
		\centering
		\begin{minipage}[b]{.44\linewidth} 
			\includegraphics[scale = .475, trim = 5cm 6.5cm 4cm 3.5cm]{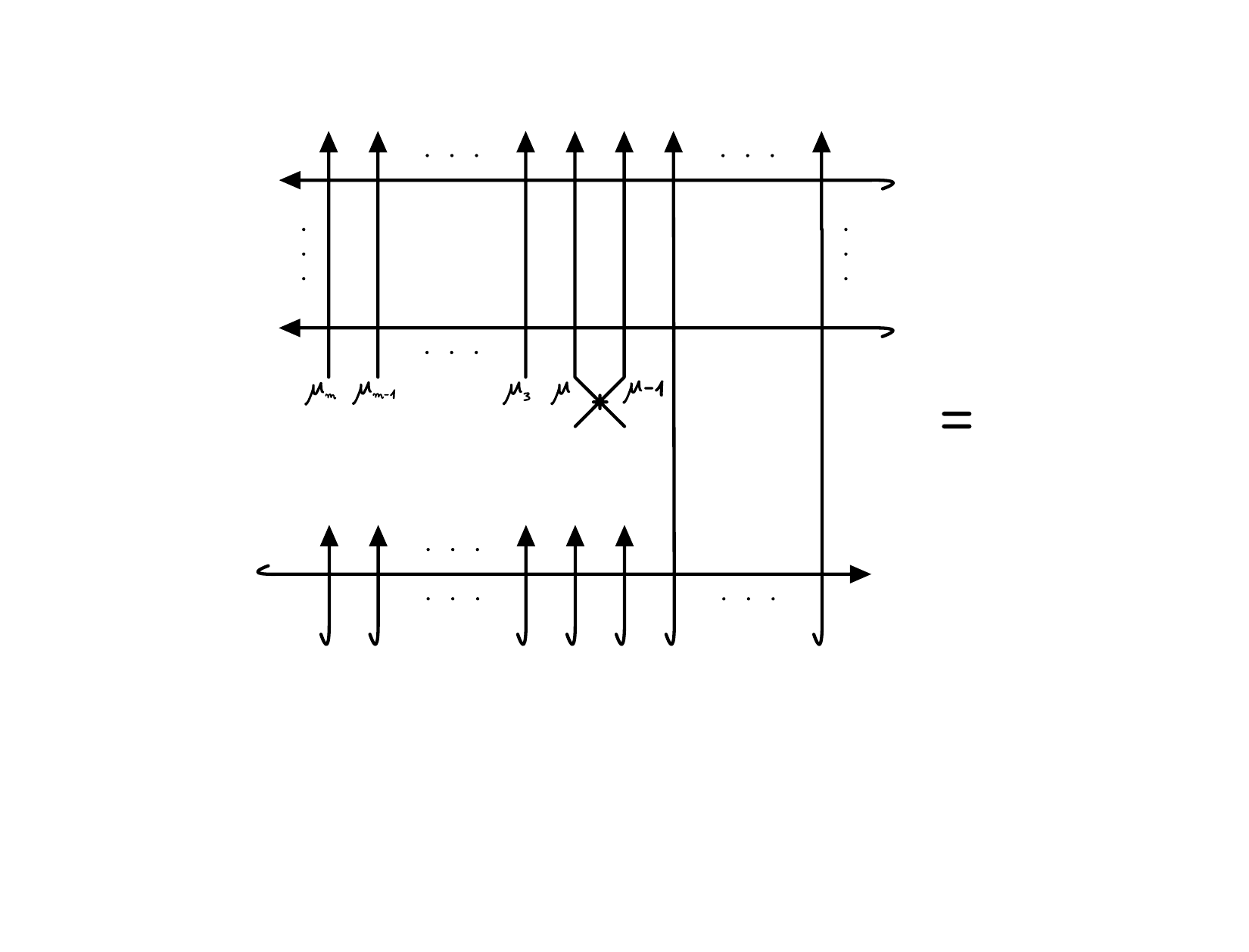}
		\end{minipage}
		\hspace{.1\linewidth}
		\begin{minipage}[b]{.44\linewidth} 
			\includegraphics[scale = .475, trim = 4cm 7cm 4cm 1.75cm]{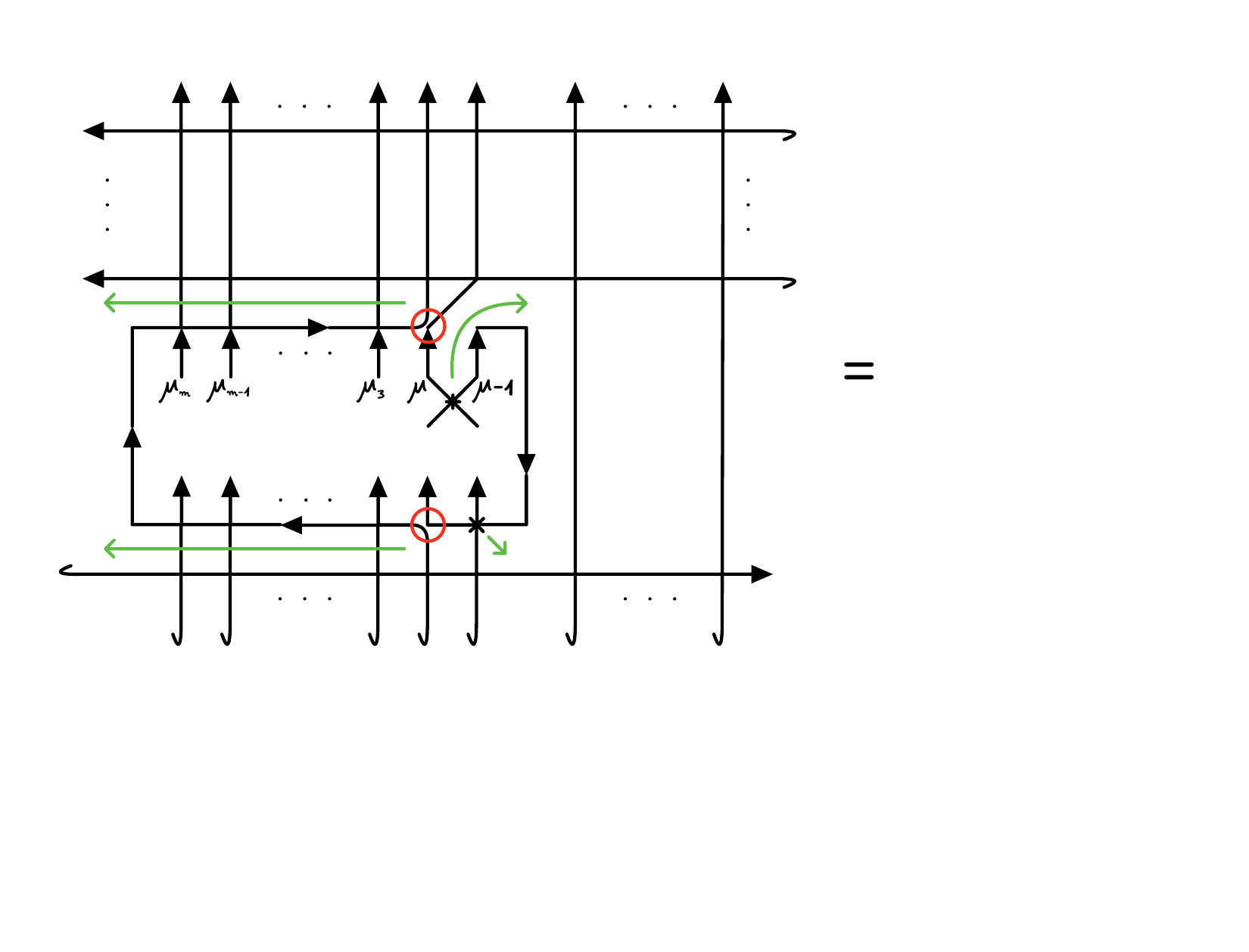}
		\end{minipage}
		\caption{}
		\label{fig:proj_red_rel_sl2_01}
	\end{figure}
	\begin{figure}[H]
		\centering
		\begin{minipage}[b]{.44\linewidth} 
			\includegraphics[scale = .475, trim = 1cm 7cm 3cm 3cm]{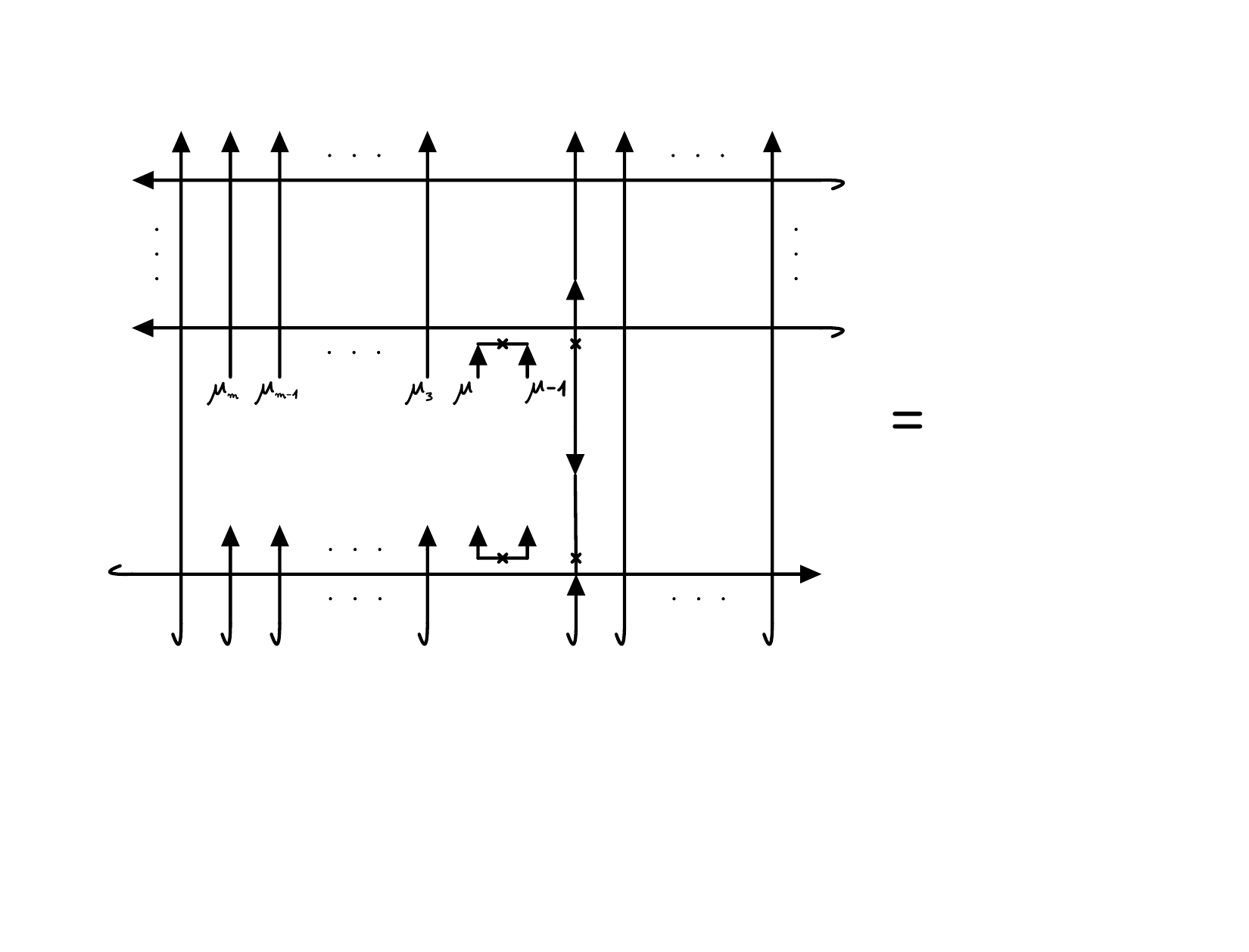}
		\end{minipage}
		\hspace{.1\linewidth}
		\begin{minipage}[b]{.44\linewidth} 
			\includegraphics[scale = .475, trim = 4cm 6cm 4cm 5cm]{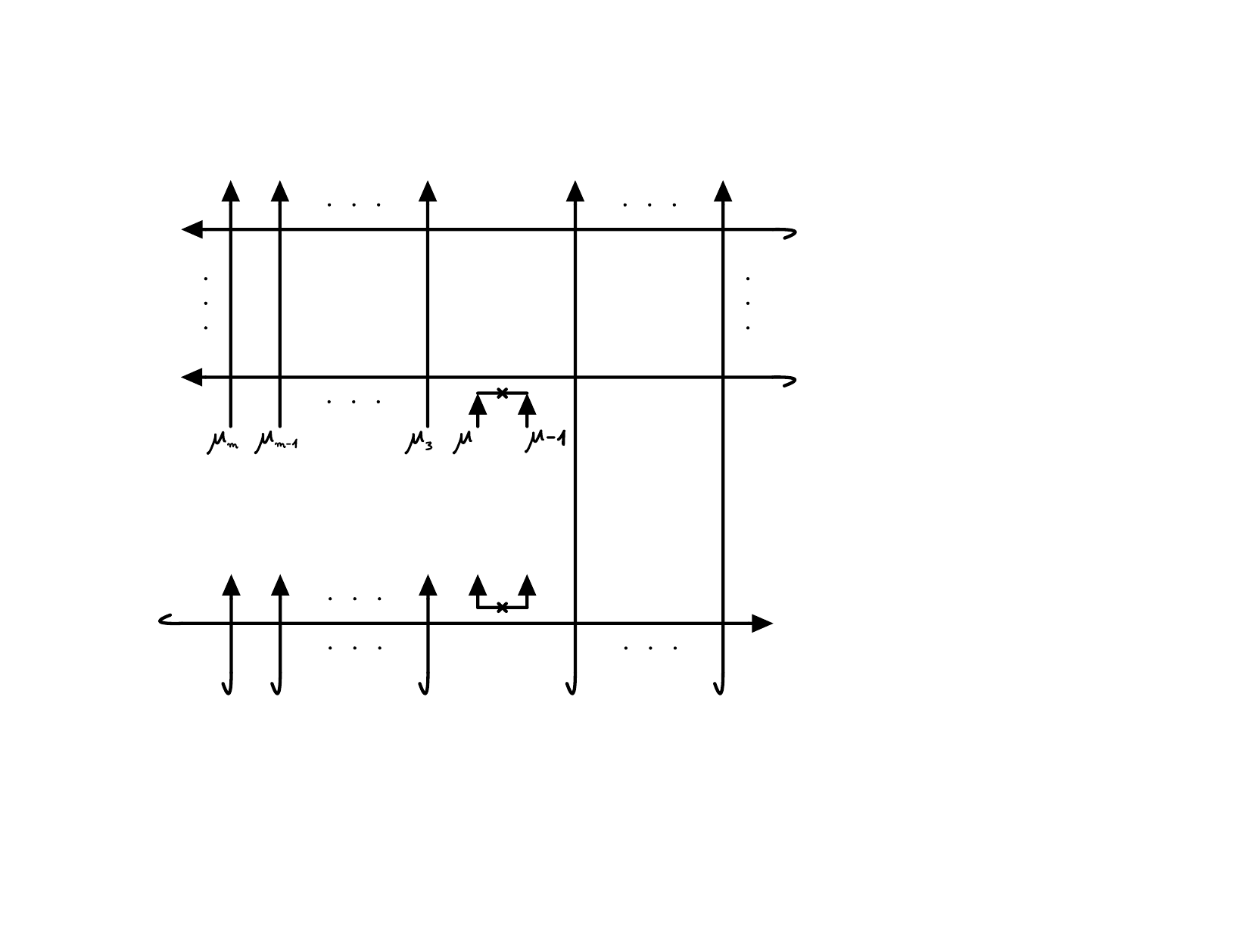}
		\end{minipage}
		\caption{}
		\label{fig:proj_red_rel_sl2_02}
	\end{figure}
	\subsection{T-systems and the Snail Operator $\tilde{X}_k$}
	\label{subsect:T-systems}
	Let us now turn back to the discussion of the Snail Operator in the special case of $\mathfrak{sl}_2$. As the general case will be described in section \ref{subsect:extT-systems} in complete detail, we intend to provide a first explanation of the quite surprising simplification that was proven in \cite{BJMST}. It can be stated in the following way: "In the tensor product of the $k$ fundamental representations that appear in the Snail Operator with $k$ loops, all but one irreducible representation cancel out. The Kirillov--Reshetikhin module $W_k$, which is the spin-$k/2$ irreducible representation of $\mathfrak{sl}_{2}\lhook\joinrel\xrightarrow{\iota}Y(\mathfrak{sl}_{2})$". As was discussed in subsection \ref{subsect:graphical_notation} where we introduced our graphical notation, every line can be regarded a fundamental representation of the Yangian $Y(\mathfrak{sl}_2)$. Drawing the Snail Operator in a slightly less compact way by not splitting up the projector $P^-$, we can see that it has $k$ closed loops (figure \ref{fig:Snail_Operator_sl2}). Again, since the $R$-matrix $R(\lambda)$ has a simple pole at $\lambda=1$, figure \ref{fig:Snail_Operator_sl2} is only understood in terms of the residue at $\mu_1=\mu_2-k-1$. As above, we multiply it by the scalar prefactor obtained from $R(\lambda)$ in the limit $\lambda\to-1$. Anyway, we explain how figure \ref{fig:Snail_Operator_sl2} can be made into a precise definition in a moment.
	\begin{figure}[H]
		\centering
		\begin{minipage}[b]{.44\linewidth} 
			\includegraphics[scale = .5, trim = 5cm 5cm 3cm 4cm]{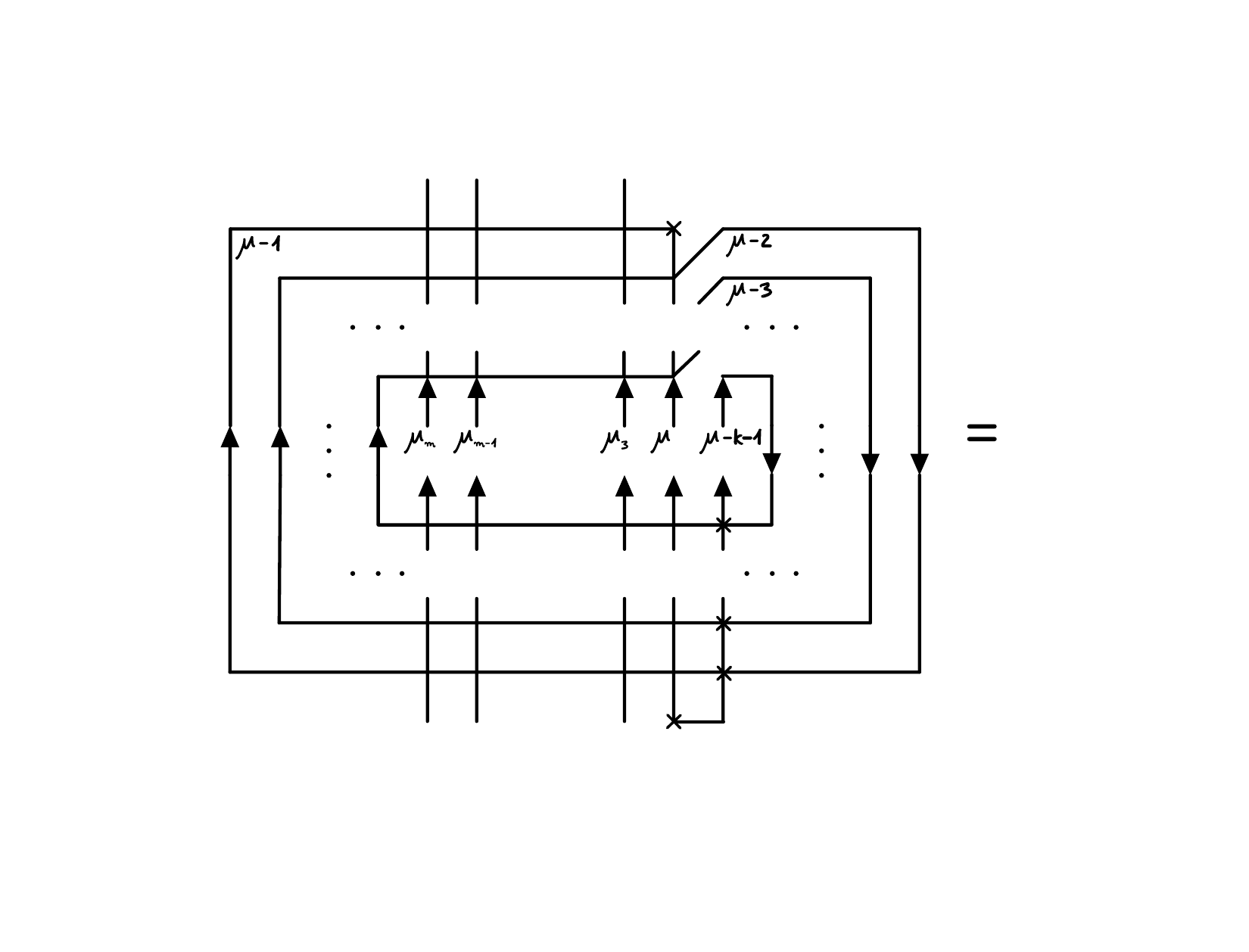}
		\end{minipage}
		\hspace{.1\linewidth}
		\begin{minipage}[b]{.44\linewidth} 
			\includegraphics[scale = .5, trim = 6.5cm 5cm 3cm 6cm]{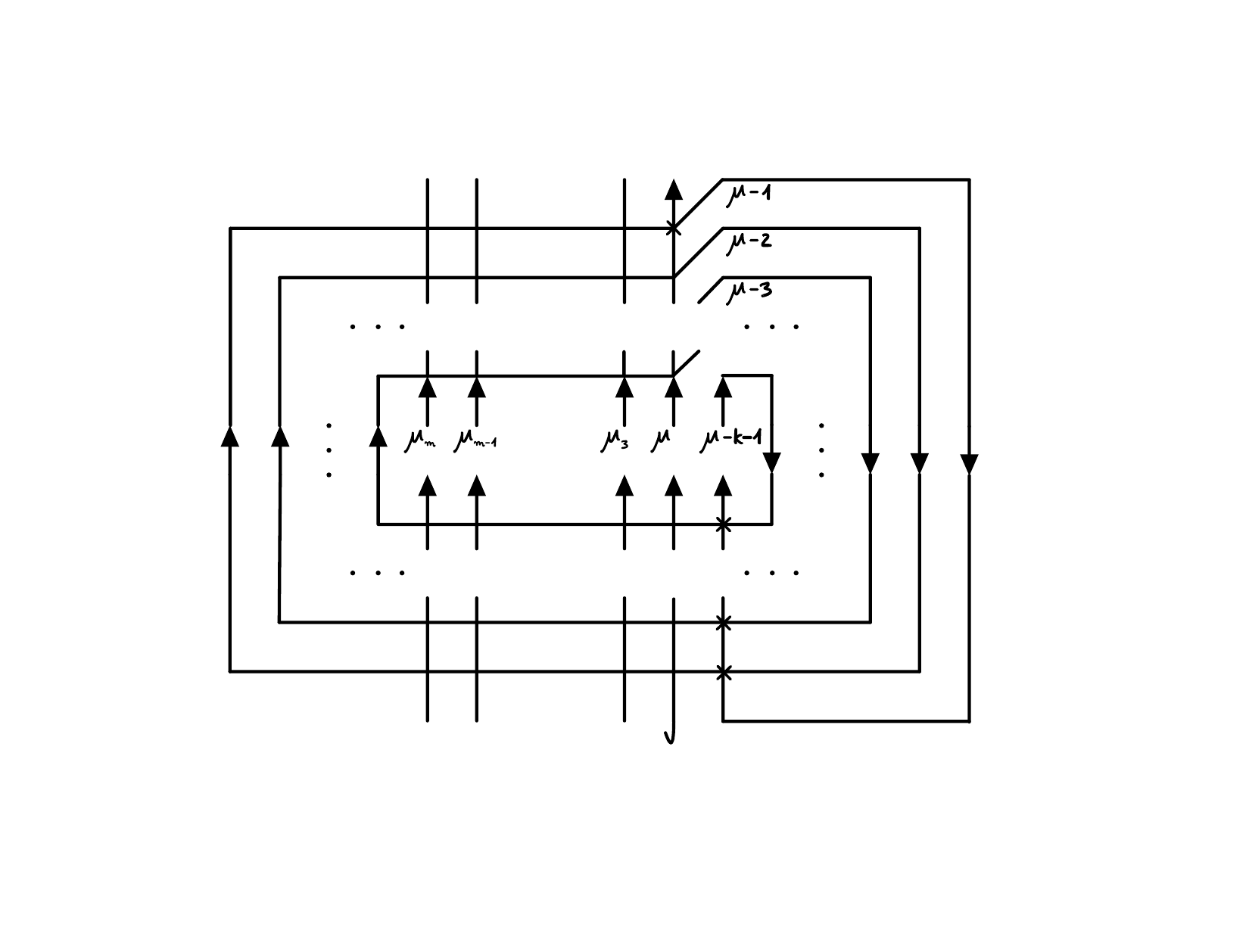}
		\end{minipage}
		\caption{The Snail Operator with $k$ (closed) loops.}
		\label{fig:Snail_Operator_sl2}
	\end{figure}
	The fundamental representations of these lines are obtained by pulling back the spin-$1/2$ fundamental representation $V=V^{(1)}$ of $\mathfrak{sl}_2$ with $\operatorname{ev}_{\mu-l}$, $l=1,\dots,k$. It is called the spin-$1/2$ evaluation representation $V^{(1)}(\mu-l)$ to the loop parameter $a=\mu-l\in \mathbb{C}$. Generally, we define the spin-$k/2$ evaluation representation $V^{(k)}(a)$ of $Y(\mathfrak{sl}_2)$ to the loop parameter $a\in\mathbb{C}$ as the pullback of $V^{(k)}$ by $\operatorname{ev}_a$ ($V^{(k)}(a):=\operatorname{ev}_a^*V^{(k)}$), where $V^{(k)}$ is the spin-$k/2$ irreducible representation of $\mathfrak{sl}_2$. It is closely related to the definition of the Kirillov--Reshetikhin module $W^{(k)}(a)$. We have $W^{(k)}(a):=V^{(k)}(a+\frac{1}{2}(k-1))$. Therefore $V=V^{(1)}(\mu-l)=W^{(1)}(\mu-l)$, $l=1,\dots,k$, are the fundamental representations of the successive lines. Looking at the right side of figure \ref{fig:Snail_Operator_sl2}, we use the identity $A_1 = \tr_{V_\alpha}(A_\alpha P_{\alpha,1}), \, A\in \End(V_1)$, to write the Snail Operator as the residue at $\mu = \mu_2$ of the product of the monodromy matrices
	\begin{align*}
		&\mathcal{T}_{\alpha_l;2,\dots,2m-1}(\mu-l;\mu_2,\dots,\mu_m,\mu_m,\dots,\mu_2):=\\
		&\tr_a\{\overline{T}_{a;2,\dots,m}(\mu-l;\mu_2,\dots,\mu_m)T_{a;m+1,\dots,2m-1}(\mu-l;\mu_m,\dots,\mu_2)P_{a,\alpha_l}\} \stackrel{P_{a,\alpha_l}=R_{a,\alpha_l}(0)}{=}\\
		& \tr_a\{R_{2,a}(\mu_2-\mu+l)R_{3,a}(\mu_{3}-\mu+l)\cdots R_{m,a}(\mu_m-\mu+l)\\
		&R_{a,m+1}(\mu-l-\mu_m)R_{a,m+2}(\mu-l-\mu_{m-1})\cdots R_{a,2m-1}(\mu-l-\mu_2)R_{a,\alpha_l}(0)\}\quad l=1,\dots,k,
	\end{align*}
	multiplied by the operator
	\begin{align*}
		\mathfrak{P}_{2,\alpha_1,\dots,\alpha_k,1}:= 2^k P^-_{\alpha_1,2}P^-_{\alpha_{2},\alpha_1}\cdots P^-_{\alpha_{k},\alpha_{k-1}}P^-_{1,\alpha_k}
	\end{align*}
	and contracted over the spaces $2,\alpha_1,\dots,\alpha_k$. Note that we used the projector identity $(P^-_{\alpha_1,2})^2= P^-_{\alpha_1,2}$ to be able to introduce the operator $\mathfrak{P}$. In fact, the operator $\mathfrak{P}$ turns out to be the projector onto the Kirillov--Reshetikhin module $W^{(k)}(\mu)$ in the tensor product of the spaces $\alpha_1,\dots,\alpha_k$ times the singlet in the tensor product of two spin $1$ representations built from the spaces $V_1\otimes\overline{V}_1$ and $V_2\otimes\overline{V}_2$. Where we identify $V_1\otimes \overline{V}_1$ with $V_1\otimes V_1^\star \cong \End(V_1)$ using the dual of the singlet in $V_1^\star\otimes (\overline{V}_1)^\star$ and similarly for $V_2$.\\
	
	However, let us explain the representation theory behind this. In the category of finite dimensional representations of the Yangian $Y(\mathfrak{sl}_2)$ we consider tensor products of the fundamental (evaluation) representations. Now, the main observation is due to Chari and Pressley (1991) \cite{CP1991} proposition 4.9 for $U_q(\widetilde{\mathfrak{sl}}_2)$, which can be translated for $Y(\mathfrak{sl}_2)$ using the correspondence between finite dimensional representation given in the paper \cite{GT}. The statement is as follows.
	\begin{prop}[special position]
		\label{prop:special_position}
		The tensor product $V:=W^{(k)}(a)\otimes W^{(l)}(b)$ of Kirillov--Reshetikhin modules has a unique proper subrepresentation $W$ iff there is a $0<p\leq\operatorname{min}\{k,l\}$ such that $b-a=\frac{k-l}{2}\pm\left(\frac{k+l}{2}-p+1\right)$. In this case the modules $W^{(k)}(a)$ and $W^{(l)}(b)$ are said to be in \textbf{special position}. We have the short exact sequence
		\begin{align*}
			W\lhook\joinrel\rightarrow W^{(k)}(a)\otimes W^{(l)}(b)\rightarrowdbl V/W,
		\end{align*}
		where the composition factors $W$ and $V/W$ are irreducible. They are given as follows.
		\begin{enumerate}
			\item If $b-a=k-p+1$, we have
			\begin{align}
				W&\cong W^{(k-p)}(a)\otimes W^{(l-p)}(b+p),\notag\\
				V/W&\cong W^{(p-1)}(a+k-p+1)\otimes W^{(k+l-p+1)}(b-k+p-1).
			\end{align}
			As a representation of $\mathfrak{sl}_2$,
			\begin{align*}
				W\cong V^{(k+l-2p)}\oplus V^{(k+l-2p-2)}\oplus\dots\oplus V^{(|m-n|)}.
			\end{align*}
			\item If $b-a=-l+p-1$, we have
			\begin{align}
				W&\cong W^{(p-1)}(a-k)\otimes W^{(k+l-p+1)}(b),\notag\\
				V/W&\cong W^{(k-p)}(a+p)\otimes W^{(l-p)}(b).
			\end{align}
			As a representation of $\mathfrak{sl}_2$,
			\begin{align*}
				W\cong V^{(k+l)}\oplus V^{(k+l-2)}\oplus\dots\oplus V^{(m+n-2p+2)}.
			\end{align*}
		\end{enumerate}
		$\odot$
	\end{prop}
	Using this proposition, we see that neighbouring lines in the Snail Operator are in special position with respect to each other. As we can forget about the spectral parameter of any trivial representation $\mathbb{C}\cong W^{(0)}(a)=:W^{(0)}$, we can write the short exact sequence between two successive lines as
	\begin{align*}
		W^{(0)}\lhook\joinrel\rightarrow W^{(1)}(\mu-1)\otimes W^{(1)}(\mu) \rightarrowdbl W^{(2)}(\mu-1).
	\end{align*}
	Considering a partition of unity with the respective projectors onto the composition factors in the Snail Operator, the projector onto $W^{(0)}$ cancels out. It can be easily checked by using the identities in figure \ref{fig:fusion_rel_sl2}. Now, writing only the irreducible composition factors of the possible short exact sequences in proposition \ref{prop:special_position}, we get equations in the Grothendieck ring, the \textit{T-systems}. Usually it is referred to the case when $k=l$, where we have \textit{the T-system} \cite{KNS}
	\begin{align}
		[W^{(k)}(\mu-1)][W^{(k)}(\mu)] = [W^{(k+1)}(\mu-1)][W^{(k-1)}(\mu)]+1,
		\label{eqn:tsys_1}
	\end{align}
	which is oftentimes written in terms of transfer matrices with the respective representations in the auxiliary space \cite{HL}. Using this equation, one can derive the T-system \footnote{ Equivalently, we can choose $l=1$ in proposition \ref{prop:special_position}.}
	\begin{align}
		[W^{(1)}(\mu-k)][W^{(k)}(\mu-k+1)] = [W^{(k+1)}(\mu-k)]+[W^{(k-1)}(\mu-k+2)],
		\label{eqn:tsys_2}
	\end{align}
	which appears in the Snail Operator successively. In fact, it was proven that the second component cancels out in every step, similar to the case $k=1$ above \cite{BJMST}. Therefore, only the Kirillov--Reshetikhin module $W^{(k)}$ remains in the Snail Operator with $k$ loops. Furthermore, it is possible to analytically continue the definition for any $k\in\mathbb{C}$. This is done in the paper \cite{BJMST} by defining a trace function
	\begin{align*}
		\operatorname{Tr}_x : U(\mathfrak{sl}_2)\otimes\mathbb{C}[x]\to \mathbb{C}[x]
	\end{align*}
	such that for any non negative integer $k$ we have
	\begin{align*}
		\operatorname{Tr}_{k+1}(A)=\operatorname{tr}_{V^{(k)}}\pi^{(k)}(A)\quad(A\in U(\mathfrak{sl}_2)).
	\end{align*}
	The analytical continuation is then defined roughly by replacing the $R$-matrices in the definition above by $L$-operators and applying the trace function. The exact details are described in the paper \cite{BJMST}. Since it is defined through a separate algebraic construction, we stop our review here and comment on it later when we discuss the higher rank case. However, the generalization that we present is not of Kirillov--Reshetikhin type for higher rank. In fact, we will see that it is a certain minimal snake module \cite{MY}.
	\section{The construction for $\mathfrak{sl}_3$}
	\label{sect:sl3snailconstr}
	Now after we explained the construction in the $\mathfrak{sl}_2$ (rank 1) case, we are in position to discuss our ansatz for the generalization to higher rank. For the basic construction we focus especially on the $\mathfrak{sl}_3$ (rank 2) case, as we can use the results of the paper \cite{BHN} to explain how some residues can already be calculated. However, the representation theoretical explanation in the second part of this section (subsection \ref{subsect:extT-systems}) works for any rank as we shall see. Let us start with the generalization of the properties of $D$.
	\begin{ppt}
		Analogous to the $\mathfrak{sl}_2$ case $D_{1,...,m}$ fulfills
		\begin{enumerate}
			\item $D_m$ is invariant under the action of $\mathfrak{sl}_{n+1}$. \label{gln+1_inv}
			\item The R-matrix relations \label{R_matrix_rel_sln+1}
			\begin{align*}
				&D_{1,...,i+1,i,...,m}(\lambda_1,...,\lambda_{i+1},\lambda_i,...,\lambda_m) =\\ &R_{i+1,i}(\lambda_{i+1,i})D_{1,...,m}(\lambda_1,...,\lambda_n)R_{i,i+1}(\lambda_{i,i+1}).
			\end{align*}
			\item Left-right reduction relations \label{left-right_sln+1}
			\begin{align*}
				\operatorname{tr}_1(D_{1,...,m}(\lambda_1,...,\lambda_m)) &= D_{2,...,m}(\lambda_2,...,\lambda_m)\\
				\operatorname{tr}_m(D_{1,...,m}(\lambda_1,...,\lambda_m)) &= D_{1,...,m-1}(\lambda_1,...,\lambda_{m-1}).
			\end{align*}
			\item In contrast to the $\mathfrak{sl}_2$ case the fundamental and the antifundamental representation are not isomorphic. This has the consequence that the rqKZ equation splits into two parts\label{rqKZ_sln+1}
			\begin{align}
				&D^{(1)}_{\bar{1}2...m}(\lambda_1-\frac{n+1}{2},\lambda_2,...,\lambda_m) =\notag\\ &A^{(1)}_{1,\bar{1}|2,...,m}(\lambda_1|\lambda_2,...,\lambda_m)\left(D_{1...m}(\lambda_1,\lambda_2,...,\lambda_m)\right):=\notag\\
				&\operatorname{tr}_1(R_{1m}(\lambda_1-\lambda_m)\cdots R_{12}(\lambda_1-\lambda_2)D_{1...m}(\lambda_1,\lambda_2,...,\lambda_m)\notag\\
				&\qquad (n+1)P^{-}_{1,\bar{1}}R_{21}(\lambda_2-\lambda_1)\cdots R_{m1}(\lambda_m-\lambda_1)),\\
				&D_{1...m}(\lambda_1-\frac{n+1}{2},\lambda_2,...,\lambda_m) =\notag\\ &A^{(2)}_{\bar{1},1|2,...,m}(\lambda_1|\lambda_2,...,\lambda_m)(D^{(1)}_{\bar{1},2...m}(\lambda_1,\lambda_2,...,\lambda_m)):=\notag\\
				&\operatorname{tr}_{\bar{1}}(\bar{\bar{R}}_{\bar{1}m}(\lambda_1-\lambda_m)\cdots \bar{\bar{R}}_{\bar{1}2}(\lambda_1-\lambda_2)D^{(1)}_{\bar{1},2...m}(\lambda_1,\lambda_2,...,\lambda_m)\notag\\
				&\qquad (n+1)P^{-}_{1,\bar{1}}\bar{R}_{2\bar{1}}(\lambda_2-\lambda_1)\cdots \bar{R}_{m\bar{1}}(\lambda_m-\lambda_1)).
			\end{align}
			\seti
		\end{enumerate}
		$\odot$
	\end{ppt}
	To keep the notation short we omit the indices $2,\dots,m$ and write $A^{(1)}_{1,\bar{1}}(\lambda_1|\lambda_2,\dots,\lambda_m)$ respectively $A^{(2)}_{\bar{1},1}(\lambda_1|\lambda_2,\dots,\lambda_m)$ from now on.
	$A^{(1)}_{1,\bar{1}}(\lambda_1|\lambda_2,\dots,\lambda_m)$ and $A^{(2)}_{\bar{1},1}(\lambda_1|\lambda_2,\dots,\lambda_m)$ are depicted in figure \ref{fig:A1andA2}, where the cross on the bottom right stands for the operator $(n+1)PP^-$ as explained in section \ref{subsect:graphical_notation}.
	\begin{figure}[H]
		\centering
		\includegraphics[scale = .6, trim = 5.75cm 8.75cm 5cm 7.25cm]{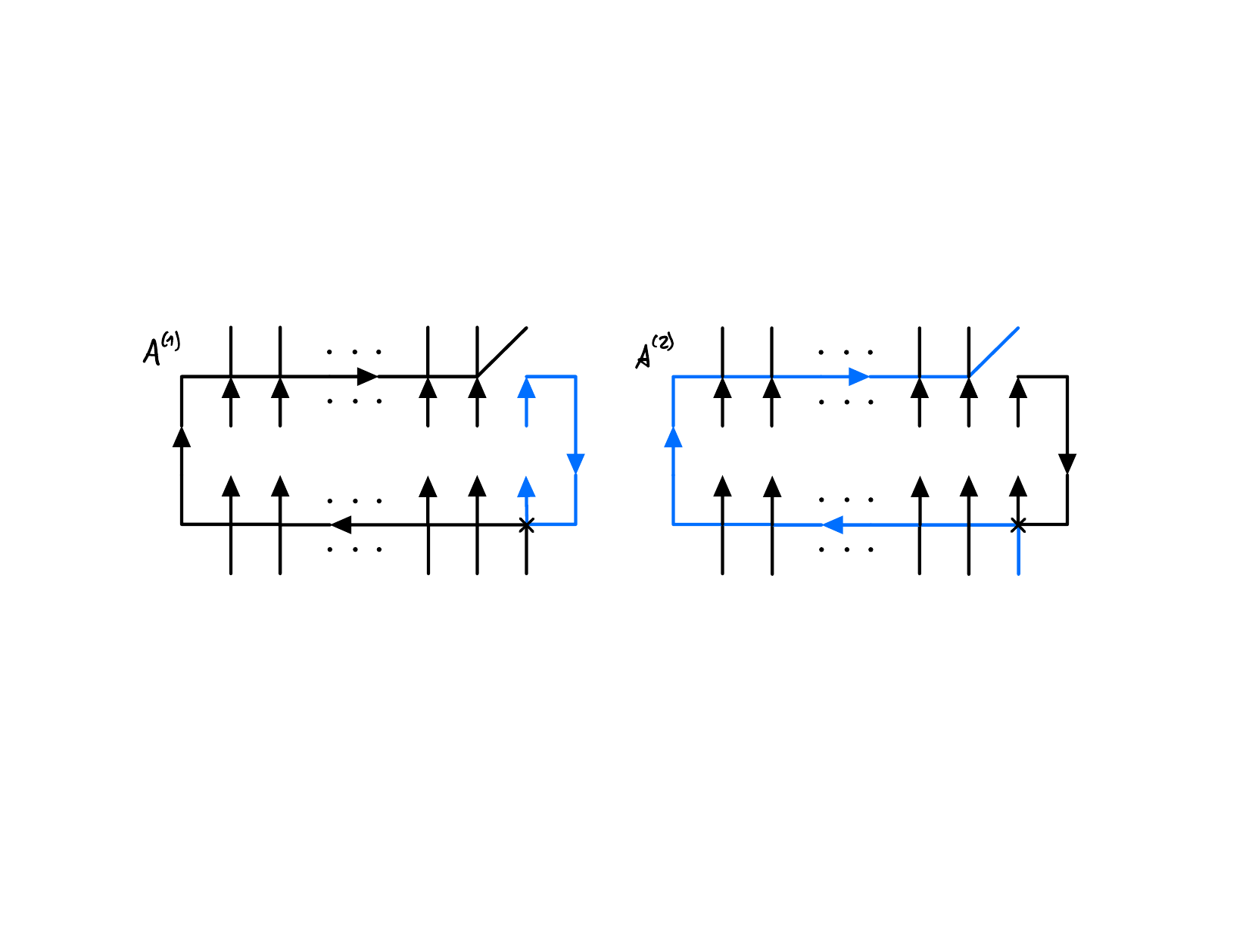}
		\caption{$A^{(1)}_{1,\bar{1}}$ and $A^{(2)}_{\bar{1},1}$.}
		\label{fig:A1andA2}
	\end{figure}
	A complete derivation of the rqKZ equation in the general case is given in the paper \cite{KNR}.
	\begin{rem}\hspace{1em}
		\begin{itemize}
			\item Due to $\rho(\lambda)\rho(-\lambda) = 1$ and $\rho(\lambda)\rho(n+1-\lambda) = -\frac{(n+1-\lambda)\lambda}{(n-\lambda)(1-\lambda)}$ the coefficients in property (\ref{R_matrix_rel_sln+1}) and (\ref{rqKZ_sln+1}) are rational.
			\item $D_{1,\dots,m}(\lambda_1,\dots,\lambda_m)$ is translationally invariant
			\begin{align*}
				D_{1,\dots,m}(\lambda_1+u,\dots,\lambda_m+u) = D_{1,\dots,m}(\lambda_1,\dots,\lambda_m)
			\end{align*}
			\item $D_{1,\dots,m}(\lambda_1,\dots,\lambda_m)$ fulfils the colour conservation rule
			\begin{align*}
				\left[D_{1,\dots,m}(\lambda_1,\dots,\lambda_m)\right]_{\epsilon_1\dots\epsilon_m}^{\bar{\epsilon}_1\dots\bar{\epsilon}_m} = 0\quad \text{if}\quad \exists k\in\{1,\dots,n\} \quad m_k(\epsilon)\neq m_k(\bar{\epsilon})
			\end{align*}
			where the components of $D$ are given by
			\begin{align*}
				\left[D_{1,\dots,m}(\lambda_1,\dots,\lambda_m)\right]_{\epsilon_1\dots\epsilon_m}^{\bar{\epsilon}_1\dots\bar{\epsilon}_m}:=D_{1,\dots,m}(\lambda_1,\dots,\lambda_m)\left((\tensor{E}{_{\epsilon_1}^{\bar{\epsilon}_1}})_1\cdots (\tensor{E}{_{\epsilon_n}^{\bar{\epsilon}_n}})_m\right),
			\end{align*}
			$\tensor{E}{_{\epsilon_i}^{\bar{\epsilon}_i}} = e_{\epsilon_i}\otimes e^{\bar{\epsilon}_i}$ and $m_k(\epsilon)$ is the number of $\epsilon_i$, $i=1,\dots,m$, with $\epsilon_i=k$. $\odot$
		\end{itemize}
	\end{rem}
	Similarly, we have a conjecture for the analytic properties of $D$. They should be proven in the same way as for the $\mathfrak{sl}_2$ case using integral formulas obtained from the vertex operator approach in the massive regime and taking the limit $q\to1$. Looking at the results in the paper \cite{BHN}, we can check that they are satisfied for $\sltr$ in the case of the one, two and three point density matrix. Generally, one could use the results in the papers \cite{KQ} \cite{FK}. In this paper we will assume them to be correct and leave a complete proof open to future work.
	\begin{conj}The analytic properties of $D$ are as follows.
		\label{conj:analyt_prop_sln+1}
		\begin{enumerate}
			\conti
			\item 	$D_{1,...,m}$ is meromorphic in $\lambda_1,...,\lambda_m$ with at most simple poles at $\lambda_i-\lambda_j \in \mathbb{Z}\backslash\{0,\pm 1,\dots,\pm n\}$.
			\label{analyt_sln+1} 
			\item $\forall\, 0<\delta<\pi:$
			\begin{align*}
				\lim\limits_{\substack{\lambda_1\to \infty \\ \lambda_1\in S_\delta}} D_{1,..,m}(\lambda_1,...,\lambda_m) = \frac{1}{n+1}\boldsymbol{1}_1D_{2,...,m}(\lambda_2,...,\lambda_m),
			\end{align*}
			where $S_\delta :=\{\lambda\in\mathbb{C}|\delta<|\arg(\lambda)|<\pi-\delta\}$. \label{asympt_sln+1} $\odot$
			\seti
		\end{enumerate}
	\end{conj}
	Finally, we need an analogue of the relation (\ref{proj_red_rel_sl2}) in Remark \ref{rem:rel7sl2} which was the key identity to be able to decouple the Snail Operator from the density matrix. Fortunately, it naturally generalizes to higher rank as follows. As we only have a singlet in the tensor product of the fundamental and antifundamental representation, we have to start with $D^{(1)}$ in the general case.
	\begin{cor}\hspace{1em}
		\label{cor:rel7_sl3}
		\begin{enumerate}
			\conti
			\item Using again the Properties (\ref{R_matrix_rel_sln+1}), (\ref{left-right_sln+1}), (\ref{rqKZ_sln+1}) and the analyticity of $D_m$ at $\lambda_1=\lambda_2$ we obtain
			\begin{align*}
				P^-_{\bar{1}2} D^{(1)}_{\bar{1},2,\dots,m}(\lambda-\frac{n+1}{2},\lambda,\dots,\lambda_n) = P^-_{\bar{1}2}D_{3,\dots,m}(\lambda_3,\dots,\lambda_n),
			\end{align*}
			where $(P^{-}_{\bar{1}1})^2=P^{-}_{\bar{1}1}$ is the projector onto the singlet. $\odot$
			\label{proj_red_rel_sln+1}
		\end{enumerate}
	\end{cor}
	The derivation is done in the same way as for $\mathfrak{sl}_2$. Let us explain it again anyway. After applying the rqKZ equation once, we use the fact that $R(0) = P$ (figure \ref{fig:proj_red_rel_sln+1_01} inside the red circles) and use the YBE to pull the second line out to the left (long green arrows in figure \ref{fig:proj_red_rel_sln+1_01}). After that, we split up the two operators $2P^-$ (the two crosses in figure \ref{fig:proj_red_rel_sln+1_01}) into the tensor product of a singlet and its dual. The result is shown in figure \ref{fig:proj_red_rel_sln+1_02} on the left. Finally, the singlet and its dual on the straight lines cancel out as they are now considered as mappings from $V^\star$ to $V$ and $V$ to $V^\star$, respectively. \footnote{ To be precise, the singlet is considered a mapping from $^*V$ to $V$ in this case. In general it relates either $^*V$ and $V$ or $V^*$ and $V$. Similarly for its dual.} Of course, they correspond to the different ways of interpreting the charge conjugation operator $C$. Using the left-right reduction relation (property (\ref{left-right_sln+1})), we obtain the result on the right in figure \ref{fig:proj_red_rel_sln+1_02}.
	\begin{figure}[H]
		\centering
		\begin{minipage}[b]{.44\linewidth} 
			\includegraphics[scale = .475, trim = 5cm 6.5cm 4cm 3.5cm]{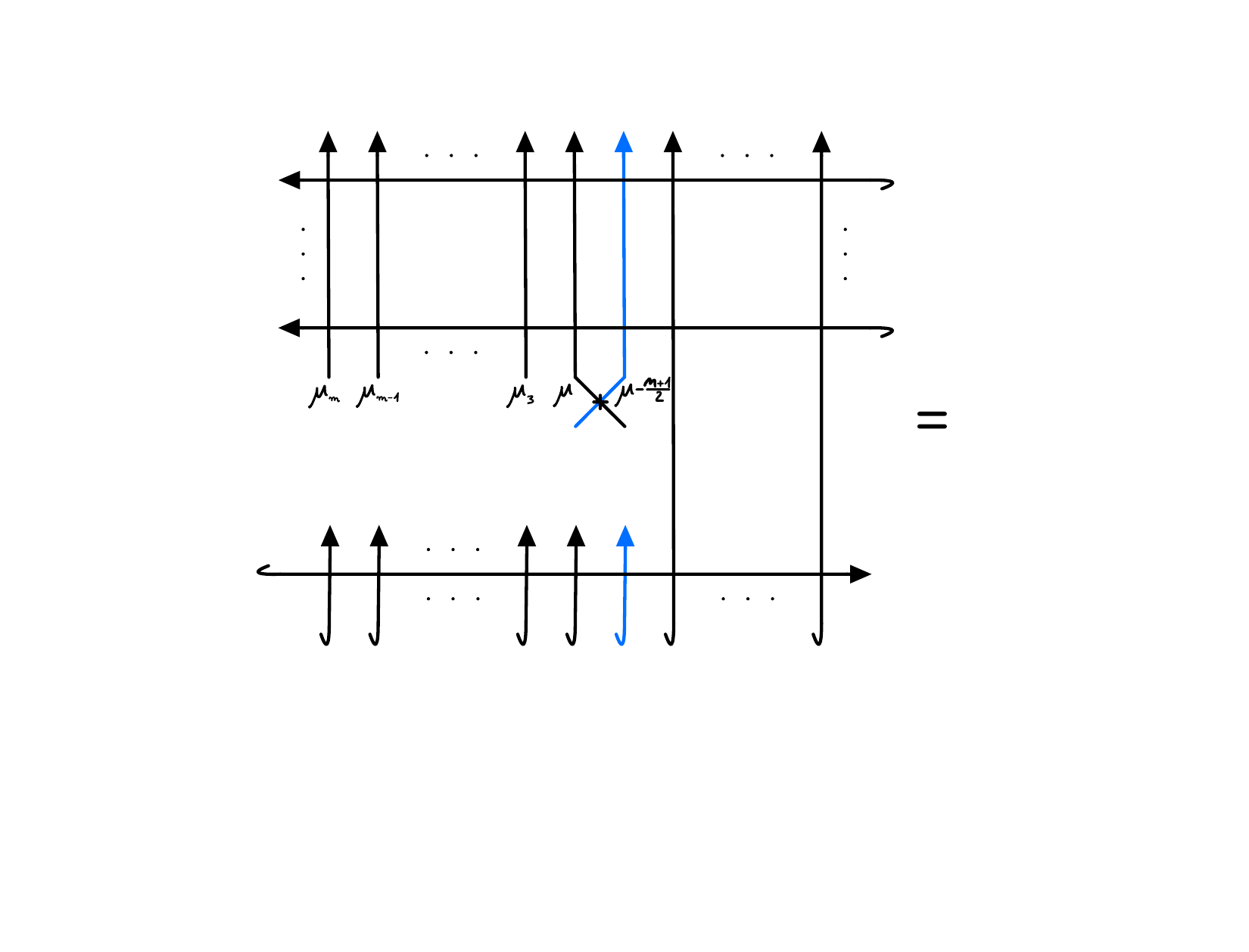}
		\end{minipage}
		\hspace{.1\linewidth}
		\begin{minipage}[b]{.44\linewidth} 
			\includegraphics[scale = .475, trim = 4cm 7cm 4cm 1.75cm]{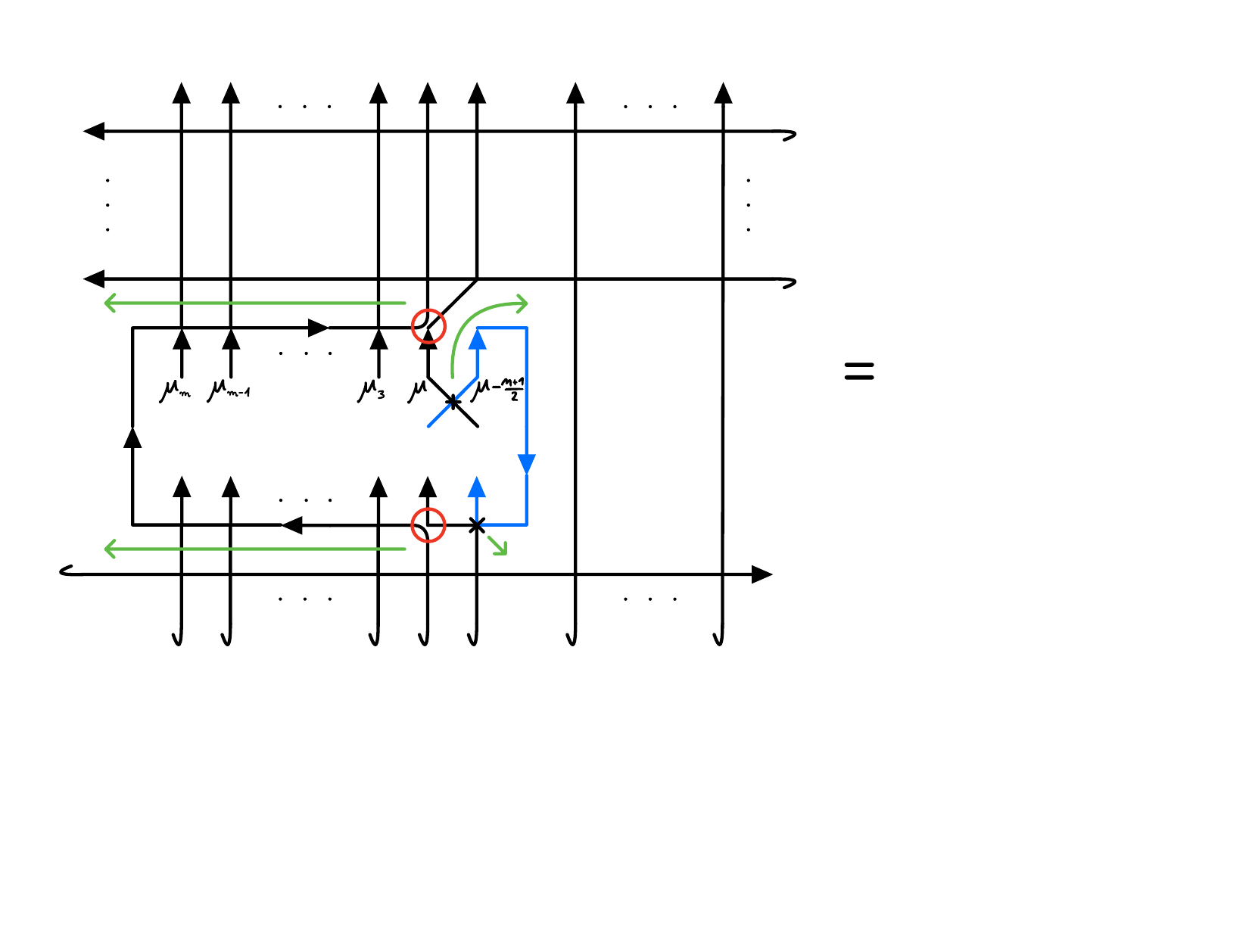}
		\end{minipage}
		\caption{}
		\label{fig:proj_red_rel_sln+1_01}
	\end{figure}
	\begin{figure}[H]
		\centering
		\begin{minipage}[b]{.44\linewidth} 
			\includegraphics[scale = .475, trim = 1cm 7cm 3cm 2.5cm]{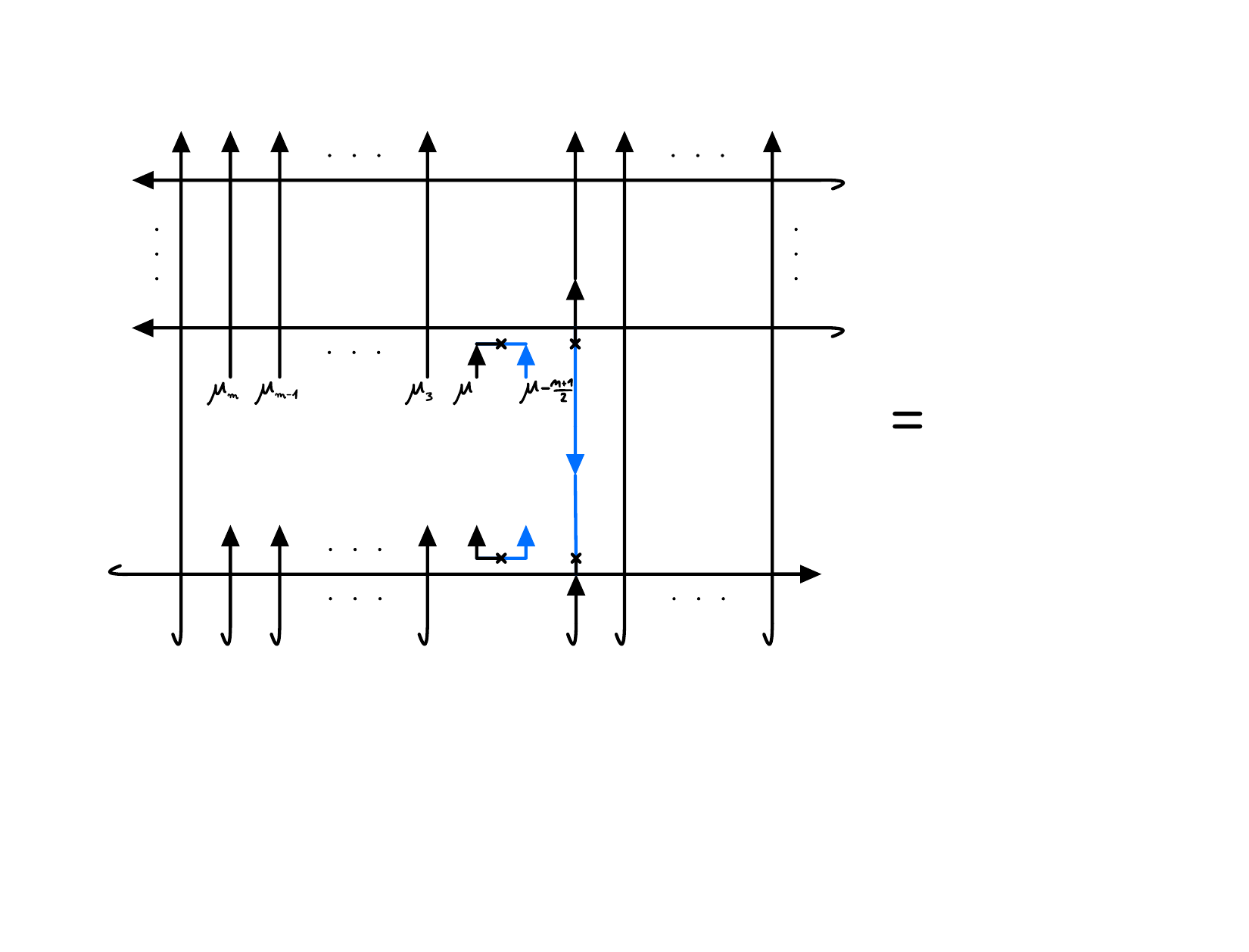}
		\end{minipage}
		\hspace{.1\linewidth}
		\begin{minipage}[b]{.44\linewidth} 
			\includegraphics[scale = .475, trim = 4cm 6cm 4cm 4.5cm]{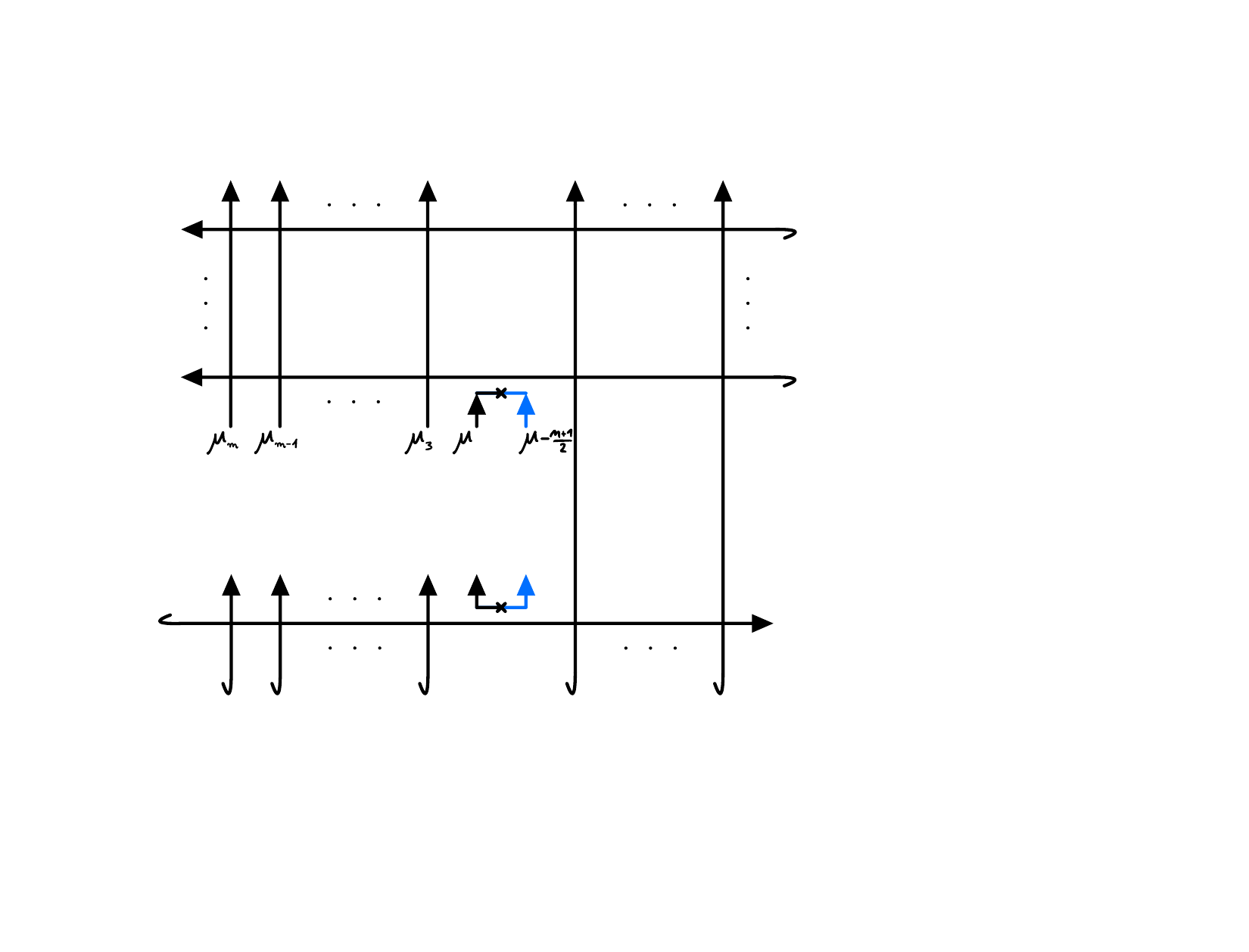}
		\end{minipage}
		\caption{}
		\label{fig:proj_red_rel_sln+1_02}
	\end{figure}
	However, in general we only obtain this reduction property for $D^{(1)}$. Looking at $D$, we can at best hope for relations that include projectors which appear in the tensor product of two fundamental representations of $\mathfrak{sl}_{n+1}$.
	
	For now, let us close the generalization of the properties of $D$ (and $D^{(1)}$). We will come back to this problem in a moment.
	The idea for the generalization of the $\mathfrak{sl}_2$ construction remains the same. Since $D_{1,...,m}$ is meromorphic in $\lambda_1,...,\lambda_m$ with at most simple poles at $\lambda_i-\lambda_j \in \mathbb{Z}\backslash\{0,\pm 1,\dots,\pm n\}$, it is completely determined by its residues and asymptotic behaviour. 
	Using the $R$-matrix relations, we can assume $i=1$, $j=2$ and $k$ positive without loss of generality. 
	In fact, looking at the residue of the operator $A^{(2)}_{\bar{1},1}(\lambda_1|\lambda_2,\dots,\lambda_m)$ at $\lambda_1 = \lambda_2-\frac{n+1}{2}$, we obtain the projector $P^-_{\bar{1}2}$ onto the singlet as before. This means that we can calculate the residues of $D_{1,\dots,m}(\lambda_1,\lambda_2,\dots,\lambda_m)$ at $\lambda_1=\lambda_2-k(n+1)$, $k=1,2,\dots$, in the exact same way as for $\mathfrak{sl}_2$. We explain it again anyway, as this is the case when we have to apply the rqKZ equation $2k-1$ times (i.e. $2k-1$ loops). Starting with $D_{1,\dots,m}(\lambda_1-k(n+1),\lambda_2,\dots,\lambda_m)$, we apply the rqKZ equation (property (\ref{rqKZ_sln+1})) $2k-1$ times.
	\begin{align}
		\notag
		&\underset{\lambda_{1} =\lambda_2 -k(n+1)}{\text{res}} D_{1,\dots,m}(\lambda_1,\dots,\lambda_m) = \underset{\lambda_{1} = \lambda_2}{\text{res}} D_{1,2,\dots,m}(\lambda_1-k(n+1),\lambda_2,\dots,\lambda_m) =\\
		&\notag
		\underset{\lambda_{1} = \lambda_2}{\text{res}} \{A^{(2)}_{\bar{1},1}(\lambda_1-(2k-1)\frac{n+1}{2}|\lambda_2,\dots,\lambda_m)A^{(1)}_{1,\bar{1}}(\lambda_1-(2k-2)\frac{n+1}{2}|\lambda_2,\dots,\lambda_m)\cdots\\
		&\cdots A^{(2)}_{\bar{1}1}(\lambda_1-\frac{n+1}{2}|\lambda_2,\dots,\lambda_m)D^{(1)}_{\bar{1},2,\dots,m}(\lambda_1-\frac{n+1}{2},\lambda_2,\dots,\lambda_m)\}.
		\label{eqn:rqKZ_2k-1_shift_sln+1}
	\end{align}
	Now, we use the relation (\ref{proj_red_rel_sln+1}) in corollary \ref{cor:rel7_sl3} to see that we can pull $D_{3,\dots,m}(\lambda_3,\dots,\lambda_m)$ out of the residue as it doesn't depend on $\lambda_{12}=:\lambda_1-\lambda_2$. This is only possible, because the residue of $A^{(2)}_{\bar{1}1}(\lambda_1-1|\lambda_2,\dots,\lambda_m)$ at $\lambda_1=\lambda_2$ contains the projector $P^-_{\bar{1},2}$ just at the right position. We obtain
	\begin{align}
		\underset{\lambda_{1} =\lambda_2 -(k+1)}{\text{res}} D_{1,\dots,m}(\lambda_1,\dots,\lambda_m) =\tilde{X}_{2k-1} D_{3,\dots,m}(\lambda_3,\dots,\lambda_m)
		\label{eqn:Snail_Operator_firstdef_higher_rank}
	\end{align}
	and see that the product
	\begin{align*}
		\tilde{X}_{2k-1}:=\underset{\lambda_{1} = \lambda_2}{\text{res}} \{A^{(2)}_{\bar{1},1}(\lambda_1-(2k-1)\frac{n+1}{2}|\lambda_2,\dots,\lambda_m)A^{(1)}_{1,\bar{1}}(\lambda_1-(2k-2)\frac{n+1}{2}|\lambda_2,\dots,\lambda_m)\cdots\\ \cdots A^{(2)}_{\bar{1}1}(\lambda_1-\frac{n+1}{2}|\lambda_2,\dots,\lambda_m)\}
	\end{align*}
	can be used as a first definition for the Snail Operator $\tilde{X}_{k}$ for higher rank.
	Suppose for simplicity $k=2$, then we apply the rqKZ equation three times (figure \ref{fig:rqKZ_sln+1_3_loops}).
	\begin{figure}[H]
		\centering
		\begin{minipage}[b]{.4\linewidth} 
			\includegraphics[scale = .45, trim = 6cm 6cm 0cm 3cm]{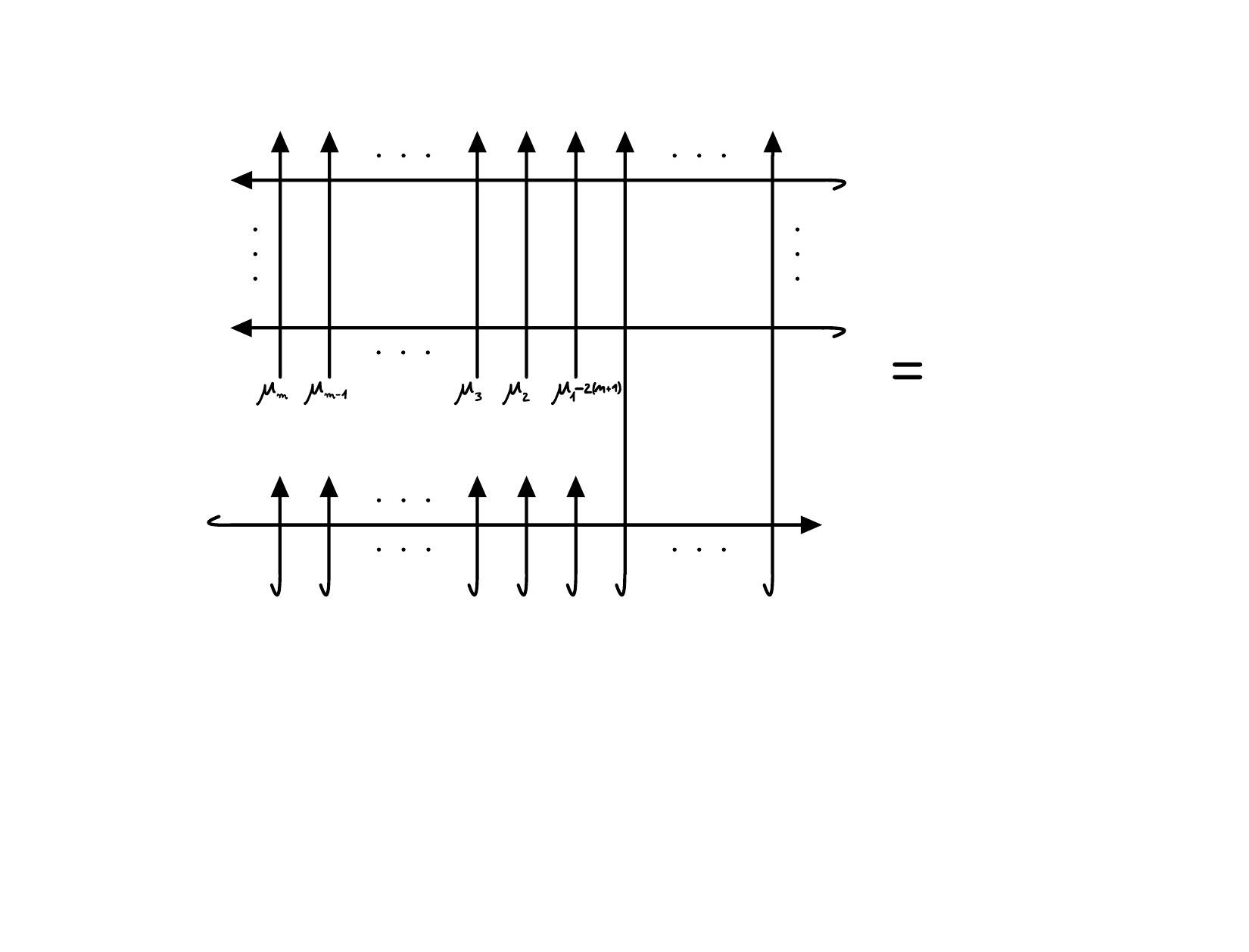}
		\end{minipage}
		\hspace{.1\linewidth}
		\begin{minipage}[b]{.4\linewidth} 
			\includegraphics[scale = .45, trim = 7cm 2.75cm 0cm 2cm]{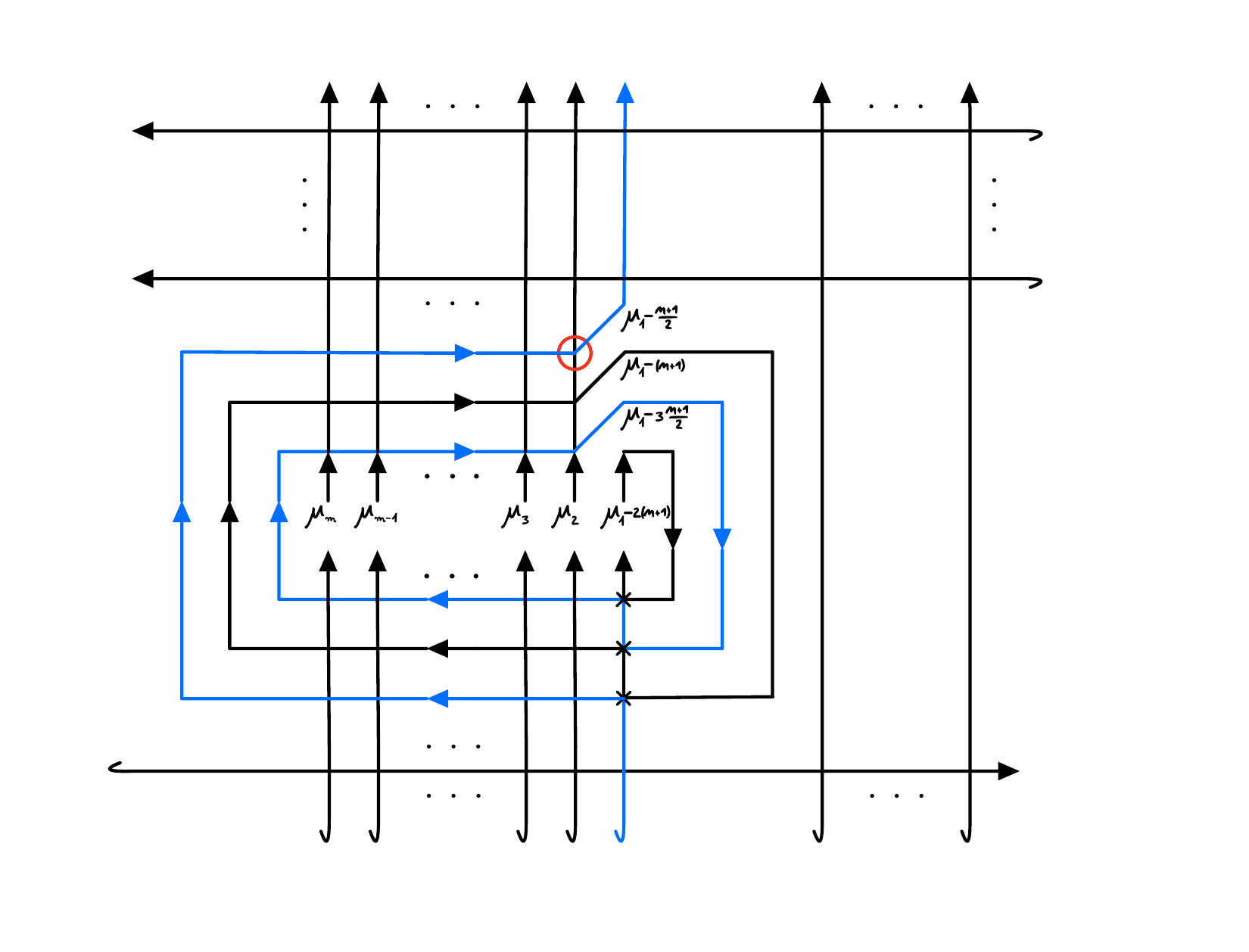}
		\end{minipage}
		\caption{Applying the rqKZ equations 3 times.}
		\label{fig:rqKZ_sln+1_3_loops}
	\end{figure}
	Taking the residue at $\mu_1 = \mu_2$, $\bar{\bar{R}}_{\bar{1}2}(\mu_1-\mu_2-\frac{n+1}{2})$ in the red circle (figure \ref{fig:rqKZ_sln+1_3_loops}) reduces to $2P^-_{\bar{1}2}$ up to a scalar prefactor. As a consequence, we can apply the relation (\ref{proj_red_rel_sln+1}) in corollary \ref{cor:rel7_sl3} to obtain the result in figure \ref{fig:Snail_with_three_loops_sln+1}.\footnote{To be precise, figure \ref{fig:Snail_with_three_loops_sln+1} has to be understood as the limit $\mu_1\to\mu_2-\frac{n+1}{2}$ of $(\mu_1-\mu_2+\frac{n+1}{2})$ times figure \ref{fig:rqKZ_sln+1_3_loops}.}
	\begin{figure}[H]
		\centering
		\includegraphics[scale = .5, trim = 3.5cm 3cm 6.5cm 2cm]{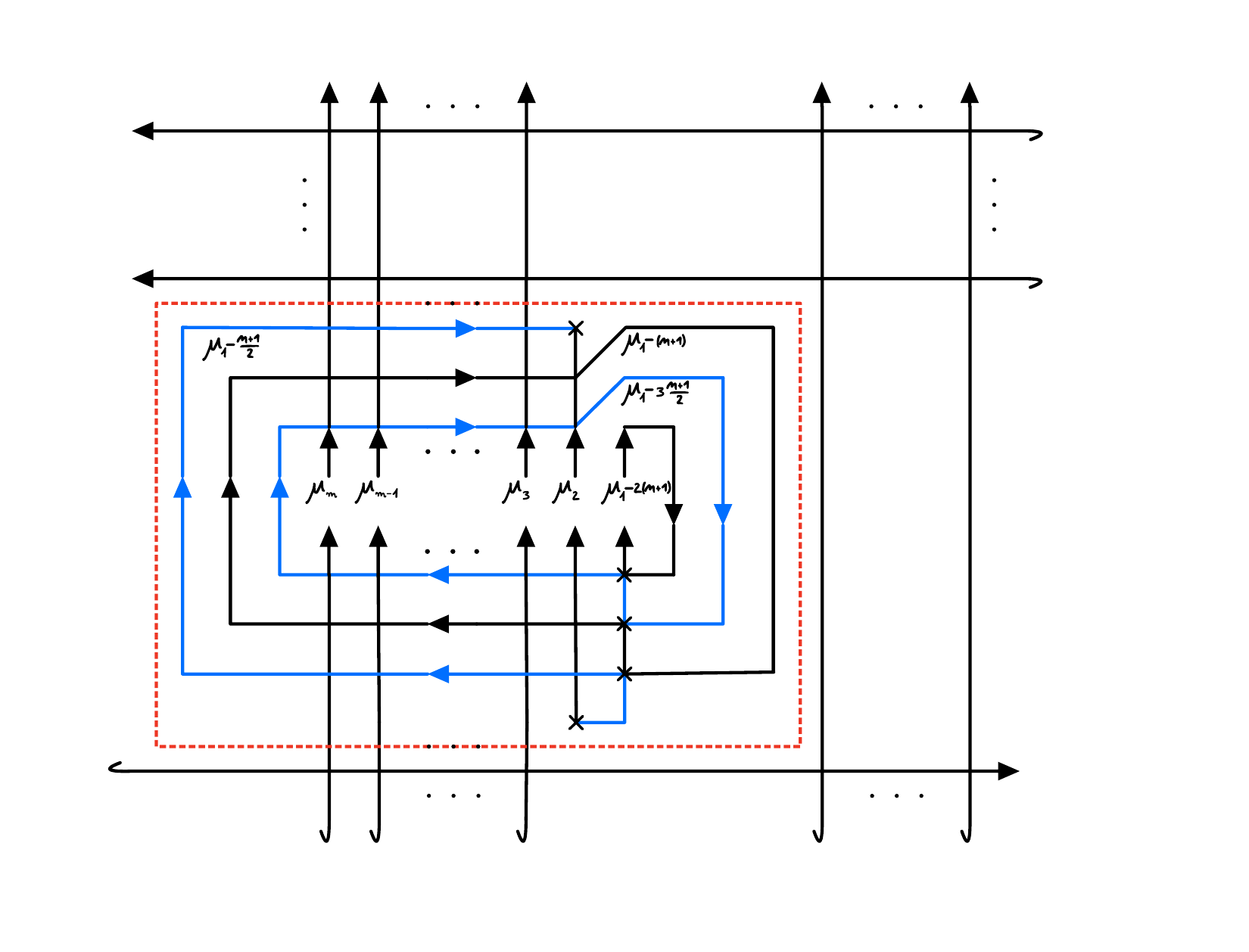}
		\caption{The Snail Operator with three loops.}
		\label{fig:Snail_with_three_loops_sln+1}
	\end{figure}
	Where we have split the operator $2P^-_{\bar{1}2}$ into the tensor product of a singlet in $\overline{V}_{\bar{1}}\otimes V_2$ (a cross with two ingoing lines) and its dual in $(\overline{V}_{\bar{1}})^\star\otimes V_2^\star$ (a cross with two outgoing lines) respectively (cf. section \ref{subsect:graphical_notation}). The operator in the box with the dashed red line (multiplied by the scalar prefactor) is the Snail Operator with three loops.
	
	Due to the fact that we have projectors in the last two loops of the Snail Operator and the spectral parameters of successive lines differ by exactly $\frac{n+1}{2}$, the Snail Operator can be further simplified by applying identities similar to figure \ref{fig:fusion_rel_sln+1}. Note that we omit the prefactor of the $R$-matrix in this figure, i.e. taking the numerical $R$-matrices $r(\lambda) = \lambda+P$ and $\bar{r} = ((\lambda+\frac{n+1}{2})1-\widetilde{C}\otimes \widetilde{C}) = \bar{\bar{r}}$ instead of $R(\lambda)$, $\bar{R}(\lambda)$ and $\bar{\bar{R}}(\lambda)$. The arguments of the $r$'s are written next to the vertices.
	\begin{figure}[H]
		\centering
		\includegraphics[scale = .6, trim = 5.75cm 9cm 5cm 4.5cm]{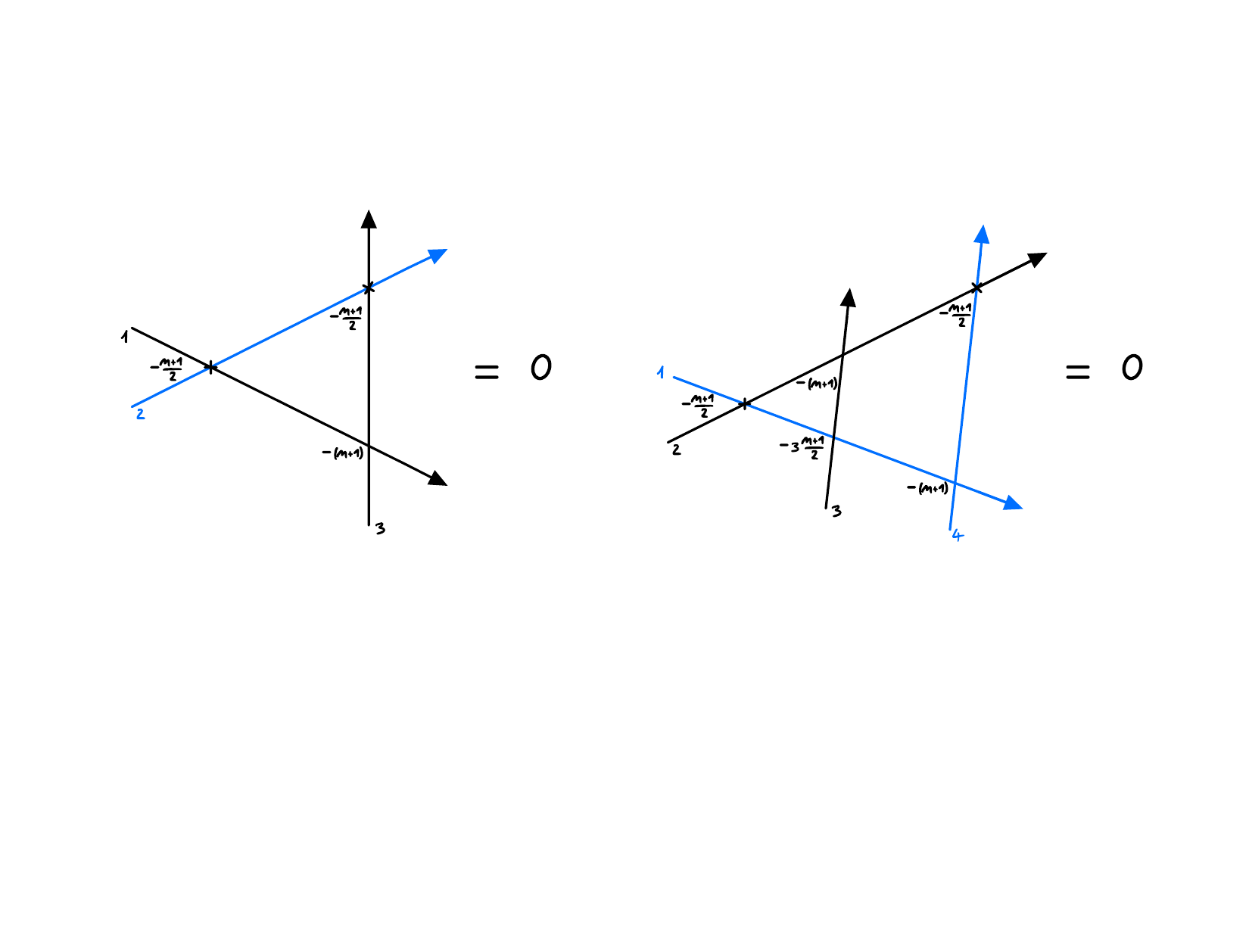}
		\caption{}
		\label{fig:fusion_rel_sln+1}
	\end{figure}
	Moreover, we claim that the $k$ loops of the Snail Operator $\tilde{X}_{k}$ again fuse to a single irreducible representation of the Yangian $Y(\mathfrak{sl}_{n+1})$, the minimal snake module which we call $S^{(k)}$. However, as a representation of $\mathfrak{sl}_{n+1}$, $S^{(k)}$ will not be irreducible anymore if $n\geq 2$. Before we give a detailed definition of snake modules \cite{MY} and the minimal snake module $S^{(k)}$, let us come back to the discussion of the pole structure of $D_{1,...,m}$ in the remaining cases. 
	We consider the residue of the density matrix $D_{1,...,m}$ at $\lambda_1=\lambda_2-k(n+1)-l$, $l=1,\dots,n$. Using the fact that $D_{1,...,m}$ has no poles at $\lambda_1-\lambda_2 \in \{0,\pm 1,\dots,\pm n\}$ (property (\ref{analyt_sln+1})) and the properties of the prefactors of $R(\lambda)$ and $\bar{R}(\lambda)$ (respectively $\bar{\bar{R}}(\lambda)$), it is easy to see that the density matrix can't have poles at $\lambda_1=\lambda_2-k(n+1)-l$ for $l=2,\dots,n$ (see Appendix \ref{App:A}). Thus, only the case $l=1$ remains for discussion. The idea is the same as in the case $l=0$. We apply the rqKZ equation to reduce $D_{1,\dots,m}(\lambda_1-k(n+1)-1,\lambda_2,\dots,\lambda_m)$ to a product of $A^{(1)}$'s and $A^{(2)}$'s acting on $D_{1,\dots,m}(\lambda_1-1,\lambda_2,\dots,\lambda_m)$
	\begin{align*}
		&D_{1,2,\dots,m}(\lambda_1-k(n+1)-1,\lambda_2,\dots,\lambda_m) =\\
		&\notag
		\underset{\lambda_{1} = \lambda_2}{\text{res}} \{A^{(2)}_{\bar{1},1}(\lambda_1-(2k-1)\frac{n+1}{2}-1|\lambda_2,\dots,\lambda_m)A^{(1)}_{1,\bar{1}}(\lambda_1-(2k-2)\frac{n+1}{2}-1|\lambda_2,\dots,\lambda_m)\cdots\\
		&\cdots A^{(1)}_{\bar{1}1}(\lambda_1-1|\lambda_2,\dots,\lambda_m)D_{1,2,\dots,m}(\lambda_1-1,\lambda_2,\dots,\lambda_m)\}.
	\end{align*}
	At $\lambda_1=\lambda_2$ we have a simple pole as long as $\lim_{\lambda_1\to\lambda_2}D_{1,2,\dots,m}(\lambda_1-1,\lambda_2,\dots,\lambda_m)\neq 0$. Taking the residue at $\lambda_1=\lambda_2$, $R_{12}(\lambda_1-\lambda_2-1)$ reduces to the projector onto the second fundamental representation with fundamental weight $\omega_2$ up to a scalar prefactor. In the case of $\mathfrak{sl}_2$ there is no fundamental weight $\omega_2$ and we observe that it is exactly the projector onto the singlet.\footnote{ Up to a minus sign which can be absorbed into the prefactor.} This is clear by the isomorphy of the fundamental and antifundamental representation (or self-duality of representations) of $\mathfrak{sl}_2$. Therefore, everything goes completely analogous to the case $l=0$ (cf. section \ref{sect:sl2snailconstr}). However, in the case of $\mathfrak{sl}_3$ the fundamental representation with fundamental weight $\omega_2$ is exactly the antifundamental representation. Thus, we expected to find the reduction relation
	\begin{align}
		R_{12}(-1)D_{1,2,\dots,m}(\lambda-1,\lambda,\dots,\lambda_m) = F_{\bar{1}}^{1,2}D^{(1)}_{\bar{1},3,\dots,m}(\lambda-\frac{1}{2},\lambda_3,\dots,\lambda_m)F_{1,2}^{\bar{1}},
	\end{align}
	where $F_{1,2}^{\bar{1}}$ is a map from $V_1\otimes V_2$ to $\overline{V}_{\bar{1}}$ ('fusion') and $F_{\bar{1}}^{1,2}$ is a map from $\overline{V}_{\bar{1}}$ to $V_1\otimes V_2$ ('defusion') such that $F_{\bar{1}}^{1,2}F_{1,2}^{\bar{1}}=R_{12}(-1)$. Since $R_{12}(-1)$ is proportional to the projector onto $\overline{V}_{\bar{1}}\subset V_1\otimes V_2$, the decomposition into $F_{1,2}^{\bar{1}}$ and $F_{\bar{1}}^{1,2}$ is unique if we demand $F_{\bar{1}}^{1,2}:=(F_{1,2}^{\bar{1}})^t$, the transpose of $F_{1,2}^{\bar{1}}$. Thus, the relation reduces $D_{1,2,\dots,m}(\lambda-1,\lambda,\dots,\lambda_m)$ to $D^{(1)}_{\bar{1},3,\dots,m}(\lambda-\frac{1}{2},\lambda_3,\dots,\lambda_m)$. Then, the idea is to apply the rqKZ-equation (Property (\ref{rqKZ_sln+1})) to $D^{(1)}_{\bar{1},3,\dots,m}(\lambda-\frac{1}{2},\lambda_3,\dots,\lambda_m)$ one more time to recover the usual density matrix $D_{1,3,\dots,m}(\lambda+1,\lambda_3,\dots,\lambda_m)$ of length $m-1$ as follows
	\begin{align}
		R_{12}(-1)&D_{1,2,\dots,m}(\lambda-1,\lambda,\dots,\lambda_m) =\notag\\ &F_{\bar{1}}^{1,2}A^{(1)}_{1,\bar{1}|3,\dots,m}(\lambda+1|\lambda_3,\dots,\lambda_m)\left(D_{1,3,\dots,m}(\lambda+1,\lambda_3,\dots,\lambda_m)\right)F_{1,2}^{\bar{1}}.
	\end{align}
	Now, as $D_{1,3,\dots,m}(\lambda+1,\lambda_3,\dots,\lambda_m)$ doesn't have a pole at $\lambda = 0$, we can pull it out of the residue as in the case $l=0$. Using mathematica and the results in the papers \cite{BHN} and \cite{KR}, we could verify this relation for $m\leq 3$ and reproduce the first few residues in these cases in the described way. However, in general when the rank is greater than two, $\omega_2$ is just the second fundamental representation and we can only hope to find a relation which projects the first two fundamental lines of $D_{1,2,\dots,m}(\lambda-1,\lambda,\dots,\lambda_m)$ onto a line with the second fundamental representation and spectral parameter $\lambda-\frac{1}{2}$. This will be clear from the discussion of the extended T-systems in the next section. The corresponding Young tableaux for the tensor product of two fundamental representations of $\mathfrak{sl}_{n+1}$ are depicted in figure \ref{fig:Young_tableaux_ff}.
	\begin{figure}[H]
		\centering
		\includegraphics[scale = .6, trim = 5cm 9.5cm 6cm 8.25cm]{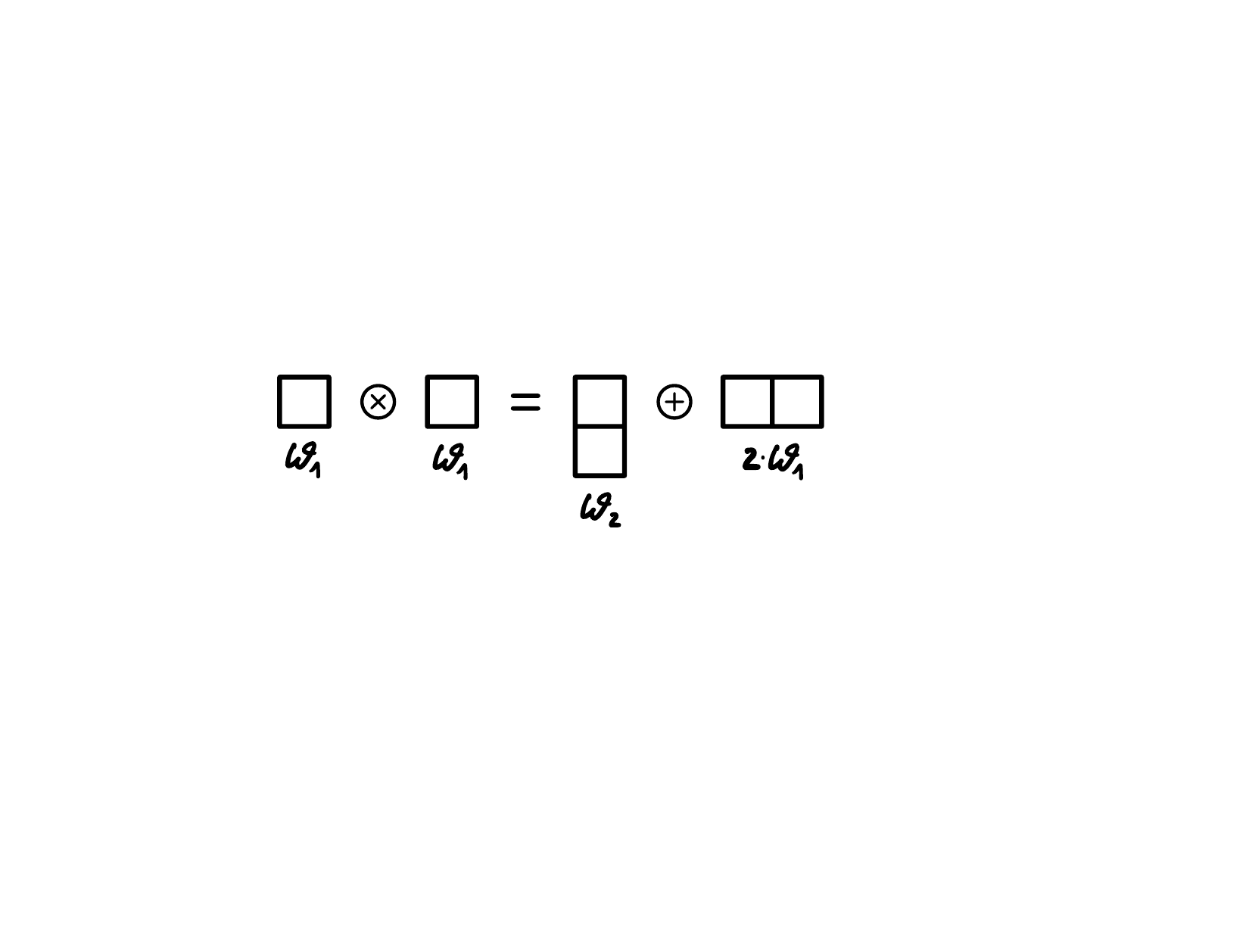}
		\caption{The tensor product of two fundamental representations of $\mathfrak{sl}_{n+1}$.}
		\label{fig:Young_tableaux_ff}
	\end{figure}
	From this point, it is unclear for us how to come back to $D_{1,3,\dots,m}(\lambda+1,\lambda_3,\dots,\lambda_m)$ in general. It seems that additional information is needed when the rank $n$ is greater than $2$. Note that this problem has already been observed in the paper \cite{KR} and a solution was presented by introducing two possible generalizations of the physical density operator. However, yet we do not know if there is a similar construction to describe the residues corresponding to the case $l=1$ above. 
	\subsection{Extended T-systems and the Snail Operator $\tilde{X}_k$}
	\label{subsect:extT-systems}
	Let us now come back to the case $l=0$ and the discussion of the Snail Operator $\tilde{X}_k$ in the general case. The idea is similar to the $\mathfrak{sl}_2$ case and the short exact sequences that appear are almost analogue. However, the representations are not the higher rank Kirillov--Reshetikhin modules, but the so called snake modules. These were introduced by Mukhin and Young in 2012 \cite{MY} for the quantum affine algebras of type A and B. Of course, they also apply to the Yangians by means of the equivalence of categories between the finite dimensional (type $\boldsymbol{1}$) representations of Yangians and of quantum affine algebras stated in \cite{GT}. We only need these results for type $A_n$ in our case (i.e. $Y(\mathfrak{sl}_{n+1})$). We draw the Snail Operator in two equivalent ways as shown in figure \ref{fig:gen_snail_op}. Since the $R$-matrix $\bar{\bar{R}}(\lambda)$ has a simple pole at $\lambda=\frac{n+1}{2}$, figure \ref{fig:gen_snail_op} is only understood in terms of the residue at $\mu_1=\mu_2-k(n+1)$. As above, we multiply it by the scalar prefactor obtained from $\bar{R}(\lambda)$ in the limit $\lambda\to-\frac{n+1}{2}$. Anyway, we explain how figure \ref{fig:Snail_Operator_sl2} can be made into a precise definition in a moment.
	\begin{figure}[H]
		\centering
		\begin{minipage}[b]{.44\linewidth} 
			\includegraphics[scale = .5, trim = 5cm 5cm 3cm 4cm]{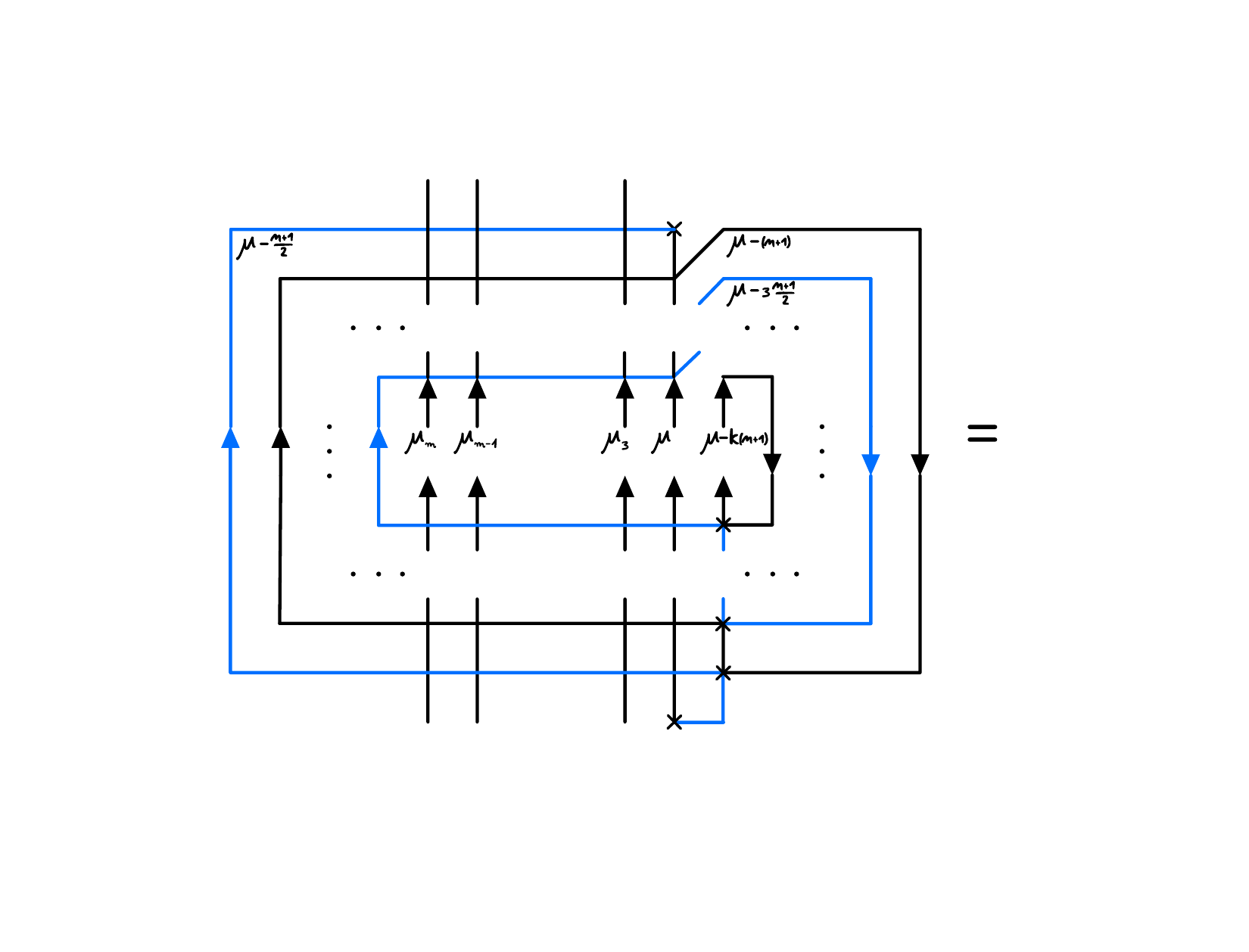}
		\end{minipage}
		\hspace{.1\linewidth}
		\begin{minipage}[b]{.44\linewidth} 
			\includegraphics[scale = .5, trim = 6.5cm 5cm 3cm 6cm]{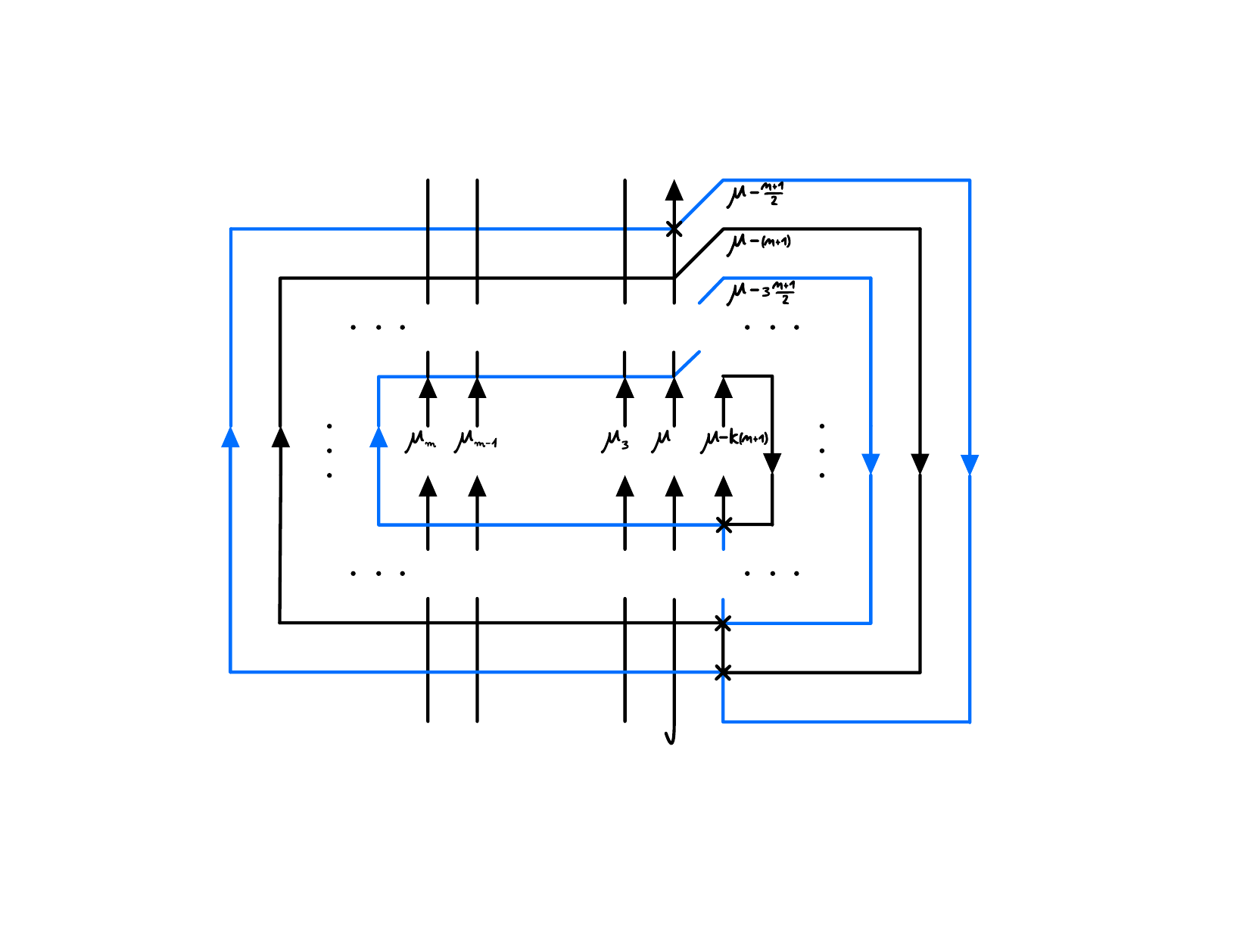}
		\end{minipage}
		\caption{The Snail Operator with $2k-1$ (closed) loops.}
		\label{fig:gen_snail_op}
	\end{figure}
	We remind the reader that every black line is regarded a fundamental representation whereas any blue line is regarded an antifundamental representation of the Yangian $Y(\mathfrak{sl}_{n+1})$. As explained in section \ref{subsect:graphical_notation} and \ref{subsect:dual modules and crossing}, they can be obtained by pulling back the fundamental (respectively antifundamental) representation $V=V^{(1)}_1$ (respectively
	$\overline{V}=V^{(1)}_n$) of $\mathfrak{sl}_{n+1}$ to the fundamental weight $\omega_1$ (respectively $\omega_n$) by the evaluation homomorphism $\operatorname{ev}_a$, where $a\in\mathbb{C}$ is the spectral parameter associated to the corresponding line. We write $\operatorname{ev_a}^*(V^{(1)}_1)=:V^{(1)}_1(a)$ and $\operatorname{ev_a}^*(V^{(1)}_n)=:V^{(1)}_n(a)$. Moreover, we define the evaluation representations of the $\mathfrak{sl}_{n+1}$ representations $V^{(k)}_i$ with highest weight $k\omega_i$, $i=1,\dots,n$, as $V^{(k)}_i(a):=\operatorname{ev}_a^*(V^{(k)}_i)$. They are related to the higher rank Kirillov--Reshetikhin modules via $W^{(k)}_i(a):=V^{(k)}_i(a+\frac{1}{2}(\frac{2ki}{n+1}+i-n-1))$, where $W^{(k)}_1(a)=W^{(k)}(a)$ recovers the definition for $\mathfrak{sl}_2$ in Section \ref{subsect:T-systems}.\footnote{Note that there are actually two evaluation homomorphisms for higher rank. Here, we fix $ev_a$ as $ev^a$ in the book \cite{CPBook} Chapter 12.} The fundamental representations of the successive lines are therefore given by $W^{(1)}_n(\mu-l(n+1)+\frac{n+1}{2})$, $l=1,\dots,k$, and $W^{(1)}_1(\mu-l(n+1))$, $l=1,\dots,k-1$. Looking at the right side of figure \ref{fig:gen_snail_op}, we use the identity $A_1 = \tr_{V_\alpha}(A_\alpha P_{\alpha,1}), \, A\in \End(V_1)$, to write the Snail Operator as the residue at $\mu = \mu_2$ of the alternating product of the monodromy matrices
	\begin{align*}
		&\bar{\mathcal{T}}_{\bar{\alpha}_{2l-1};2,\dots,2m-1}(\mu-(2l-1)\frac{n+1}{2};\mu_2,\dots,\mu_m,\mu_m,\dots,\mu_2):=\\
		&\tr_{\bar{a}}\{\overline{T}_{\bar{a};2,\dots,m}(\mu-(2l-1)\frac{n+1}{2};\mu_2,\dots,\mu_m) T_{\bar{a};m+1,\dots,2m-1}(\mu-(2l-1)\frac{n+1}{2};\mu_m,\dots,\mu_2)P_{\bar{a},\bar{\alpha}_{2l-1}}\}\\ &\stackrel{P_{\bar{a},\bar{\alpha}_{2l-1}}=R_{\bar{a},\bar{\alpha}_{2l-1}}(0)}{=}\\
		&\tr_{\bar{a}}\{\bar{R}_{2,\bar{a}}(\mu_2-\mu+(2l-1)\frac{n+1}{2})\bar{R}_{3,\bar{a}}(\mu_{3}-\mu+(2l-1)\frac{n+1}{2})\cdots \bar{R}_{m,\bar{a}}(\mu_m-\mu+(2l-1)\frac{n+1}{2})\\
		&\bar{\bar{R}}_{\bar{a},m+1}(\mu-(2l-1)\frac{n+1}{2}-\mu_m)\bar{\bar{R}}_{\bar{a},m+2}(\mu-(2l-1)\frac{n+1}{2}-\mu_{m-1})\cdots\bar{\bar{R}}_{\bar{a},2m-1}(\mu-(2l-1)\frac{n+1}{2}-\mu_2)\cdot\\
		&\cdot R_{\bar{a},\bar{\alpha}_{2l-1}}(0)\}\quad l=1,\dots,k,
	\end{align*}
	and
	\begin{align*}
		&\mathcal{T}_{\alpha_{2l};2,\dots,2m-1}(\mu-l(n+1);\mu_2,\dots,\mu_m,\mu_m,\dots,\mu_2):=\\
		&\tr_a\{\overline{T}_{a;2,\dots,m}(\mu-l(n+1);\mu_2,\dots,\mu_m)T_{a;m+1,\dots,2m-1}(\mu-l(n+1);\mu_m,\dots,\mu_2)P_{a,\alpha_{2l}}\} \stackrel{P_{a,\alpha_{2l}}=R_{a,\alpha_{2l}}(0)}{=}\\
		& \tr_a\{R_{2,a}(\mu_2-\mu+l(n+1))R_{3,a}(\mu_{3}-\mu+l(n+1))\cdots R_{m,a}(\mu_m-\mu+l(n+1))\\
		&R_{a,m+1}(\mu-l(n+1)-\mu_m)R_{a,m+2}(\mu-l(n+1)-\mu_{m-1})\cdots R_{a,2m-1}(\mu-l(n+1)-\mu_2)R_{a,\alpha_{2l}}(0)\}\\
		&l=1,\dots,k-1,
	\end{align*}
	multiplied by the operator
	\begin{align*}
		\mathfrak{P}_{2,\alpha_1,\dots,\alpha_{2k-1},1}:= (n+1)^{2k-1} P^-_{\bar{\alpha}_1,2}P^-_{\alpha_{2},\bar{\alpha}_1}\cdots P^-_{\bar{\alpha}_{2k-1},\alpha_{2k-2}}P^-_{1,\bar{\alpha}_{2k-1}}
	\end{align*}
	and contracted over the spaces $2,\bar{\alpha}_1,\alpha_2,\bar{\alpha}_3,\dots,\alpha_{2k-2},\bar{\alpha}_{2k-1}$. Note that we used the projector identity $(P^-_{\bar{\alpha}_1,2})^2= P^-_{\bar{\alpha}_1,2}$ to be able to introduce the operator $\mathfrak{P}$. 
	As in the $\mathfrak{sl}_2$ case, we expect that $\mathfrak{P}$ is a projector on some subrepresentation in the tensor product of the spaces $\bar{\alpha}_1,\alpha_2,\bar{\alpha}_3,\dots,\alpha_{2k-2},\bar{\alpha}_{2k-1}$ times the supposed singlet in the tensor product of two adjoint representations built from the spaces $V_1\otimes\overline{V}_1$ and $V_2\otimes\overline{V}_2$ when acting on the $R$-matrices on the vertical line with spectral parameter $\lambda$.\footnote{$\mathfrak{P}$ itself has rank $(n+1)^{2k-1}$, but it is further reduced due to the fusion properties of the $R$-matrices.} Where we again identify $V_1\otimes \overline{V}_1$ with $V_1\otimes V_1^\star \cong \End(V_1)$ using the dual of the singlet in $V_1^\star\otimes (\overline{V}_1)^\star$ and similarly for $V_2$. Therefore, we are interested in the irreducible composition factors in the tensor product of fundamental and antifundamental representations of the Yangian $Y(\mathfrak{sl}_{n+1})$. Though, in contrast to the $\mathfrak{sl}_2$ case the representation theory for rank ($n\geq2$) is quite different. The irreducible representations of the Yangian $Y(\mathfrak{sl}_{n+1})$ are not necessary irreducible as representations of $\mathfrak{sl}_{n+1}$ anymore. However, for any representation of $\mathfrak{sl}_{n+1}$ we still obtain representations of the Yangian of the same (Lie algebra) weight using $\operatorname{ev}_a$.
	For any other type where no evaluation homomorphism is present, even this property doesn't hold. In this case one can consider minimal affinizations (MA) among all the possible affinizations \footnote{a representation of $Y(\mathfrak{g})$ that has the representation of $\mathfrak{g}$ as a proper $\mathfrak{g}$-subrepresentation}. Luckily, the fundamental modules always have only one (up to equivalence \footnote{isomorphy as representations of $\mathfrak{g}$}) (minimal) affinization (cf. \cite{CP1994} section 6).\\
	
	However, in order to understand the irreducible composition factors in the tensor products of the fundamental (evaluation) representations of the Yangian $Y(\mathfrak{sl}_{n+1})$, we need to dig deeper into the representation theory. We intend to give a short review of the paper \cite{MY} of Mukhin and Young and explain how their extended T-systems naturally generalize the rank $1$ case explained above. As the paper \cite{MY} is about the quantum affine algebras, we may later again use the result of the paper \cite{GT} to come back to the Yangians. We shall also refer the reader to the papers \cite{HL}, \cite{FM}, \cite{CP1994}, \cite{CP1991} and the books \cite{Carter}, \cite{CPBook}, where the proofs and some basic definitions can be found.
	
	Let us restrict ourselves to the category $\mathrm{C}$ of finite dimensional $U_q(\tilde{\mathfrak{g}})$-modules of type 1. As explained in remark \ref{rem:type_1}, we can equivalently consider the category of finite dimensional (type 1) representations of the quantum loop algebra $U_q(\mathcal{L}(\mathfrak{g}))$. Therefore, $C^{1/2}\equiv1$ and the $\mathcal{H}_{i,r}$ and $\mathcal{K}_i$ mutually commute.\footnote{ Since $U_q(\tilde{\mathfrak{g}})$ is a Hopf algebra, $\mathrm{C}$ is an abelian monoidal category.} As a consequence, we can write any object $V\in\mathrm{C}$ as a direct sum of common generalized eigenspaces for the action of the $\mathcal{K}_i$ and $\mathcal{H}_{i,r}$, the so called \textit{loop-} or \textit{l-weight-spaces} of $V$. The eigenvalues are given as follows (cf. \cite{FM} prop. 2.4).
	\begin{prop}[l-weight]
		The eigenvalues of the $\mathcal{H}_{i,r}$ ($r>0$) in an \textit{l-weight-space} $W$ of $V$ are always of the form
		\begin{align}
			\frac{q^m-q^{-m}}{m(q-q^{-1})}\left(\sum_{r=1}^{k_i}(a_{ir})^m-\sum_{s=1}^{l_i}(b_{is})^m\right)\qquad a_{ir},\, b_{ir}\in\mathbb{C}^\times.\label{eqn:l-weight}
		\end{align}
		They completely determine the eigenvalues of $\mathcal{H}_{i,r}\,(r<0)$ and $\mathcal{K}_i$ on $W$.
		The collection of eigenvalues (\ref{eqn:l-weight}) is called the $l$\textit{-weight} of $W$. $\odot$
	\end{prop}
	\begin{definition}[$q$\textit{-character}]
		Define the $q$\textit{-character} of $V$ as a Laurent polynomial with positive integer coefficients in some indeterminates $Y_{i,a}\,(i\in I,\, a\in\mathbb{C}^\times)$, that encode the decomposition of $V$ into $l$\textit{-weight-spaces}.
		
		The collection of eigenvalues (\ref{eqn:l-weight}) is encoded by the Laurent monomial
		\begin{align}
			\prod_{i\in I}\left(\prod_{r=1}^{k_i}Y_{i,a_{ir}}\prod_{s=1}^{l_i}Y_{i,b_{is}}^{-1}\right)\label{eqn:l-weight_poly}
		\end{align}
		and the coefficient of it in the $\boldsymbol{q}$\textit{\textbf{-character}} of $V$ is the dimension of $W$.
		We equivalently say that the monomial (\ref{eqn:l-weight_poly}) is the $l$\textit{-weight} of $W$. Moreover, define $\mathcal{Y} := \mathbb{Z}[Y_{i,a}^{\pm 1}]_{i\in I;a\in\mathbb{C}^\times}$ and let $\mathcal{P}\subset\mathcal{Y}$ be the multiplicative abelian subgroup of all monomials. Then $\mathcal{P}$ is in bijection with the set of all l-weights and $\chi_q(V)\in \mathcal{Y}$ denotes the $q$\textit{-character} of $V\in\mathrm{C}$. $\odot$
	\end{definition}
	\begin{rem}
		Certainly, the notion of the $Y_{i,a}$ can be interpreted as a short hand notation for the Drinfeld polynomial $P_i(u)=1-ua$. We will recall some of their properties in terms of the $Y_{i,a}$ in a moment. However, we shall refer the reader to the results of Chari and Pressley \cite{CPBook}, \cite{CP1991}, \cite{CP1994}, \cite{CP1996} in the case of quantum affine algebras and the result of Drinfeld \cite{D2} for the Yangians. In particular, for each $i\in I$ and $a\in \mathbb{C}^\times$ we can define an irreducible representation $V_{\omega_i}(a):=V(\textbf{P}_a^{(i)})$ to the highest weight $\textbf{P}^{(i)}_a$, which is the $I$-tuple of polynomials, such that $P_i(u)=1-ua$ and $P_j(u)=1,$ $\forall j\neq i$. Anyway, we prefer to define everything in terms of the $Y_{i,a}$ as in \cite{FM}, \cite{FR}, \cite{HL} and \cite{MY}. We summarize that $V_{\omega_i}(a)$ is the representation with highest $l$-weight $Y_{i,a}$.
	\end{rem}
	The properties of the q-character $\chi_q$ and the relation to the usual character $\chi$ of the corresponding $U_q(\mathfrak{g})$-module are summarized in the following theorem (cf. \cite{FM} theorem 2.2 and \cite{FR} section 3).
	\begin{thm}[properties of $\chi_q$]\hspace{1em}
		\label{thm:properties_of_chi_q}
		\begin{enumerate}
			\item $\chi_q$ is an injective ring-homomorphism from the Grothendieck ring $\text{Rep}(U_q(\tilde{\mathfrak{g}}))$ to $\mathcal{Y}$.
			\item For any finite-dimensional representation $V$ of $U_q(\tilde{\mathfrak{g}})$, $\chi_q(V)\in\mathbb{Z}_+[Y_{i,a}^{\pm 1}]_{i\in I;a\in\mathbb{C}^\times}(=:\mathcal{Y}_+)$.
			\item Let $\chi : \text{Rep}(U_q(\mathfrak{g}))\to \mathbb{Z}[e^{\pm \omega_i}]_{i\in I}$ be the $U_q(\mathfrak{g})$-character homomorphism, let $\operatorname{wt}:\mathcal{Y}\to\mathbb{Z}[e^{\pm \omega_i}]_{i\in I}$ be the homomorphism \footnote{induced by the homomorphism of abelian groups $\operatorname{wt}:\mathcal{P}\to P$, $Y_{i,a}\mapsto \omega_i$} defined by $Y^{\pm1}_{i,a}\mapsto e^{\pm\omega_i}$ and let res$:\text{Rep}(U_q(\tilde{\mathfrak{g}}))\to \text{Rep}(U_q(\mathfrak{g}))$ be the restriction homomorphism. Then the diagram
			\begin{align*}
				\xymatrix @!=0cm @C=4cm @R=2cm{
					\text{Rep}(U_q(\tilde{\mathfrak{g}})) \ar[r]^{\chi_q} \ar[d]_{\text{res}}    &  \mathbb{Z}[Y_{i,a}^{\pm 1}]_{i\in I, a\in\mathbb{C}^\times}   \ar[d]^{\operatorname{wt}}  \\
					\text{Rep}(U_q(\mathfrak{g})) \ar[r]^{\chi}&   \mathbb{Z}[e^{\pm \omega_i}]_{i\in I}   
				}
			\end{align*}
			commutes (i.e. $\chi_q(V)$ reduces to $\chi(V)$ on the subalgebra $U_q(\mathfrak{g}) \leq U_q(\tilde{\mathfrak{g}})$).
			\item $\text{Rep}(U_q(\tilde{\mathfrak{g}}))$ is a commutative ring that is isomorphic to $\mathbb{Z}[t_{i,a}]_{i\in I;a\in\mathbb{C}^\times}$, where $t_{i,a}$ is the class of $V_{\omega_i}(a)$. \footnote{ $V_{\omega_i}(a) = V_i^{(1)}(a)$ when $\mathfrak{g}$ is of of type A.} $\odot$
		\end{enumerate}
	\end{thm}
	As for the representation theory of $\mathfrak{g}$ (or rather $U_q(\mathfrak{g})$) we can introduce the notion of \textit{highest (l-)weights} for $U_q(\tilde{\mathfrak{g}})$. It can be stated in the following way (cf. \cite{CPBook} and \cite{CP1994}).
	
	For each $j\in I$, a monomial $m = \prod_{i\in I, a\in\mathbb{C}^\times}Y_{i,a}^{u_{i,a}}\in \mathcal{P}$ is said to be $j$\textit{-dominant} (resp. $j$-anti-dominant) $\Leftrightarrow\, u_{j,a}\geq 0$ (resp. $u_{j,a}\leq 0$) for all $a\in \mathbb{C}^\times$. It is said to be \textit{(anti-)dominant} if it is $j$\textit{(-anti)-dominant} for all $j\in I$. We denote by $\mathcal{P}^+\subset\mathcal{P}$ the \textit{set of dominant monomials}.
	
	Let $\mathcal{P}(V):=\{m\in \mathcal{P}: m\text{ is a monomial of }\chi_q(V)\}\subset\mathcal{P}$ and let $m\in\mathcal{P}(V)$ be dominant, then a vector $\ket{m}\in V_m\backslash\{0\}$ is called a \textit{highest} $l$\textit{-weight vector} with \textit{highest} $l$\textit{-weight} $m$, if $\mathcal{X}_{i,r}^+\ket{m}=0 \text{ for all } i\in I, r\in\mathbb{Z}$ and $\ket{m}$ is a simultaneous eigenvector for the $\mathcal{K}_i$ and $\mathcal{H}_{i,r}$. $V$ is called a \textit{highest} $l$\textit{-weight representation} with highest l-weight $m$, if $V=U_q(\tilde{\mathfrak{g}})\ket{m}$.
	
	Now, it is known that for each $m\in\mathcal{P}^+$, there is a unique finite-dimensional simple module, denoted $L(m)$ that is highest l-weight with highest l-weight $m$. Conversely, every finite dimensional irreducible $U_q(\tilde{\mathfrak{g}})$-module is of this form for some $m\in\mathcal{P}^+$.
	
	In addition, we should add the following definition.
	\begin{definition}[special, thin, prime, real]\hspace{1em}
		\begin{enumerate}
			\item A module $V\in\mathrm{C}$ is said to be \textit{special (resp. anti-specail)} if $\chi_q(V)$ has exactly one dominant (resp. anti-dominant) monomial.
			\item It is called \textit{thin} if no l-weight space of $V$ has dimension greater than one.
			\item $V$ is said to be \textit{prime} if it is not isomorphic to a tensor product of two nontrivial $U_q(\hat{\mathfrak{g}})$-modules.
			\item $V$ is called \textit{real}, if $V\otimes V$ is simple. $\odot$
		\end{enumerate}
	\end{definition}
	As $\mathcal{P}$ is an 'affine' analogue of the weight lattice we can also introduce an 'affine' analogue of the root-lattice. Let's also assume that $\mathfrak{g}$ is single laced for simplicity (see e.g. \cite{MY} section 2.3 for the general case).
	\begin{definition}[the 'affine' root lattice]
		For $i\in I$, $a\in\mathbb{C}^\times$ and $A=(a_{ij})$ the Cartan matrix define
		\begin{align}
			A_{i,a} = Y_{i,aq}Y_{i,aq^{-1}}\prod_{j\neq i}Y^{a_{ij}}_{j,a},
		\end{align}
		then $\operatorname{wt}(A_{i,a}) = \alpha_i$, i.e. $A_{i,a}$ can be viewed as an \textit{'affine' simple root}.
		
		Let $\mathcal{Q}$ be the subgroup of $\mathcal{P}$ generated by the $A_{i,a}, \,i\in I\, a\in \mathbb{C}^\times,$ and let $\mathcal{Q}^\pm$ be the monoid generated by $A_{i,a}^{\pm1},\,i\in I\, a\in \mathbb{C}^\times$. $\mathcal{Q}$ can be considered the \textit{'affine' root lattice} and $\mathcal{Q}^+$ ($\mathcal{Q}^-$) the sets of positive (negative) 'affine' simple roots. $\odot$
	\end{definition}
	The 'affine' weight lattice $\mathcal{P}$ and the 'affine' root lattice $\mathcal{Q}$ are compatible with the usual weight lattice $P$ and root lattice $Q$ as follows.
	\begin{cor}[partial order]\hspace{1em}
		\begin{enumerate}
			\item There is a \textit{partial order} $\leq$ on $\mathcal{P}$ such that $m\leq m'$ iff $m'm^{-1}\in\mathcal{Q}^+$.
			\item The partial order on $\mathcal{P}$ is compatible with the partial order on $P$ in the sense $m\leq m' \, \Rightarrow \operatorname{wt} m \leq \operatorname{wt} m'$.
			\item For all $m^+\in\mathcal{P}^+$ we have $\mathcal{P}(L(m^+))\subset m^+\mathcal{Q}^-$. $\odot$
		\end{enumerate}
	\end{cor}
	Therefore, we can conclude that $\chi_q(L(m)) = m \left(1+\sum_{p}M_p\right)$, where the $M_p$ are monomials in the variables $A_{i,a}^{-1}$.
	
	Let us now focus on the type A, where we have an evaluation homomorphism. We refer the reader to the original paper \cite{MY} for type B.
	
	We recap some facts for the case $\mathfrak{g}=\mathfrak{sl}_2$ in terms of l-weights (see \cite{CP1991}). Due to \textit{Jimbos homomorphism} $\text{ev}_a:U_q(\hat{\mathfrak{sl}_2})\to U_q(\mathfrak{sl}_2)$ we can get type 1 \textit{spin} $k/2$ evaluation representations $V^{(k)}(a)$ by pulling back with $\text{ev}_a$.
	$V^{(k)}(a)$ is a highest $l$\textit{-weight representation} with highest $l$\textit{-weight} $Y_{aq^{k-1}}Y_{aq^{k-3}}\cdots Y_{aq^{-k+1}}=:S_k(a)$, called $q$\textit{-String}.
	Let $V= V^{(k)}(a)\otimes V^{(l)}(b)$ and $ 0\leq p < \min\{k,l\}$ be an integer. Then $V$ is irreducible iff $b/a \neq q^{\pm(k+l-2p)}$. In this case, $S_k(a)$ and $S_l(b)$ are said to be in \textit{general position}. Otherwise $S_k(a)$ and $S_l(b)$ are in \textit{special position} and $V= V^{(k)}(a)\otimes V^{(l)}(b)$ has a unique proper submodule (c.f. \ref{prop:special_position}). In addition, every finite dimensional simple $U_q(\hat{\mathfrak{sl}_2})$-module is isomorphic to a tensor product of evaluation representations.
	The module $W^{(k)}(a) := V^{(k)}(aq^{k-1})$ is called Kirillov--Reshetikhin module.
	
	Coming back to the general case, we have seen above (thm. \ref{thm:properties_of_chi_q} (4)) that the Grothendieck ring of $\mathrm{C}$ is the polynomial ring over $\mathbb{Z}$ in the classes $[V_{\omega_i}(a)]\;(i\in I,\, a\in \mathbb{C}^\times)$ of fundamental modules. Kirillov--Reshetikhin modules are then defined through $W_{i}^{(k)}(a) := L(S_k^i(aq^{k-1})) (= \operatorname{ev}_{aq^{k-1}}^*(V^{(k)}_i))$ as above, where $S_k^i(aq^{k-1})$ is obtained from $S_k(aq^{k-1})$ via $Y_{aq^l}\mapsto Y_{i,aq^l}$.\footnote{ The definition is as in the Yangian case above, but in terms of their Drinfeld polynomials. Clearly, this also explains how KR modules are naturally defined for any other Dynkin type symmetry.} Note that we have $\text{wt}(W_{i}^{(k)}(a)) = k\omega_i$. In particular, $W^{(1)}_{i}(a)$ coincides with the fundamental module $V_{\omega_i}(a)$.
	
	The classes $[W_{i}^{(k)}(a)]$ in $\mathcal{Y}$ satisfy the \textit{T-system}
	\begin{align}
		[W_{i}^{(k)}(a)][W_{i}^{(k)}(aq^2)]=[W_{i}^{(k+1)}(a)][W_{i}^{(k-1)}(aq^2)]+\prod_{a_{ij}=-1}[W_{j}^{(k)}(aq)].
	\end{align}
	It generalizes the T-system (\ref{eqn:tsys_1}) and is sometimes referred to as \textit{the T-system} \cite{KNS} \cite{N} \cite{H}. In fact, it can be used to calculate the class $[W^{(k)}_i(a)]$ inductively as a polynomial in the classes of fundamental modules $[V_{\omega_i}(a)]$, $i\in I,\, a\in \mathbb{C}^\times$ (c.f. \cite{HL}). In physics, this system of equations is usually written in terms of transfer matrices with corresponding auxiliary spaces \footnote{"The product in $\text{Rep}(U_q(\tilde{\mathfrak{g}}))$ describes traces of tensor products"}. However, the Kirillov--Reshetikhin modules and the T-system only cover a small part of the prime irreducible modules for rank $n\geq 2$. In type A, they cover exactly the evaluation representations of representations of $U_q(\sln)$ which have a rectangular Young tableau. In particular, the tensor product of fundamental and antifundamental modules above is not included. Moreover, this product contains prime simple composition factors that are not evaluation representations. Thus, in order to understand all the simple objects in $\mathrm{C}$ we should consider more general modules as the so called Snake modules introduced in \cite{MY2} by Mukhin and Young. In fact, it turns out that they cover all the prime simple objects in $\mathrm{C}$.\\
	
	The Snake modules in type A are defined as follows. Define subsets $\mathcal{X}:=\{(i,k)\in I\times\mathbb{Z}:i-k=1\mod 2\}$ and $\mathcal{W}:={(i,k):(i,k-1)\in \mathcal{X}}\subset I\times \mathbb{Z}$. We fix $a\in\mathbb{C}^\times$ and only work with representations whose $q$\textit{-characters} lie in the subring $\mathbb{Z}[Y_{i,aq^k}^{\pm1}]_{(i,k)\in\mathcal{X}}\subset \mathcal{Y}$. Indeed, these form a subcategory $\mathrm{C}_{\mathbb{Z}}$ of $\mathrm{C}$ closed under taking tensor products. Conversely, every simple object $S$ of $\mathrm{C}$ can be written as a tensor product $S_1(a_1)\otimes\cdots\otimes S_k(a_k)$ for some simple objects $S_1,...,S_k \in \mathcal{C}_{\mathbb{Z}}$ and $\frac{a_i}{a_j}\in \mathbb{C}\backslash q^{2\mathbb{Z}}$, here $S(a)$ is the pullback of $S$ by $\tau_a \in \text{Aut}U_q(\hat{\mathfrak{g}})$ (see \cite{HL} 3.6 and 3.7 and \cite{C2002} for the proof). By abuse of notation, we set $Y_{i,aq^k}=: Y_{i,k}$, $A_{i,aq^k}=:A_{i,k}$, $\mathbb{Z}[Y_{i,k}^{\pm1}]_{(i,k)\in\mathcal{X}} = \mathcal{Y}_{\mathbb{Z}}$ and $\mathbb{Z}[A_{i,k}^{\pm1}]_{(i,k)\in\mathcal{W}} = \mathcal{Q}_{\mathbb{Z}}$.
	\begin{definition}[snake position, snakes and snake modules]
		Let $(i,k)\in\mathcal{X}$.
		\begin{enumerate}
			\item A point $(i',k')\in\mathcal{X}$ is said to be in \textbf{snake position} with respect to $(i,k)$ iff $k'-k \geq |i'-i|+2$.
			\begin{enumerate}
				\item The point $(i',k')$ is in \textbf{minimal snake position} to $(i,k)$ iff $k'-k$ is equal to the lower bound.
				\item We say that $(i',k')\in\mathcal{X}$ is in \textbf{prime snake position} with respect to $(i,k)$ iff\\ $\min\{i'+i,2n+2-i-i'\}\geq k'-k \geq |i'-i|+2$.
			\end{enumerate}
			\item A finite sequence $(i_t,k_t)$ ($1\leq t\leq M \in \mathbb{N}$) of points in $\mathcal{X}$ is a \textbf{snake} iff for all $2\leq t\leq M$, $(i_t,k_t)$ is in \textit{snake position} with respect to $(i_{t-1},k_{t-1})$.
			\begin{enumerate}
				\item It is a \textbf{minimal} (resp. \textbf{prime}) \textbf{snake} iff any two successive points are in minimal (resp. prime) snake position to each other.
			\end{enumerate}
			\item The simple module $L(m)$ is called a (\textbf{minimal/prime}) \textbf{snake module} iff $m = \prod_{t=1}^{M}Y_{i_t,k_t}$ for some (minimal/prime) snake $(i_t,k_t)_{1\leq t\leq M}$. $\odot$
		\end{enumerate}
	\end{definition}
	One can now proof the following properties (cf. \cite{MY2},\cite{MY},\cite{CH},\cite{BJY}).\newpage
	\begin{thm}[snake modules]\hspace{1em}
		\label{thm:snake_modules}
		\begin{enumerate}
			\item Snake modules are special, anti-special and thin.
			\item A snake module is prime iff its snake is prime.
			\item Prime snake modules are real.
			\item If a snake module is not prime then it is isomorphic to a tensor product of prime snake modules defined uniquely up to permutation. $\odot$
		\end{enumerate}
	\end{thm}
	Thus, prime snake modules are prime, special, anti-special, thin, and real and any snake module decomposes into a tensor product of prime snakes. We claim that these are exactly the building blocks that we were looking for to solve our problem of finding the irreducible composition factors in our tensor product of fundamental and antifundamental representations above. Indeed, we can use the main result of Mukhin and Young \cite{MY}, which is the existence of a short exact sequence called the extended T-system. To write it down in a nice way, it is helpful to add the following definition.
	\begin{definition}[neighbouring points and neighbouring snakes]\hspace{1em}
		\begin{enumerate}
			\item For any two successive points $(i,k)$ and $(i',k')$ define the \textbf{neighbouring points} by
			\begin{align*}
				&\mathbb{X}_{i,k}^{i',k'}:=
				\begin{cases}
					((\frac{1}{2}(i+k+i'-k'),\frac{1}{2}(i+k-i'+k')))  & {\scriptstyle k+i>k'-i'}\\
					\emptyset & {\scriptstyle k+i=k'-i'}
				\end{cases}\\
				&\mathbb{Y}_{i,k}^{i',k'}:=
				\begin{cases}
					((\frac{1}{2}(i'+k'+i-k),\frac{1}{2}(i'+k'-i+k)))  & {\scriptstyle k+N+1-i>k'-N-1+i'}\\
					\emptyset & {\scriptstyle k+N+1-i=k'-N-1+i'.}
				\end{cases}
			\end{align*}
			\item For any prime snake $(i_t,k_t)_{1\leq k\leq M}$ we define its \textbf{neighbouring snakes} $\mathbb{X}:=\mathbb{X}_{(i_t,k_t)_{1\leq k\leq M}}$ and $\mathbb{Y}:=\mathbb{Y}_{(i_t,k_t)_{1\leq k\leq M}}$ by concatenating its neighbouring points. $\odot$
		\end{enumerate}
	\end{definition}
	Now, the \textit{extended T-system} is stated as follows (cf. \cite{MY} theorem 4.1).
	\begin{thm}[the extended T-system]
		\label{thm:the_extended_T-system}
		Let $(i_t,k_t)\in\mathcal{X}$, $1\leq t\leq M$, be a prime snake of length $M\geq 2$. Let $\mathbb{X}$ and $\mathbb{Y}$ be its neighbouring snakes. Then we have the following relation in the Grothendieck ring $\operatorname{Rep}(U_q(\tilde{\mathfrak{g}}))$.
		\begin{align}
			\notag \left[L\left(\prod_{t=1}^{M-1}Y_{i_t,k_t}\right)\right] \left[L\left(\prod_{t=2}^{M}Y_{i_t,k_t}\right)\right] = \left[L\left(\prod_{t=2}^{M-1}Y_{i_t,k_t}\right)\right] \left[L\left(\prod_{t=1}^{M}Y_{i_t,k_t}\right)\right]\\ +\left[L\left(\prod_{(i,k)\in\mathbb{X}}Y_{i_t,k_t}\right)\right] \left[L\left(\prod_{(i,k)\in\mathbb{Y}}Y_{i_t,k_t}\right)\right],
		\end{align}
		where the summands on the right hand side are classes of irreducible modules, i.e.
		\begin{align*}
			&L\left(\prod_{t=2}^{M-1}Y_{i_t,k_t}\prod_{t=1}^{M}Y_{i_t,k_t}\right)\cong L\left(\prod_{t=2}^{M-1}Y_{i_t,k_t}\right)\otimes L\left(\prod_{t=1}^{M}Y_{i_t,k_t}\right)\\
			&L\left(\prod_{(i,k)\in\mathbb{X}}Y_{i_t,k_t}\prod_{(i,k)\in\mathbb{Y}}Y_{i_t,k_t}\right)\cong L\left(\prod_{(i,k)\in\mathbb{X}}Y_{i_t,k_t}\right)\otimes L\left(\prod_{(i,k)\in\mathbb{Y}}Y_{i_t,k_t}\right).\quad\odot
		\end{align*}
	\end{thm}
	Note that in the case of Kirillov--Reshetikhin modules, the theorem is the standard \textit{T-system}.
	Let us draw an example of the \textit{extended T-system} for $A_4$.
	\begin{figure}[H]
		\centering
		\includegraphics[scale = .9, trim = 1cm 14.5cm 0cm 6.5cm]{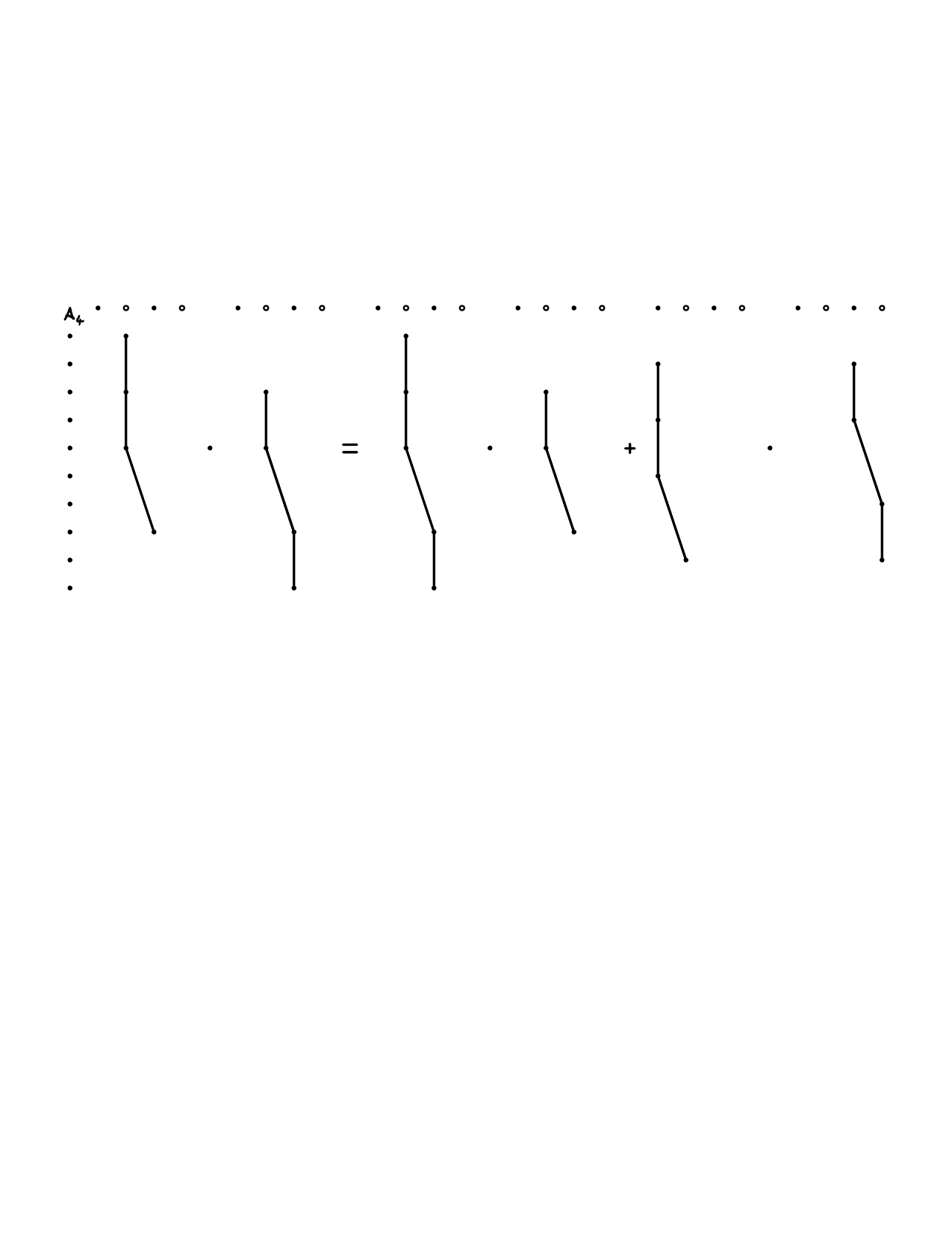}
		\caption{An example of the extended T-system for $A_4$.}
	\end{figure}
	We are now in the position to analyse the tensor product of fundamental and antifundamental lines in the Snail Operator. Using the correspondence in \cite{GT} between finite dimensional representations of Yangians and quantum affine algebras we note that everything can be defined in exactly the same way. We take $\mu\in\mathbb{C}$ fixed and set $Y_{i,\mu+\frac{k}{2}}=:Y_{i,k}$ as well as $A_{i,\mu+\frac{k}{2}}=:A_{i,k}$. We note that the loop variables of the successive lines in the snail operator are in minimal snake position. Using the extended T-system, we can now prove the following assertions.
	\begin{thm}[Snake and Snail]
		\label{thm:snake_and_snail}
		The tensor product
		\begin{align}
			W_{N(m)}^{(1)}(\mu+\frac{k}{2})\otimes W_{N(m+1)}^{(1)}(\mu+\frac{k+n+1}{2})\otimes\cdots\otimes W_{N(m+l)}^{(1)}(\mu+\frac{k+l(n+1)}{2}), \, k\in\mathbb{Z}
		\end{align}
		of $l+1$ many antifundamental and fundamental representations has Fibonacci$(l+1)$ many composition factors, one of which is the minimal snake module 
		\begin{align}
			S^{(l+1)}_m(\mu+\frac{k}{2}):=L(\prod_{t=0}^{l}Y_{N(t+m),k+t(n+1)}), \, k\in\mathbb{Z},
		\end{align}
		where $N(t):=\left\{^{1,\, t\text{ even}}_{n,\, t\text{ odd}}\right.$. In particular, we can prove the existence of short exact sequences of the form
		\begin{align}
			[S^{(1)}_{m+1}(\mu-l\frac{n+1}{2})][S^{(l)}_m(\mu-(l-1)\frac{n+1}{2})] = [S^{(l+1)}_{m+1}(\mu-l\frac{n+1}{2})]+[S^{(l-1)}_{m+1}(\mu-(l-2)\frac{n+1}{2})]
			\label{eqn:new_ext_t_sys}
		\end{align}
		which we may also call extended T-systems. $\odot$
	\end{thm}
	The proof of the theorem is given in the Appendix \ref{App:A}.
	We state the following conjecture.
	\begin{conj}[Snake in the Snail]
		The extended T-systems (\ref{eqn:new_ext_t_sys}) ($m$ odd) appear in the successive lines of the snail operator and the component corresponding to $[S^{(l-1)}_{m+1}(\mu-(l+2)\frac{n+1}{2})]$ cancels out, i.e. the snail operator $\tilde{X}_k$ can be defined through a single irreducible representation of the Yangian just like in the $\mathfrak{sl}_2$ case. The minimal snake module $S^{(k)}_{1}(\mu-(k-1)\frac{n+1}{2})$. $\odot$
	\end{conj}
	We checked the first few steps of this conjecture in Mathematica. So $n = 2$ and up to
	$k = 4$ loops as explained above. Moreover, we can also set $n = 3, 4,\dots$ in the program,
	but the calculational effort grows rapidly. However, a general proof might be done by
	induction and considering a corresponding partition of unity. We haven’t analysed the
	projectors yet, but a proof should be possible in the future. Let us emphasize again, that
	the Snail Operator $X_k$ only applies to our problem when the number of loops $k$ is odd. This is discussed in section \ref{sect:sl3snailconstr}. We should also say that it is possible to calculate the $q$-character of the minimal snake module $S^{(k)}_m$ using the path description and Theorem 6.5 in \cite{MY}. Moreover, we can analyse $S^{(k)}_m$ as a representation of $\sln$ and calculate the Young diagrams of all its irreducible subrepresentations. At last, we remark that snake modules can be seen as part of a cluster algebra. Thus, the extended T-systems can be understood as explicit cluster relations. We have seen that it is helpful to understand these kind of relations. We have seen that it can be helpful to understand these kind of relations.
	
	\section*{Acknowledgement}
	
	We thank our friends and colleagues Frank Göhmann, Andreas Klümper and Alexander Razumov for helpful discussions and advices during the time of research and the writing of this paper.
	
	Henrik Jürgens especially thanks Matthias Wendt for the ideas and the almost weekly discussions over more than a year regarding in particular the representation-theoretical background that was used in this paper.
	
	Henrik Jürgens thanks his friend Alec Cooper for helpful discussions, ideas and useful suggestions during the course of our research.
	
	All authors were supported by DFG through the research unit FOR 2316. 
	Henrik Jürgens acknowledges financial support from Deutsche 
	Forschungsgemeinschaft through DFG project BO 340/5-1.
	
	Henrik Jürgens acknowledges the
	support of the grant FAR UNIMORE project CUP-E93C23002040005.

	\newpage
	\appendix
	\section{}
	\label{App:A}
	Let the action of the density matrices $D$ and $D^{(1)}$ be defined as 
	\begin{align*}
		D_{1,\dots,m}(\lambda_1,\dots,\lambda_n)(X_{1,\dots,m}) &:= \operatorname{tr}_{1,\dots,m}\left(D_{1,\dots,m}(\lambda_1,\dots,\lambda_m)X_{1,\dots,m}\right),\\
		D^{(1)}_{\bar{1},2,\dots,m}(\lambda_1,\dots,\lambda_m)(X_{\bar{1},2,\dots,m}) &:= \operatorname{tr}_{\bar{1},2,\dots,m}\left(D^{(1)}_{\bar{1},2,\dots,m}(\lambda_1,\dots,\lambda_m)X_{\bar{1},2,\dots,m}\right),
	\end{align*}
	where on the right hand side $D$ is understood as an element of $\End(V^{\otimes m})$ via the transposition isomorphism ($\End(V^{\otimes m})^\star\cong \End(V^{\otimes m})$).
	Then, one can write the two reduced qKZ-equations as
	\begin{align}
		&D^{(1)}_{\bar{1},2,\dots,m}(\lambda_1-\frac{n+1}{2},\lambda_2,\dots,\lambda_m) = A^{(1)}_{1,\bar{1}|2,\dots,m}(\lambda_1|\lambda_2,\dots,\lambda_m)\left(D_{1,\dots,m}(\lambda_1,\lambda_2,\dots,\lambda_m)\right):=\notag\\ &\operatorname{tr}_1\left(R_{1m}(\lambda_1-\lambda_m)\cdots R_{12}(\lambda_1-\lambda_2)D_{1,\dots,m}(\lambda_1,\lambda_2,\dots,\lambda_m)(n+1)P^{-}_{1\bar{1}}R_{21}(\lambda_2-\lambda_1)\cdots R_{m1}(\lambda_m-\lambda_1)\right),\\
		&D_{1,\dots,m}(\lambda_1-\frac{n+1}{2},\lambda_2,\dots,\lambda_m) = A^{(2)}_{\bar{1},1|2,\dots,m}(\lambda_1|\lambda_2,\dots,\lambda_m)\left(D^{(1)}_{\bar{1},2,\dots,m}(\lambda_1,\lambda_2,\dots,\lambda_m)\right):=\notag\\ &\operatorname{tr}_{\bar{1}}\left(\bar{\bar{R}}_{\bar{1}m}(\lambda_1-\lambda_m)\cdots \bar{\bar{R}}_{\bar{1}2}(\lambda_1-\lambda_2)D^{(1)}_{\bar{1},2,\dots,m}(\lambda_1,\lambda_2,\dots,\lambda_m)(n+1)P^{-}_{1\bar{1}}\bar{R}_{2\bar{1}}(\lambda_2-\lambda_1)\cdots \bar{R}_{m\bar{1}}(\lambda_m-\lambda_1)\right),\label{eqn:second_rqKZ}
	\end{align}
	where $(P^{-}_{1\bar{1}})^2=P^{-}_{1\bar{1}}$ is the projector onto the singlet in the tensor product $V\otimes\overline{V}$ of the fundamental and antifundamental representation of $\mathfrak{sl}_n$.
	\begin{proof}[Proof of the pole structure of $D_m$]
		We use the rqKZ equation to calculate $D_{1,\dots,m}(\lambda_1-k(n+1)-l,\lambda_2,\dots,\lambda_m)$, $k\in\mathbb{N}$, $l=0,1,\dots,n$, in terms of a product of $R$-matrices and $D_{1,\dots,m}(\lambda_1-l,\lambda_2,\dots,\lambda_m)$. We assume that the parameters $\lambda_3,\dots,\lambda_m$ are in general position (in the sense of property (\ref{analyt_sln+1}) conjecture \ref{conj:analyt_prop_sln+1}) and therefore just write out the $R$-matrices that depend on the difference $\lambda_1-\lambda_2$.
		\begin{align*}
			&\underset{\lambda_{1} = \lambda_2}{\text{res}}D_{1,\dots,m}(\lambda_1-k(n+1)-l,\lambda_2,\dots,\lambda_m)=\\
			&\underset{\lambda_{1} = \lambda_2}{\text{res}}\tr_{\bar{1}}[\dots\bar{\bar{R}}_{\bar{1}2}(\lambda_1-\lambda_2-k(n+1)+\frac{n+1}{2}-l)\tr_1[\dots R_{12}(\lambda_1-\lambda_2-(k-1)(n+1)-l)\dots\\ &\dots\tr_{\bar{1}}[\dots\bar{\bar{R}}_{\bar{1}2}(\lambda_1-\lambda_2-\frac{n+1}{2}-l)\tr_1[\dots R_{12}(\lambda_1-\lambda_2-l)D_{1,\dots,m}(\lambda_1-l,\lambda_2,\dots,\lambda_m)\\
			&R_{21}(\lambda_2-\lambda_1+l)\dots]\bar{R}_{2\bar{1}}(\lambda_2-\lambda_1+\frac{n+1}{2}+l)\dots]\dots\\&\dots R_{21}(\lambda_2-\lambda_1+(k-1)(n+1)+l)\dots]\bar{R}_{2\bar{1}}(\lambda_2-\lambda_1+k(n+1)-\frac{n+1}{2}+l)\dots]
		\end{align*}
		Since $D_{1,\dots,m}(\lambda_1-l,\lambda_2,\dots,\lambda_m)$, $l=0,1,\dots,n$, has no poles at $\lambda_1=\lambda_2$ due to property (\ref{analyt_sln+1}) conjecture \ref{conj:analyt_prop_sln+1}, the pole at $\lambda_1=\lambda_2$ is completely determined from the prefactors of the $R$-matrices. Writing them out, we obtain
		\begin{align*}
			\prod_{j=1}^{k}&\frac{1}{\lambda_{12}-j(n+1)-l}\cdot\frac{1}{\lambda_{12}-(j-1)(n+1)-l}\cdot\frac{1}{\lambda_{12}-(j-1)(n+1)-l+1}\cdot\\&\cdot\frac{1}{\lambda_{12}-(j-1)(n+1)-l-1},
		\end{align*}
		which has a simple pole at $\lambda_{12}(=\lambda_1-\lambda_2)=0$ iff $l=0$ or $l=1$. Thus, the density matrix $D_{1,\dots,m}(\lambda_1,\lambda_2,\dots,\lambda_m)$ has no poles at $\lambda_1=\lambda_2-k(n+1)-l$ for $l=2,\dots,n$ and at most simple poles for $l=0,1$.
	\end{proof}
	\begin{proof}[Proof of theorem \ref{thm:snake_and_snail}]
		To proof the first part of the theorem, we use the $q$-characters of the fundamental modules $W_1^{(1)}(a) = L(Y_{1,0})$ and $W_n^{(1)}(a) = L(Y_{n,0})$. They are very easy to obtain by using the path formula (6.3) in theorem 6.5 of \cite{MY}. Namely,
		\begin{align}
			&\chi_q(W_1^{(1)}(a))=Y_{1,0}+Y^{-1}_{n,n+1}+\sum_{k=1}^{n-1}Y_{k+1,k}Y^{-1}_{k,k+1},\\
			&\chi_q(W_n^{(1)}(a))=Y^{-1}_{1,n+1}+Y_{n,0}+\sum_{k=1}^{n-1}Y^{-1}_{k+1,n+1-k}Y_{k,n+1-(k+1)}.
		\end{align}
		Using the multiplicativity of the $q$-character theorem \ref{thm:properties_of_chi_q}, we can easily calculate the $q$-character of the tensor product
		\begin{align}
			W_{N(m)}^{(1)}(\mu+\frac{k}{2})\otimes W_{N(m+1)}^{(1)}(\mu+\frac{k+n+1}{2})\otimes\cdots\otimes W_{N(m+l)}^{(1)}(\mu+\frac{k+l(n+1)}{2}).
			\label{eqn:tensor_prod_proof}
		\end{align}
		Assume $m$ is odd, then the product is of the form
		\begin{align*}
			&(Y^{-1}_{1,k+n+1}+Y_{n,k}+\dots)(Y_{1,k+n+1}+Y^{-1}_{n,k+2(n+1)}+\dots)(Y^{-1}_{1,k+3(n+1)}+Y_{n,k+2(n+1)}+\dots)\\
			&(Y_{1,k+3(n+1)}+Y^{-1}_{n,k+4(n+1)}+\dots)\cdots,
		\end{align*}
		where we omitted all the non dominant monomials in every bracket, that can't multiply to dominant monomials due the shift in the second index (i.e. the spectral parameter). We want to count the number of dominant monomials. Fortunately, this is an easy combinatoric task. Starting from the left we choose one of the two l-weights in the first bracket. If we choose the anti-dominant one we have to choose the dominant one in the next factor if we want to multiply to a dominant monomial, otherwise we still have the free choice for the next factor as soon as it is not the last. We pick one of the two options and look at the next factor if possible. Obviously, the situation is the exact same. Thus, the problem of finding the number of dominant monomials in the $q$-character is equivalent to the number of options to partition $l+1$ objects into boxes of size $m\leq 2$. The solution to this problem is
		\begin{align}
			\sum_{k=0}^{\lfloor \frac{l}{2}\rfloor}\binom{l-k}{k},
		\end{align}
		which is equal to the $(l+1)$th Fibonacci number Fibonacci$(l+1)$. The argument for $m$ even works the exact same way. Thus, the number of irreducible composition factors is smaller or equal to Fibonacci$(l+1)$. Moreover, all the dominant monomials obtained in this way are the dominant monomials of certain snake modules. Due to theorem \ref{thm:snake_modules} we know that these are isomorphic to the tensor product of minimal snakes. Now, let's assume that we have an irreducible composition Factor of the tensor product (\ref{eqn:tensor_prod_proof}). Then, it has to be of the form $L(m)$ where $m$ is one of the dominant monomials in the $q$-character of (\ref{eqn:tensor_prod_proof}). As we have seen, $L(m)$ is a snake module. Thus, since snake modules are special, it contains only one dominant monomial. Subtracting the $q$-character of $L(m)$ from the $q$-character in the tensor product (\ref{eqn:tensor_prod_proof}), the $q$-character of the other composition factor contains all the $\text{Fibonacci}(l+1)-1$ remaining dominant monomials. Therefore, by induction, we conclude that there must be exactly Fibonacci$(l+1)$ many composition factors.
		\\
		
		For the second part of the theorem we have to prove the existence of a short exact sequence of the form
		\begin{align}
			[S^{(1)}_{m+1}(\mu)][S^{(l)}_m(\mu+\frac{n+1}{2})] = [S^{(l+1)}_{m+1}(\mu)]+[S^{(l-1)}_{m+1}(\mu+2\frac{n+1}{2})],
		\end{align}
		where we omitted the term $l\frac{n+1}{2}$ since it can be absorbed into the spectral parameter. We prove the existence by induction over $l$. The case $l=1$ is clear as it is one of the simplest cases ($M=2$) of the \textit{extended T-system} (theorem \ref{thm:the_extended_T-system}). Thus, we only need to prove the step $l-1\to  l$. By the \textit{extended T-system} (theorem \ref{thm:the_extended_T-system}), we have
		\begin{align}
			[S^{(l)}_m(\mu+\frac{n+1}{2})][S^{(l)}_{m+1}(\mu)] = [S^{(l+1)}_{m+1}(\mu)][S^{(l-1)}_{m}(\mu+\frac{n+1}{2})] + 1.
			\label{eqn:proof_ex_seq}
		\end{align}
		We multiply the induction hypothesis by $[S^{(l)}_{m}(\mu+\frac{n+1}{2})]$ and obtain
		\begin{align}
			\notag&[S^{(1)}_{m+1}(\mu)][S^{(l-1)}_m(\mu+\frac{n+1}{2})][S^{(l)}_{m}(\mu+\frac{n+1}{2})]=\\
			\notag&[S^{(l)}_{m+1}(\mu)][S^{(l)}_{m}(\mu+\frac{n+1}{2})]+[S^{(l-2)}_{m+1}(\mu+2\frac{n+1}{2})][S^{(l)}_{m}(\mu+\frac{n+1}{2})]\stackrel{(\ref{eqn:proof_ex_seq})}{=}\\
			\notag&[S^{(l+1)}_{m+1}(\mu)][S^{(l-1)}_{m}(\mu+\frac{n+1}{2})] + 1+[S^{(l-1)}_{m+1}(\mu+2\frac{n+1}{2})][S^{(l-1)}_{m}(\mu+\frac{n+1}{2})] - 1=\\
			&([S^{(l+1)}_{m+1}(\mu)]+[S^{(l-1)}_{m+1}(\mu+2\frac{n+1}{2})])[S^{(l-1)}_{m}(\mu+\frac{n+1}{2})],
			\label{eqn:second_part_up_to_factor}
		\end{align}
		where we applied equation (\ref{eqn:proof_ex_seq}) to both terms on the left hand side.
		This is the desired equation multiplied by $[S^{(l-1)}_{m}(\mu+\frac{n+1}{2})]$. As the Grothendieck ring is an integral domain, (\ref{eqn:second_part_up_to_factor}) proves the second part of the theorem.
	\end{proof}

\end{document}